\documentclass[conference]{IEEEtran}
\IEEEoverridecommandlockouts

\usepackage[T1]{fontenc}
\usepackage[utf8]{inputenc}
\usepackage{graphicx}
\usepackage[colorlinks]{hyperref}
\usepackage[nospace,noadjust]{cite}
\usepackage{amsmath,amsfonts,amsthm,amssymb}
\usepackage[caption=false,font=footnotesize,hangindent=10pt]{subfig}
\usepackage[linesnumbered,vlined,boxed,ruled]{algorithm2e}
\usepackage{color,xparse}
\usepackage{multirow}
\usepackage{pifont}

\def\BibTeX{{\rm B\kern-.05em{\sc i\kern-.025em b}\kern-.08em
T\kern-.1667em\lower.7ex\hbox{E}\kern-.125emX}}

\usepackage{tikz}
\usetikzlibrary{arrows,arrows.meta,calc,shapes.geometric,shapes.misc,positioning,backgrounds,fit}

\graphicspath{{./}{figure/}}
\def\Snospace~{\S{}}

\SetVlineSkip{2pt}

\SetCommentSty{myAlgComment}
\SetKwProg{Fn}{Function}{:}{}

\SetAlFnt{\small}
\SetAlCapFnt{\small}
\SetAlCapNameFnt{\small}

\newcommand{\header}[1]{\smallskip\noindent\textbf{#1}}
\newcommand{\bullethdr}[1]{\noindent\textbullet~\textbf{#1}}
\newcommand{\remarks}{\noindent\textbf{Remarks}}
\newcommand{\remark}{\noindent\textbf{Remark}}

\NewDocumentCommand{\E}{o m}{\mathbb{E}\IfValueTF{#1}{_{#1}}{}\left\{#2\right\}}
\newcommand{\argmax}{\mathop{\arg\max}}

\newcommand{\CA}{\mathcal{A}}
\newcommand{\CS}{\mathcal{S}}
\newcommand{\OPT}{\mathrm{OPT}}
\newcommand{\bx}{\mathbf{x}}

\newtheorem{definition}{Definition}
\newtheorem{problem}{Problem}
\newtheorem{theorem}{Theorem}
\newtheorem{lemma}{Lemma}
\newtheorem{example}{Example}

\SetKwFunction{ProcessEdges}{ProcessEdges}
\SetKwFunction{ReduceRedundancy}{ReduceRedundancy}
\SetKwProg{myproc}{Procedure}{}{}

\title{Tracking Influential Nodes in Time-Decaying Dynamic Interaction Networks}

\author{%
\IEEEauthorblockN{%
Junzhou Zhao\textsuperscript{1}\quad
Shuo Shang\textsuperscript{2}\quad
Pinghui Wang\textsuperscript{3}\quad
John C.S. Lui\textsuperscript{4}\quad
Xiangliang Zhang\textsuperscript{1}%
}
\medskip
\IEEEauthorblockA{%
\textsuperscript{1}\textit{King Abdullah University of Science and
Technology, Saudi Arabia} \\
{\fontsize{9}{9}\ttfamily\upshape\{junzhou.zhao, xiangliang.zhang\}@kaust.edu.sa} \\
\textsuperscript{2}\textit{Inception Institute of Artificial Intelligence, UAE} \\
{\fontsize{9}{9}\ttfamily\upshape jedi.shang@gmail.com} \\
\textsuperscript{3}\textit{Xi'an Jiaotong University, China}\\
{\fontsize{9}{9}\ttfamily\upshape phwang@mail.xjtu.edu.cn}\\
\textsuperscript{4}\textit{The Chinese University of Hong Kong, Hong Kong}\\
{\fontsize{9}{9}\ttfamily\upshape cslui@cse.cuhk.edu.hk}}
\thanks{*Shuo Shang and Xiangliang Zhang are the corresponding authors.}}

\begin{document}
\maketitle

\begin{abstract}
Identifying influential nodes that can jointly trigger the maximum influence
spread in networks is a fundamental problem in many applications such as viral
marketing, online advertising, and disease control.
Most existing studies assume that social influence is static and they fail to
capture the dynamics of influence in reality.
In this work, we address the dynamic influence challenge by designing efficient
streaming methods that can identify influential nodes from highly dynamic node
interaction streams.
We first propose a general {\em time-decaying dynamic interaction network} (TDN)
model to model node interaction streams with the ability to smoothly discard
outdated data.
Based on the TDN model, we design three algorithms, i.e., {\sc SieveADN}, {\sc
BasicReduction}, and {\sc HistApprox}.
{\sc SieveADN} identifies influential nodes from a special kind of TDNs with
efficiency.
{\sc BasicReduction} uses {\sc SieveADN} as a basic building block to identify
influential nodes from general TDNs.
{\sc HistApprox} significantly improves the efficiency of {\sc BasicReduction}.
More importantly, we theoretically show that all three algorithms enjoy constant
factor approximation guarantees.
Experiments conducted on various real interaction datasets demonstrate that our
approach finds near-optimal solutions with speed at least $5$ to $15$ times faster
than baseline methods.

\end{abstract}

\section{\textbf{Introduction}}
\label{sec:introduction}

Online social networks allow their users to connect and interact with each other
such as one user re-tweets/re-shares another user's tweets/posts on
Twitter/Facebook.
Interactions between connected users can cause members in the network to be
influenced.
For example, a video goes viral on Twitter after being re-tweeted many times, a
rumor spreads like a wildfire on Facebook via re-sharing, etc.
In these scenarios, users in the network are influenced (i.e., they watched the
video or got the rumor) via a cascade of user interactions.
Understanding and leveraging social influence have been hot in both academia and
business.
For example, in academia, identifying $k$ users who can jointly trigger the
maximum influence spread in a network is known as the {\em influence maximization}
(IM) problem~\cite{Kempe2003}; in business, leveraging social influence to boost
product sales is known as viral marketing.

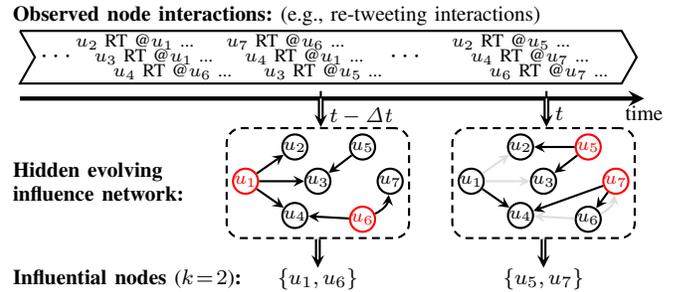
\begin{figure}[t]
\centering
\begin{tikzpicture}[
every node/.style={font=\footnotesize,inner sep=0pt},
nd/.style={circle,thick,draw,minimum size=9pt,font=\scriptsize},
rt/.style={anchor=north west, inner sep=1pt,font=\scriptsize},
arr/.style={->,>=stealth,thick},
flow/.style={->,>=stealth,thick,double},
del/.style={gray!25},
box/.style={rectangle,draw,thick,densely dashed,rounded corners=2mm,inner sep=2pt},
hd/.style={anchor=west,align=left,font=\bfseries\footnotesize},
tic/.style={anchor=north west,inner sep=3pt},
]

\coordinate (bc) at (-1,0);

\node[above left=.2 and .4 of bc,hd] {Observed node interactions: {\normalfont (e.g., re-tweeting interactions)}};

\coordinate[left=.3 of bc] (st_nw);
\coordinate[below=.7 of st_nw] (st_sw);
\coordinate[right=8 of st_nw] (st_ne);
\coordinate[right=8 of st_sw] (st_se);
\coordinate[below right=.35 and .2 of st_nw] (st_w);
\coordinate[below right=.35 and .2 of st_ne] (st_e);
\draw[thick] (st_nw) -- (st_ne) -- (st_e) -- (st_se) -- (st_sw) -- (st_w) -- (st_nw);

\coordinate[right=.4 of bc] (rb1);
\node[rt] at (rb1) {$u_2$ RT @$u_1$ ...};
\node[rt,below right=.2 and .25 of rb1] (tmp) {$u_3$ RT @$u_1$ ...};
\node[rt,below right=.4 and .5 of rb1] {$u_4$ RT @$u_6$ ...};

\node[left=.2 of tmp,font=\normalsize] {$\cdots$};

\coordinate[right=2 of rb1] (rb2);
\node[rt] at (rb2) {$u_7$ RT @$u_6$ ...};
\node[rt,below right=.2 and .25 of rb2] (tmp) {$u_4$ RT @$u_1$ ...};
\node[rt,below right=.4 and .5 of rb2] {$u_3$ RT @$u_5$ ...};

\node[right=.3 of tmp,font=\normalsize] {$\cdots$};

\coordinate[right=3 of rb2] (rb3);
\node[rt] at (rb3) {$u_2$ RT @$u_5$ ...};
\node[rt,below right=.2 and .25 of rb3] {$u_4$ RT @$u_7$ ...};
\node[rt,below right=.4 and .5 of rb3] {$u_6$ RT @$u_7$ ...};

\coordinate[below left=.9 and .3 of bc] (rb4);
\draw[arr,ultra thick] (rb4) -- ++(8.3,0) node[below=2pt,anchor=north] {time};

\coordinate[right=4 of rb4] (tc1);
\draw[thick] (tc1) node[tic] (nt) {$t-\mathit{\Delta t}$} -- ++(0,.1);

\coordinate[right=3 of tc1] (tc2);
\draw[thick] (tc2) node[tic] (ntd) {$t$} -- ++(0,.1);

\node[below left=2 and .4 of bc,hd] {Hidden evolving\\influence network:};
\node[below left=3.3 and .4 of bc,hd] (h3) {Influential nodes {\normalfont($k\!=\!2$)}:};

\begin{scope}[xshift=1.7cm,yshift=-2cm]
\node[nd,red] (u1) at (0,0) {$u_1$};
\node[nd,above right=.2 and .4 of u1] (u2) {$u_2$};
\node[nd,right=.6 of u1] (u3) {$u_3$};
\node[nd,below right=.2 and .4 of u1] (u4) {$u_4$};
\node[nd,above right=.2 and 1.3 of u1] (u5) {$u_5$};
\node[nd,red,below right=.25 and 1.3 of u1] (u6) {$u_6$};
\node[nd,right=.6 of u3] (u7) {$u_7$};

\draw[arr] (u1) -- (u2);
\draw[arr] (u1) -- (u3);
\draw[arr] (u1) -- (u4);
\draw[arr] (u5) -- (u3);
\draw[arr] (u6) -- (u4);
\draw[arr] (u6) edge[bend right] (u7);

\node[box,fit={(u1) (u2) (u6) (u7)}] (gt1) {};
\draw[flow] (tc1) -- (gt1.north -| tc1);
\coordinate[right=.3 of u7] (arrow_left);

\node at (h3 -| gt1) (nds1) {$\{u_1,u_6\}$};
\draw[flow] (gt1) -- (nds1);
\end{scope}

\begin{scope}[xshift=4.7cm,yshift=-2cm]
\node[nd] (u1) at (0,0) {$u_1$};
\node[nd,above right=.2 and .4 of u1] (u2) {$u_2$};
\node[nd,right=.6 of u1] (u3) {$u_3$};
\node[nd,below right=.2 and .4 of u1] (u4) {$u_4$};
\node[nd,red,above right=.2 and 1.3 of u1] (u5) {$u_5$};
\node[nd,below right=.25 and 1.3 of u1] (u6) {$u_6$};
\node[nd,red,right=.6 of u3] (u7) {$u_7$};

\draw[arr,del] (u1) -- (u2);
\draw[arr,del] (u1) -- (u3);
\draw[arr] (u1) -- (u4);
\draw[arr] (u5) -- (u2);
\draw[arr] (u5) -- (u3);
\draw[arr,del] (u6) -- (u4);
\draw[arr,del] (u6) edge[bend right] (u7);
\draw[arr] (u7) -- (u4);
\draw[arr] (u7) -- (u6);

\node[box,fit={(u1) (u2) (u6) (u7)}] (gt2) {};
\draw[flow] (tc2) -- (gt2.north -| tc2);
\coordinate[left=.3 of u1] (arrow_right);

\node at (h3 -| gt2) (nds2) {$\{u_5,u_7\}$};
\draw[flow] (gt2) -- (nds2);
\end{scope}

\end{tikzpicture}

\caption{Observed node interactions (``$u$ RT @$v$ ...''
denotes that user $u$ re-tweeted user $v$'s tweet) and hidden evolving
influence network (a directed edge $(u,v)$ denotes that $u$ influenced $v$).
Influential nodes also evolve.}
\label{fig:dyn}
\end{figure}

User interactions are first-hand evidences reflecting one user's influence on
another, e.g., user $a$ re-tweeting user $b$'s tweet implies that $b$ influenced
$a$.
Most studies on social
influence~\cite{Kempe2003,Borgs2014,Chen2010,Tang2014,Lucier2015,Tang2015,Chen2016d,Litou2017,Lin2017}
need to estimate {\em influence probabilities} based on observed user
interactions~\cite{Saito2008,Goyal2010,Kutzkov2013}.
Then, user influence is evaluated on an {\em influence network}, which is a
directed graph with influence probabilities among nodes.

\header{Dynamic influence challenge.}
One major issue in these studies is that both influence probabilities and
influence network are assumed to be static, i.e., {\em social influence is assumed
to be static}.
However, social influence in real-world could be dynamic driven by the highly
dynamic user interactions.
For instance, user $a$ frequently re-tweeted user $b$'s tweets in the past few
weeks, but stopped re-tweeting recently because $b$ posted offensive content and
$a$ unfollowed $b$, thereby $b$ no longer exerting influence on $a$.
Indeed, \cite{Myers2014} reported that Twitter network is highly dynamic with
about $9\%$ of all connections changing in every month.
Moreover, user interactions can be drastically affected by external out-of-network
sources such as mass media, newspapers, and TV stations~\cite{Myers2012}.
Consequently, it is no longer suitable to assume that social influence is static;
otherwise the identified influential users in IM may just quickly become outdated.
This raises the following problem: {\em in a world with highly dynamic user
interactions, how can we efficiently identify $k$ most influential users at any
time, as illustrated in Fig.~\ref{fig:dyn}?}

A straightforward way to handle the dynamics of influence is that we re-compute
everything from scratch every time we need to query influential users, i.e.,
re-estimate influence probabilities and re-identify influential users on the
updated influence network.
Obviously, this approach incurs too much computational overhead which may be
unaffordable if we need to query influential users frequently.
There have been some recent research efforts trying to address the dynamic
influence challenge, such as the heuristic
approaches~\cite{Aggarwal2012,Zhuang2013}, building updatable
sketches~\cite{Ohsaka2016,Yang2017c}, and the interchange greedy
approach~\cite{Song2017}.
However, these methods either do not have theoretical guarantees on the quality of
selected users (e.g., heuristic approaches~\cite{Aggarwal2012,Zhuang2013}), or
they cannot handle highly dynamic data (e.g., the interchange greedy
approach~\cite{Song2017} actually degrades to the re-computation approach when
influential users change significantly over time).
In addition, these methods assume that influence probabilities are given in
advance; however, estimating influence probabilities itself could be a high
complexity inference task~\cite{Saito2008,Goyal2010,Xiang2010,Kutzkov2013},
especially when influence probabilities are time-varying.

\header{Present work: a streaming optimization approach.}
In this work, we explore the potential of designing efficient streaming methods to
address the dynamic influence challenge.

When user interactions are continuously generated and aggregated in chronological
order, they form a stream.
An appealing approach for identifying influential users is to process this user
interaction stream directly, in a streaming fashion~\cite{Alon1999}.
Specifically, we attempt to maintain some compact {\em intermediate results} while
processing the stream.
We keep updating these intermediate results when new user interactions are
examined.
At any query time, we can quickly obtain a solution using the maintained
intermediate results.
This streaming approach has the potential to allow us to track influential users
over time continuously.
However, to materialize this ideal streaming approach, we need to carefully
address following concerns.

\textbullet\,\textit{Data recency.}
Older user interactions are less significant than more recent ones in evaluating
users' current influence.
For example, the observation that user $b$'s tweet was re-tweeted a year ago, is
less valuable than the observation that $b$'s tweet was re-tweeted yesterday, when
evaluating $b$'s current influence.
The streaming approach is required to have a mechanism that can properly discard
outdated data.

\textbullet\,\textit{Space and update efficiency.}
Space used by storing intermediate results should be compact and upper bounded
with the progression of the stream.
Meanwhile, the update operation should be as efficient as possible so that we can
handle high-speed user interaction streams which are common in practice.

\textbullet\,\textit{Solution quality.}
The output solution should be close to the optimal solution at any query time.

This paper is essentially devoted to address above concerns by proposing a general
streaming model and designing a set of streaming algorithms.

To address the data recency concern, a commonly used streaming model in the
literature is the {\em sliding-window}
model~\cite{Datar2002,Braverman2007,Epasto2017} where only the most recent $W$
elements in the stream remain active and the rest are discarded.
For example, \cite{Wang2017} recently developed a streaming method based on the
sliding-window model to solve IM in a streaming fashion.
However, the sliding-window model has its limitation which can be exposed by the
following example.

\begin{example}\label{eg:sw}
Suppose we want to identify influential users on Twitter based on re-tweeting
interactions, i.e., if a user's tweets were re-tweeted by many other users, the
user is considered to be influential.
Alice is an influential user on Twitter for many years.
Recently, Alice is ill and has been in hospital for weeks.
During this period, Alice cannot use Twitter.
Because Alice disappeared from her followers' timelines, no user re-tweeted her
tweets during this period.
\end{example}

In Example~\ref{eg:sw}, if the sliding-window size is too small that no
re-tweeting interaction related to Alice is observed, then Alice will be
considered not influential, even though she has been influential for many years
and her absence is merely temporal.
This example demonstrates that sliding-window model does not discard outdated data
in a smooth manner and results in unstable solutions.
It thus motivates us to find better models.

\bullethdr{TDN model.}
In this work, we propose a general {\em \textbf{t}ime-decaying \textbf{d}ynamic
interaction \textbf{n}etwork} (TDN) model to enable smoothly discarding outdated
user interactions, rather than the non-smooth manner in sliding-window model.
In TDN, each user interaction is assigned a \emph{lifetime}.
The lifetime is the maximum time span that a user interaction can remain active.
Lifetime automatically decreases as time elapses.
If the lifetime becomes zero, the corresponding user interaction is discarded.
By choosing different lifetime assignments, TDN model can be configured to discard
outdated user interactions in various ways, which include the sliding-window model
as a special case.
In short, TDN is a general streaming model to address the data recency issue.

\bullethdr{Efficient streaming algorithms.}
We address the other concerns by designing three streaming algorithms, i.e., {\sc
SieveADN}, {\sc BasicReduction}, and {\sc HistApprox}, all based on the TDN
model.
{\sc SieveADN} can identify influential nodes over a special kind of TDNs.
{\sc BasicReduction} leverages {\sc SieveADN} as a basic building block to
identify influential nodes over general TDNs.
{\sc HistApprox} significantly improves the efficiency of {\sc BasicReduction}.
Our streaming algorithms are inspired by the streaming submodular optimization
(SSO) techniques~\cite{Badanidiyuru2014a,Epasto2017}.
Current SSO techniques can only handle insertion-only~\cite{Badanidiyuru2014a} and
sliding-window~\cite{Epasto2017} streams.
To the best of our knowledge, the work in this paper is the first to handle more
general time-decaying streams.
More importantly, we theoretically show that our approach can find near-optimal
solutions with both time and space efficiency.

\header{Contributions.}
In summary, our contributions are as follows:
\begin{itemize}
\item We propose a general TDN model to model user interaction streaming data with
the ability to discard outdated user interactions smoothly
(\autoref{sec:preliminaries}).
\item We design three streaming algorithms based on the TDN model, namely {\sc
SieveADN}, {\sc BasicReduction}, and {\sc HistApprox}.
Our algorithms are applicable to time-decaying streams and achieve a constant
$(1/2-\epsilon)$ approximation guarantee (\autoref{sec:basicreduction}
and~\autoref{sec:histapprox}).
\item We conduct extensive experiments on various real interaction datasets.
The results demonstrate that our approach outputs near-optimal solutions with
speed at least $5$ to $15$ times faster than baseline methods.
(\autoref{sec:experiments}).
\end{itemize}

\section{Preliminaries and Problem Formulation}
\label{sec:preliminaries}

Notice that interactions are not necessarily occurred between two users but could
be between any two entities, or nodes in networks.
In this section, we first formally define the general node interaction data.
Then we propose a time-decaying dynamic interaction network model.
Finally, we formulate the problem of tracking influential nodes.

\subsection{Node Interactions}
\label{ss:interactions}

\begin{definition}[Interaction] \label{def:interaction} An interaction between two
nodes in a network is a triplet $\langle u,v,\tau\rangle$ representing that node
$u$ exerts an influence on node $v$ at time $\tau$.
\end{definition}

For example, user $v$ re-tweets/re-shares user $u$'s tweet/post at time $\tau$ on
Twitter/Facebook, user $v$ adopted a product recommended by user $u$ at time
$\tau$, etc.
In these scenarios, $u$ influenced $v$.
An interaction $\langle u,v,\tau\rangle$ is a direct evidence indicating that $u$
influences $v$.
If we observe many such evidences, then we say that $u$ has strong influence on
$v$.

In some scenarios, we may not directly observe the interaction between two nodes,
but if they do have an influence relationship, we are still able to convert these
scenarios to the scenario in Definition~\ref{def:interaction}, e.g., by one-mode
projection.

\begin{example}[One-mode Projection]
User $u$ bought a T-shirt recently.
His friend $v$ also bought a same T-shirt two days later at time $\tau$.
Then, it is very likely that $u$ influenced $v$.
We still denote this interaction by $\langle u,v,\tau\rangle$.
\end{example}

When interactions are continuously generated and aggregated, they form an
interaction stream.

\begin{definition}[Interaction Stream]\label{def:stream}
An interaction stream is an infinite set of interactions generated in discrete
time, denoted by $\CS\triangleq \{\langle u,v,\tau \rangle\colon u,v$ are two
distinct nodes, $\tau=1,2,\ldots\}$.
\end{definition}

For ease of presentation, we use discrete time in this work, and we allow a batch
of node interactions arriving at the same time.
Interaction stream will be the input of our algorithms.

As we discussed previously, older interactions are less significant than more
recent ones in evaluating current influence.
Next, we propose a time-decaying mechanism to satisfy this recency requirement
desired by the streaming approach.

\subsection{Time-Decaying Dynamic Interaction Network Model}
\label{sec:TDN_model}

We propose a simple and general dynamic network model to model an interaction
stream.
The model leverages a {\em time-decaying mechanism} to smoothly discard
outdated interactions.
We refer to our model as the {\em \textbf{t}ime-decaying \textbf{d}ynamic
interaction \textbf{n}etwork} (TDN) model.

Formally, a TDN at time $t$ is simply a directed network denoted by $G_t\triangleq
(V_t,E_t)$, where $V_t$ is a set of nodes and $E_t$ is a set of edges {\em
survived} by time $t$.
Each edge $(u, v, \tau)\in E_t$ is directed and timestamped representing an
interaction.
We assume there is no self-loop edge (i.e., a user cannot influence himself) but
allow multiple edges between two nodes (e.g., a user influences another user
multiple times at different time).

TDN model leverages a {\em time-decaying mechanism} to handle continuous node/edge
additions and deletions while evolving.
The time-decaying mechanism works as follows: an edge is added to the TDN at its
creation time; the edge survives in the network for a while then expires; when the
edge expires, the edge is removed from the network; if edges attached to a node
all expire, the node is removed from the network.

Formally, for an edge $e=(u,v,\tau)$ arrived at time $\tau$, it is assigned a
lifetime $l_\tau(e)\in\{1,\ldots,L\}$ upper bounded by $L$.
The edge's lifetime decreases as time elapses, and at time $t\geq\tau$, its
lifetime decreases to $l_t(e)=l_\tau(e) - t + \tau$.
If $l_{t'}(e)=0$ at some time $t'>\tau$, edge $e$ is removed from the network.
This also implies that $e\in E_t$ iff $\tau\leq t < \tau + l_\tau(e)$.

Note that lifetime plays the same role as a {\em time-decaying weight}.
If an edge has a long lifetime at its creation time, the edge is considered to be
important and will survive in the network for a long time.
An example of such a TDN evolving from time $t$ to $t+1$ is given in
Fig.~\ref{fig:tdn}.

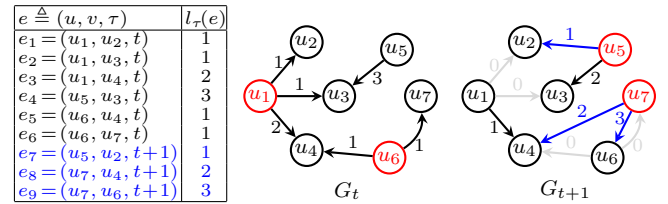
\begin{figure}[htp]
\centering
\begin{tikzpicture}[
every node/.style={font=\footnotesize,inner sep=0pt},
nd/.style={circle,thick,draw,minimum size=12pt},
arr/.style={-stealth,thick},
lb/.style={font=\scriptsize,inner sep=2pt, midway,above},
del/.style={gray!30},
]

\node[inner sep=0,font=\scriptsize] at (0,-.1) {
\renewcommand{\arraystretch}{.8}
\begin{tabular}{|@{\,}l@{\,}|@{\,}c@{}|}
\hline
$e\triangleq (u,v,\tau)$ & $l_\tau\!(e)$ \\
\hline
$e_1\!=\!(u_1,u_2,t)$ & $1$ \\
$e_2\!=\!(u_1,u_3,t)$ & $1$ \\
$e_3\!=\!(u_1,u_4,t)$ & $2$ \\
$e_4\!=\!(u_5,u_3,t)$ & $3$ \\
$e_5\!=\!(u_6,u_4,t)$ & $1$ \\
$e_6\!=\!(u_6,u_7,t)$ & $1$ \\
\textcolor{blue}{$e_7\!=\!(u_5,u_2,t\!+\!1)$} & \textcolor{blue}{$1$} \\
\textcolor{blue}{$e_8\!=\!(u_7,u_4,t\!+\!1)$} & \textcolor{blue}{$2$} \\
\textcolor{blue}{$e_9\!=\!(u_7,u_6,t\!+\!1)$} & \textcolor{blue}{$3$} \\
\hline
\end{tabular}
};

\begin{scope}[xshift=1.9cm]
\node[nd,red] (u1) at (0,0) {$u_1$};
\node[nd,above right=.4 and .3 of u1] (u2) {$u_2$};
\node[nd,right=.6 of u1] (u3) {$u_3$};
\node[nd,below right=.4 and .3 of u1] (u4) {$u_4$};
\node[nd,above right=.3 and 1.5 of u1] (u5) {$u_5$};
\node[nd,red,below right=.5 and 1.4 of u1] (u6) {$u_6$};
\node[nd,above right=.5 and .1 of u6] (u7) {$u_7$};

\draw[arr] (u1) -- (u2) node[lb,pos=.3] {$1$};
\draw[arr] (u1) -- (u3) node[lb] {$1$};
\draw[arr] (u1) -- (u4) node[lb,pos=.2,below]{$2$};
\draw[arr] (u5) -- (u3) node[lb,pos=.2,below] {$3$};
\draw[arr] (u6) -- (u4) node[lb,pos=.4] {$1$};
\draw[arr] (u6) edge[bend right] node[lb,pos=.6,below=1pt] {$1$} (u7);

\node[below right=1.1 and 1 of u1,anchor=center] {$G_t$};
\end{scope}

\begin{scope}[xshift=4.8cm]
\node[nd] (u1) at (0,0) {$u_1$};
\node[nd,above right=.4 and .3 of u1] (u2) {$u_2$};
\node[nd,right=.6 of u1] (u3) {$u_3$};
\node[nd,below right=.4 and .3 of u1] (u4) {$u_4$};
\node[nd,red,above right=.3 and 1.5 of u1] (u5) {$u_5$};
\node[nd,below right=.5 and 1.4 of u1] (u6) {$u_6$};
\node[nd,red,above right=.5 and .1 of u6] (u7) {$u_7$};

\draw[arr,del] (u1) -- (u2) node[lb,pos=.2] {$0$};
\draw[arr,del] (u1) -- (u3) node[lb] {$0$};
\draw[arr] (u1) -- (u4) node[lb,pos=.2,below]{$1$};
\draw[arr,blue] (u5) -- (u2) node[lb,midway] {$1$};
\draw[arr] (u5) -- (u3) node[lb,pos=.2,below] {$2$};
\draw[arr,del] (u6) -- (u4) node[lb,pos=.4] {$0$};
\draw[arr,del] (u6) edge[bend right] node[lb,pos=.6,below=1pt] {$0$} (u7);
\draw[arr,blue] (u7) -- (u4) node[lb] {$2$};
\draw[arr,blue] (u7) -- (u6) node[lb,above=1pt,pos=.7] {$3$};

\node[below right=1.1 and 1 of u1,anchor=center] {$G_{t+1}$};
\end{scope}

\end{tikzpicture}

\caption{A TDN example.
Label on each edge denotes the edge's current lifetime.
Influential nodes ($k=2$) also evolve from time $t$ to $t+1$.}
\label{fig:tdn}
\end{figure}

We find that such a simple time-decaying mechanism makes TDN model highly
configurable by choosing different lifetime assignment methods.
Here, we consider a few special TDNs.

\begin{example}\label{eg:adn}
A network only grows with no node/edge deletions.
It is equivalent to saying that every edge in a TDN has an infinite lifetime,
i.e., $l_\tau(e)=\infty$.
We refer to such networks as \textbf{a}ddition-only \textbf{d}ynamic interaction
\textbf{n}etworks (ADNs).
\end{example}

\begin{example}
A network consists of edges in the most recent $W$ time steps.
It is equivalent to saying that every edge in a TDN has same lifetime $W$, i.e.,
$l_\tau(e)=W$.
We refer to such networks as sliding-window dynamic interaction networks.
\end{example}

\begin{example}\label{eg:pdn}
At each time step, we first delete each existing edge with probability $p$, then
add the new arrival edges.
To understand why this kind of dynamics can also be modeled using TDN, we can
think of deleting an edge as a Bernoulli trial with success probability $p$.
Therefore, an edge surviving for $l$ time steps in the graph has probability
$(1-p)^{l-1}p$, aka the geometric distribution.
Hence, above dynamic process is equivalent to saying that each edge in a TDN has
a lifetime independently sampled from a geometric distribution, i.e.,
$\mathit{Pr}(l_\tau(e)=l)=(1-p)^{l-1}p$.
We refer to such networks as probabilistic time-decaying dynamic interaction
networks.
\end{example}

The time-decaying mechanism not only helps discard outdated edges in $G_t$ but
also reduce the storage of maintaining $G_t$ in computer main memory.
For example, if edge lifetimes follow a geometric distribution
$\mathit{Pr}(l_\tau(e)=l)=(1-p)^{l-1}p$ (with $L=\infty$), and we assume at most
$m$ new edges arrive at each time, then the memory needed to store $G_t$ at any
time $t$ is upper bounded by $O(\sum_{i=0}^\infty m(1-p)^i)=O(m/p)$.

In the following discussion, in order to keep our methodology general, we will not
assume a particular form of $l_\tau(e)$ but {\bf assume that $l_\tau(e)$ is given
as a user-chosen input} to our framework so that $G_t$ can be stored in computer
main memory.
We also introduce some shortcut notations for later use.
Given an edge $e\in E_t$, we will use $u_e, v_e, \tau_e$ and $l_e$ to denote the
edge's attributes, and lifetime at time $t$, respectively.

\subsection{Problem Formulation}

Up till now, we have modeled an interaction stream as a TDN.
We now define the {\em influence spread} on TDNs.

\begin{definition}[Influence Spread on TDNs]\label{def:influence}
At time $t$, the influence spread of a set of nodes $S\subseteq V_t$ is the
number of distinct nodes that are reachable from $S$ in $G_t$, i.e.,
\[
f_t(S)\triangleq
|\{v\in V_t\colon v\text{ is reachable from nodes $S$ in }G_t\}|.
\]
\end{definition}

Remember that each edge $e\in E_t$ represents that node $u_e$ can influence node
$v_e$.
Thus, Definition~\ref{def:influence} actually states that, if nodes in $S$ can
influence many nodes in $G_t$ in a cascading manner, $S$ has large influence in
$G_t$.
It is not hard to see that $f_t$ satisfies following property~\cite{Kempe2003}.

\begin{theorem}\label{thm:submodular}
$f_t\colon 2^{V_t}\mapsto\mathbb{R}_{\geq 0}$ defined in
Definition~\ref{def:influence} is a normalized monotone submodular set
function\footnote{A set function $f\colon 2^V\mapsto \mathbb{R}_{\geq 0}$ is
monotone if $f(S)\leq f(T)$ holds for all $S\subseteq T\subseteq V$; $f$ is
submodular if $f(S\cup \{s\}) - f(S)\geq f(T\cup\{s\}) - f(T)$ holds for all
$S\subseteq T\subseteq V$ and $s\in V\backslash T$; $f$ is normalized if
$f(\emptyset)=0$ (\cite{Nemhauser1978}).}.
\end{theorem}

Of course, readers can define more complicated influence spread in TDNs.
As long as Theorem~\ref{thm:submodular} holds, our developed framework, which will
be elaborated on in the remainder of this paper, is still applicable.
We are now ready to formulate the influential nodes tracking problem.

\begin{problem}[Tracking Influential Nodes in TDNs]\label{problem:tracking}
Let $G_t=(V_t,E_t)$ denote an evolving TDN at time $t$.
Let $k>0$ denote a given budget.
At any time $t$, we want to find a subset of nodes $S_t^*\subseteq V_t$ with
cardinality at most $k$ such that these nodes have the maximum influence spread
on $G_t$, i.e., $f_t(S_t^*) = \max_{S\subseteq V_t\wedge |S|\leq k}f_t(S)$.
\end{problem}

\remarks.

\textbullet\,Figure~\ref{fig:tdn} gives an example to show the time-varying nature
of influential nodes in TDNs.

\textbullet\,Problem~\ref{problem:tracking} is NP-hard in general.
When $G_t$ is large in scale, it is only practical to find approximate solutions.
We say a solution $S_t$ is an $\alpha$-approximate solution if $f_t(S_t)\geq
\alpha f_t(S_t^*)$ where $0<\alpha<1$.

\textbullet\,A straightforward way to solve Problem~\ref{problem:tracking} is to
re-run existing algorithms designed for static networks, e.g., the greedy
algorithm~\cite{Kempe2003}, at every time the network is updated.
This approach gives a $(1-1/e)$-approximate solution with time complexity
$O(k|V_t|\gamma)$ where $\gamma$ is the time complexity of evaluating $f_t$ on
$G_t$.
We hope to find faster methods than this baseline method with comparable
approximation guarantees.

\section{A Basic Approach}
\label{sec:basicreduction}

In this section, we elaborate a basic approach on solving
Problem~\ref{problem:tracking}.
To motivate this basic approach, we first consider solving a special problem:
tracking influential nodes over addition-only dynamic interaction networks (ADNs,
refer to Example~\ref{eg:adn}).
We find that this special problem is closely related to a well-studied
insertion-only streaming submodular optimization problem~\cite{Badanidiyuru2014a}.
It thus inspires us to design {\sc SieveADN} to solve this special problem
efficiently.
Using {\sc SieveADN} as a basic building block, we show how to design {\sc
BasicReduction} to solve Problem~\ref{problem:tracking} on general TDNs.

\subsection{\textsc{SieveADN}: Tracking Influential Nodes over ADNs}

In an ADN, each edge has an infinite lifetime.
Arriving new edges simply accumulate in $G_t$.
Therefore, the influence spread of a fixed set of nodes cannot decrease, i.e.,
$f_t(S)\geq f_{t'}(S)$ holds for all $S\subseteq V_{t'}$ whenever $t\geq t'$.

We find that identifying influential nodes on ADNs is related to a well-studied
{\em insertion-only streaming submodular optimization (SSO)
problem}~\cite{Badanidiyuru2014a}.
In the follows, we first briefly recap the insertion-only SSO problem, and
describe a streaming algorithm, called {\sc SieveStreaming}, that solves the
insertion-only SSO problem efficiently.

The insertion-only SSO problem considers maximizing a monotone submodular set
function $f$ over a set of elements $U$ with a cardinality constraint, i.e.,
choosing at most $k$ elements from $U$ to maximize $f$.
Each element in the set is allowed to be accessed only once in a streaming fashion.
The goal is to find algorithms that use sublinear memory and time.
One such algorithm is {\sc SieveStreaming}~\cite{Badanidiyuru2014a}.

\textsc{SieveStreaming} lazily maintains a set of thresholds
$\Theta\triangleq\{\frac{(1+\epsilon)^i}{2k}\colon (1+\epsilon)^i\in [\Delta,
2k\Delta], i\in\mathbb{Z}\}$, where $\Delta\triangleq\max_uf(\{u\})$ is the
maximum value of a singleton element seeing in the stream so far.
Each threshold $\theta\in\Theta$ is associated with a set $S_\theta$ which is
initialized to be empty.
For each arriving element $v$, its marginal gain w.r.t.~each set $S_\theta$ is
calculated, i.e., $\delta_{S_\theta}(v)\triangleq
f(S_\theta\cup\{v\})-f(S_\theta)$.
If $\delta_{S_\theta}(v)\geq\theta$ and $|S_\theta|<k$, $v$ is saved into
$S_\theta$; otherwise $v$ is not saved.
At query time, \textsc{SieveStreaming} returns a set $S_{\theta^\star}$ that has
the maximum value, i.e.,
$f(S_{\theta^\star})=\max\{f(S_\theta)\colon\theta\in\Theta\}$.
\textsc{SieveStreaming} is proven to achieve an $(1/2-\epsilon)$ approximation
guarantee.

We leverage {\sc SieveStreaming} to solve our special problem of tracking
influential nodes over ADNs as follows.
Let $\bar{E}_t$ denote a set of edges arrived at time $t$.
Let $\bar{V}_t$ denote a set of nodes whose influence spread changes due to adding
new edges $\bar{E}_t$ in $G_t$.
We feed each node in $\bar{V}_t$ to {\sc SieveStreaming} whose output is our
solution.
We refer to this algorithm as {\sc SieveADN}, as illustrated in Fig.~\ref{fig:ADN}
and described in Alg.~\ref{alg:addition-only}.

\begin{figure}[htp]
\centering

\begin{tikzpicture}[
every node/.style={inner sep=0pt,text depth=.25ex,font=\footnotesize},
nd/.style={circle, draw, thick, minimum size=5pt},
txt/.style={rectangle,inner sep=2pt,align=center},
box/.style={rectangle, draw, thick, blue, inner sep=2pt},
dbox/.style={box, inner sep=2pt, rounded corners=1mm, dashed},
th/.style={draw,rectangle,minimum width=2ex,minimum height=2ex},
item/.style={rectangle,draw=black,fill=gray,minimum width=2ex, minimum height=1.6pt},
arr/.style={thick, blue, -stealth},
]

\draw[step=0.4] (0,0) grid (4.8,.4);
\coordinate (O) at (.2,.2);
\foreach \i in {0,...,8}\node[nd, right=\i*0.4-.1 of O] (u\i) at (O) {};
\foreach \i in {9,...,11}\node[nd, blue, fill, right=\i*0.4-.1 of O] (u\i) {};
\node[dbox, fit={(u9) (u11)}] (b1) {};
\node[txt,blue,inner sep=0, right=2pt of b1] (t_t) {$\bar{V}_t$};

\node[txt, left=.2 of u0] (t_ins) {node stream:};

\node[box, txt, above=.25 of b1] (b2) {new edges $\bar{E}_t$};
\draw[arr] (b2) -- (b1);

\coordinate[below=.6 of u6] (A1);
\node[th] (th1) at (A1) {};
\draw (A1) ++(-.8ex, -.3ex) -- ++(.8ex, 0) -- ++(0, .6ex) -- ++(.8ex, 0);
\node[th,below=.2 of th1] (C1) {};
\foreach \i in {0,...,3} \node[item, above =\i*0.05 of C1.south] {};
\draw[-stealth] (th1) -- (C1);
\node[txt,left=.2 of th1] {thresholds $\Theta$:};
\node[txt,left=.2 of C1] {candidates $\{S_\theta\}_{\theta\in\Theta}$:};

\coordinate[right=.8 of A1] (A2);
\node[th] (th2) at (A2) {};
\draw (A2) ++(-.8ex, -.5ex) -- ++(.8ex, 0) -- ++(0, 1ex) -- ++(.8ex, 0);
\node[th,below=.2 of th2] (C2) {};
\foreach \i in {0,...,2} \node[item, above =\i*0.05 of C2.south] {};
\draw[-stealth] (th2) -- (C2);

\coordinate[right=.8 of A2] (Ad);
\node[txt] at (Ad) {$\cdots$};
\node[txt,below=.4 of Ad] {$\cdots$};

\coordinate[right=.8 of Ad] (A3);
\node[th] (th3) at (A3) {};
\draw (A3) ++(-.8ex, -.7ex) -- ++(.8ex, 0) -- ++(0, 1.4ex) -- ++(.8ex, 0);
\node[th,below=.2 of th3] (C3) {};
\foreach \i in {0} \node[item, above =\i*0.05 of C3.south] {};
\draw[-stealth] (th3) -- (C3);

\draw[-stealth] (u9.south) -- (th1.north east);
\draw[-stealth] (u9.south) -- (th2.north east);
\draw[-stealth] (u9.south) -- (th3.north west);

\node[box,black,below=1.5 of u9] (Am) {max};
\draw[-stealth] (C1.south east) -- (Am);
\draw[-stealth] (C2.south) -- (Am);
\draw[-stealth] (C3.south west) -- (Am);

\node[box,txt,below=.25 of Am] (rst) {influential nodes $S_t$};
\draw[arr] (Am) -- (rst);

\end{tikzpicture}

\caption{Illustration of {\sc SieveADN}}
\label{fig:ADN}
\end{figure}
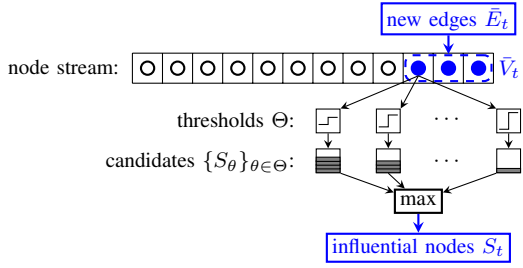

\begin{algorithm}[htp]
\KwIn{A sequence of edges arriving over time; $k$ and $\epsilon$}
\KwOut{A set of influential nodes at time $t$}

$\Delta\gets 0, \Theta\gets\emptyset$, and $\{S_\theta\}_{\theta\in\Theta}$
denotes a family of sets\;
\For{arriving edges $\bar{E}_t$ at time $t=1,2,\ldots$}{
$\bar{V}_t\gets$ a set of nodes whose influence spread changes\;
\tcp{Lines~\ref{ln:lazy1}-\ref{ln:lazy2} lazily maintain a set of thresholds.}
$\Delta\gets\max\{\Delta, \max_{v\in\bar{V}_t} f_t(\{v\})\}$\;\label{ln:lazy1}
$\Theta'\gets\{\frac{(1+\epsilon)^i}{2k}\colon (1+\epsilon)^i\in
[\Delta, 2k\Delta], i\in\mathbb{Z}\}$\;
Delete $S_\theta$ for $\theta\in\Theta\backslash\Theta'$ and let
$S_\theta\gets\emptyset$ for $\theta\in\Theta'\backslash\Theta$\;
$\Theta\gets\Theta'$\;\label{ln:lazy2}
\tcp{Lines~\ref{ln:for_node}-\ref{ln:add} filter nodes by thresholds.}
\ForEach{node $v\in \bar{V}_t$}{\label{ln:for_node}
\ForEach{threshold $\theta\in\Theta$}{\label{ln:th}
\If{$|S_\theta|<k$ and $\delta_{S_\theta}(v)\geq\theta$}{
$S_\theta\gets S_\theta\cup\{v\}$\; \label{ln:add}
}
}
}
\tcp{Return a solution having the maximum value.}
$S_t\gets S_{\theta^\star}$ where $\theta^\star=\argmax_{\theta\in\Theta}f_t(S_\theta)$\;
}
\caption{\textsc{SieveADN}}
\label{alg:addition-only}
\end{algorithm}

We emphasize that our special problem has two major differences with the
insertion-only SSO problem.
In the insertion-only SSO problem, each element appears only once in the stream,
and the objective $f$ is invariant to time.
While in our problem, we allow same nodes to appear multiple times in the node
stream, and our objective $f_t$ is time-varying.
Therefore, we still need to strictly prove that {\sc SieveADN} outputs solutions
with constant approximation guarantees.
Thanks to the property of ADNs: when new edges are added, the influence spread of
a fixed set of nodes cannot decrease.
We can leverage this property and demonstrate that {\sc SieveADN} indeed preserves
the approximation guarantee of {\sc SieveStreaming}, and achieves an approximation
factor of $(1/2-\epsilon)$.

\begin{theorem}\label{thm:ADN}
{\sc SieveADN} achieves an $(1/2-\epsilon)$ approximation guarantee.
\end{theorem}
\begin{proof}
Please refer to our Technique Report~\cite{TR}.
\end{proof}

The intermediate results maintained by {\sc SieveADN} are sets
$\{S_\theta\}_{\theta\in\Theta}$.
The following theorem states the time and space complexity of updating and storing
these sets.

\begin{theorem}\label{thm:SieveADN-complexity}
Let $b$ denote the average size of set $\bar{V}_t$.
Then {\sc SieveADN} uses $O(b\gamma\epsilon^{-1}\log k)$ time to process each
batch of edges (where $\gamma$ is the time complexity of evaluating $f_t$) and
$O(k\epsilon^{-1}\log k)$ space to maintain intermediate results.
\end{theorem}
\begin{proof}
Please refer to our Technique Report~\cite{TR}.
\end{proof}

\remarks.

\textbullet\,$b$ is typically small and $b\ll |V_t|$ in practice.
Note that even if $b=|V_t|$, \textsc{SieveADN} still has lower time complexity
than greedy algorithm which has time complexity $O(k|V_t|\gamma)$.

\textbullet\,Lines~\ref{ln:for_node}-\ref{ln:add} in Alg.~\ref{alg:addition-only}
can be easily implemented using parallel computation to further reduce the running
time.

\subsection{\textsc{BasicReduction}: Using \textsc{SieveADN} as A
Basic Building Block}

{\sc SieveADN} can be used as a basic building block to design a basic method,
called {\sc BasicReduction}, to solve Problem~\ref{problem:tracking}.
In the follows, we first describe {\sc BasicReduction}, and then use an example to
help readers understand its correctness.

Recall that $\bar{E}_t$ is a set of edges arrived at time $t$.
We partition edges $\bar{E}_t$ into (at most) $L$ groups by their lifetimes.
Let $\bar{E}_l^{(t)}\subseteq\bar{E}_t$ denote the group having lifetime $l$,
i.e., $\bar{E}_l^{(t)}\triangleq\{e\colon e\in\bar{E}_t\wedge l_e=l\}$, and
$\bar{E}_t=\cup_{l=1}^L\bar{E}_l^{(t)}$.

{\sc BasicReduction} maintains $L$ {\sc SieveADN} instances at each time $t$,
denoted by $\{\CA_i^{(t)}\}_{i=1}^L$.
At each time $t$, $\CA_i^{(t)}$ only processes edges
$\cup_{l=i}^L\bar{E}_l^{(t)}$, i.e., edges with lifetime no less than $i$.
The relationship between input edges and {\sc SieveADN} instances is
illustrated in Fig.~\ref{f:reduction}.

\begin{figure}[htp]
\centering
\subfloat[Processing new edges $\bar{E}_t$ at time $t$\label{f:reduction}]{%
\newcommand{\DrawBrac}[2][]{%
\draw [thick, #1] (#2) ++(-1ex, 2ex) -- ++(0ex, -2ex) -- ++(2ex, 0ex)
-- ++(0ex, 2ex)}

\begin{tikzpicture}[
>=stealth,
every node/.style={thick, inner sep=0pt, font=\scriptsize},
snd/.style={circle, draw, minimum size=15pt},
item/.style={rectangle,fill=black, minimum width=1ex, minimum height=1.2pt},
rarr/.style={<-, thick},
txt/.style={inner sep=0pt,align=center,font=\footnotesize},
]

\coordinate (c1) at(0,0) {};
\foreach \i in {1,...,4}{\node[snd, right=\i*1.2 of c1] (s\i) {$\CA_\i^{(t)}$};}
\node[right=.25 of s4] (sd) {$\cdots$};
\node[snd, right=.25 of sd] (sL) {$\CA_L^{(t)}$};
\node[txt,left=1ex of s1] {{\sc SieveADN}\\instances};

\node[txt, inner sep=2pt, below =.3 of s1] (temp) {$S_t$};
\draw[rarr] (temp) -- (s1);
\foreach \i [count=\x] in {2,3,4,L}{\draw[->, thick] (s\i) -- ++(0, -6mm);}

\foreach \i in {1,2,3,4,L} \coordinate[above = 1.25 of s\i] (p\i);
\node[right of=p4] (pd) {$\cdots$};

\foreach \i in {1,2,3,4,L} \draw[rarr]        (s1) -- (p\i);
\foreach \i in {2,3,4,L}   \draw[blue, rarr]  (s2) -- (p\i);
\foreach \i in {3,4,L}     \draw[red, rarr]   (s3) -- (p\i);
\foreach \i in {4,L}       \draw[green, rarr] (s4) -- (p\i);
\draw[purple, rarr] (sL) -- (pL);

\node[txt,above left=.15 and .6 of p1] {new edges\\at time $t$};
\foreach \p/\r in {1/4, 2/3, 3/2, 4/1, L/1}{
\DrawBrac{p\p};
\foreach \i in {1,...,\r}{\node[item, above =(\i-1)*0.1+0.05 of p\p] {};}
}

\coordinate[above=.7 of p1] (pbase);
\foreach \i in {1,2,3,4,L}{\node[txt] at (pbase -| p\i) {$\bar{E}_\i^{(t)}$};}

\end{tikzpicture}

}\\
\subfloat[Shifting {\sc SieveADN} instances from time $t$ to time $t+1$
\label{f:shifting}]{\newcommand{\DrawCross}[2][]{%
\draw[red,thick,#1] (#2)++(-10pt, -10pt) -- ++(20pt, 20pt)
(#2)++(-10pt, 10pt)  -- ++(20pt, -20pt)}

\begin{tikzpicture}[
every node/.style={thick, inner sep=0pt, minimum size=19pt, font=\tiny},
snd/.style={circle, draw},
dots/.style={circle, font=\footnotesize},
item/.style={rectangle,draw,fill=black, minimum width=1ex, minimum height=2pt},
arr/.style={-latex, thick, dashed},
txt_over_nd/.style={rectangle,fill=white,opacity=.1,text opacity=1, minimum size=0pt},
txt/.style={minimum size=0pt,inner sep=0pt,font=\footnotesize,anchor=west},
]

\coordinate (c1) at(0,1.4) {};
\node[snd, right=1.25 of c1, gray, dashed] (o1) {$\CA_1^{(t)}$};
\foreach \i in {2,...,3}{
\node[snd, right = \i*1.2 of c1] (o\i) {$\CA_\i^{(t)}$};}
\node[dots,right=4*1.2 of c1] (od) {$\cdots$};
\node[snd, right=5*1.2 of c1] (oL1) {};
\node[txt_over_nd] at (oL1) {$\CA_{L\!-\!1}^{(t)}$};
\node[snd, right=6*1.2 of c1] (oL) {$\CA_L^{(t)}$};

\coordinate (c2) at(0,0) {};
\foreach \i in {1,...,3}{
\node[snd, right = \i*1.2 of c2] (n\i) {};
\node[txt_over_nd] at (n\i) {$\CA_\i^{(\!t\!+\!1\!)}$};
}
\node[dots,right=4*1.2 of c2] (nd) {$\cdots$};
\node[snd, right=5*1.2 of c2] (nL1) {};
\node[txt_over_nd] at (nL1) {$\CA_{L\!-\!1}^{(\!t\!+\!1\!)}$};
\node[snd, right=6*1.2 of c2, blue] (nL) {};
\node[txt_over_nd, blue,fill=white] at (nL) {$\CA_L^{(\!t\!+\!1\!)}$};

\draw[arr] (o2) -- (n1);
\draw[arr] (o3) -- (n2);
\draw[arr] (od) -- (n3);
\draw[arr] (oL1) -- (nd);
\draw[arr] (oL) -- (nL1);

\DrawCross{o1};

\node[txt,right=.2 of c1] {$t$};
\node[txt,right=.2 of c2] {$t\!+\!1$};
\draw[-stealth,thick] (c1) -- (c2);
\end{tikzpicture}

}
\caption{Illustration of {\sc BasicReduction}.}
\label{fig:reduction}
\end{figure}
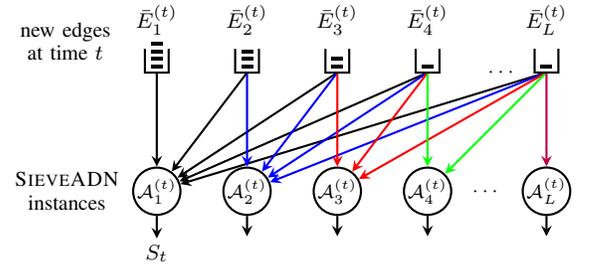
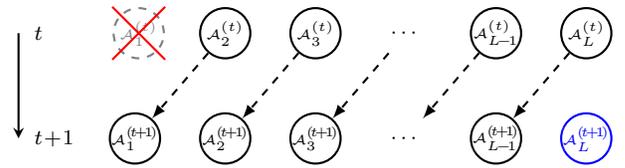

These {\sc SieveADN} instances are maintained in a way that existing instances
gradually expire and are terminated as edges processed in them expire; meanwhile,
new instances are created as new edges arrive.
Specifically, after processing edges $\bar{E}_t$, we do following operations at
the beginning of time step $t+1$ to prepare for processing edges $\bar{E}_{t+1}$:
(1) $\CA_1^{(t)}$ expires and is terminated;
(2) $\CA_i^{(t)}$ is renamed to $\CA_{i-1}^{(t+1)}$, for $i=2,\ldots,L$;
(3) a new {\sc SieveADN} instance
$\CA_L^{(t+1)}$ is created and appended at the tail.
In short, {\sc SieveADN} instances from index $2$ to $L$ at time $t$ are
``shifted'' one unit to the left, and a new {\sc SieveADN} instance is appended at
the tail, as illustrated in Fig.~\ref{f:shifting}.
Then, {\sc BasicReduction} processes $\bar{E}_{t+1}$ similarly to processing
$\bar{E}_t$ previously.
It is easier to understand its execution using the following example.

\begin{example}
{\sc BasicReduction} processes the TDN in Fig.~\ref{fig:tdn} as follows.
$L=3$ {\sc SieveADN} instances are maintained.
$\CA_1^{(t)},\CA_2^{(t)}$ and $\CA_3^{(t)}$ process edges arrived at time $t$
according to their lifetimes (see the table below).
At time $t+1$, $\CA_1^{(t)}$ expires; $\CA_2^{(t)}$ is renamed to
$\CA_1^{(t+1)}$ and continues processing edges $\{e_7,e_8,e_9\}$; $\CA_3^{(t)}$
is renamed to $\CA_2^{(t+1)}$ and continues processing edges $\{e_8,e_9\}$.
$\CA_3^{(t+1)}$ is newly created, and processes edge $e_9$.
\begin{center}
\scriptsize
\begin{tabular}{|l||l|}
\hline
\ding{115}\quad $\CA_1^{(t)}\colon\{e_1,e_2,e_3,e_4,e_5,e_6\}$
& \ding{108}\quad $\CA_1^{(t+1)}\colon\{e_3,e_4,{\color{blue}e_7,e_8,e_9}\}$ \\
\hline
\ding{108}\quad $\CA_2^{(t)}\colon\{e_3,e_4\}$
& \ding{110}\quad $\CA_2^{(t+1)}\colon\{e_4,{\color{blue} e_8,e_9}\}$        \\
\hline
\ding{110}\quad $\CA_3^{(t)}\colon\{e_4\}$
& \ding{117}\quad $\CA_3^{(t+1)}\colon\{{\color{blue} e_9}\}$                \\
\hline
\end{tabular}
\end{center}
The procedure repeats.
It is clear to see that $\CA_1^{(t)}$ always processed all of the edges in $G_t$
at any time $t$.
\end{example}

Because at any time $t$, $\CA_1^{(t)}$ processed all of the edges in $G_t$, the
output of $\CA_1^{(t)}$ is the solution at time $t$.
The complete pseudo-code of {\sc BasicReduction} is given in
Alg.~\ref{alg:reduction}.

\begin{algorithm}[ht]
\KwIn{A sequence of edges arriving over time}
\KwOut{A set of influential nodes at time $t$}
Initialize {\sc SieveADN} instances $\{\CA_i^{(1)}\colon i=1,\ldots,L\}$\;
\For{$t=1,2,\ldots$}{
\lFor{$i=1,\ldots,L$}{Feed $\CA_i^{(t)}$ with edges $\cup_{l=i}^L\bar{E}_l^{(t)}$}
$S_t\gets$ output of $\CA_1^{(t)}$\;
Terminate $\CA_1^{(t)}$\;
\lFor{$i=2,\ldots,L$}{$\CA_{i-1}^{(t+1)}\gets\CA_i^{(t)}$}
Create and initialize $\CA_L^{(t+1)}$\;
}
\caption{\textsc{BasicReduction}}
\label{alg:reduction}
\end{algorithm}

Since $\CA_1^{(t)}$ is a {\sc SieveADN} instance, its output has an approximation
factor of $(1/2-\epsilon)$ according to Theorem~\ref{thm:ADN}.
We hence have the following conclusion.

\begin{theorem}\label{thm:basic-reduction}
{\sc BasicReduction} achieves an $(1/2-\epsilon)$ approximation guarantee on
TDNs.
\end{theorem}

Furthermore, because {\sc BasicReduction} contains $L$ {\sc SieveADN} instances,
so its time complexity and space complexity are both $L$ times larger (assuming
{\sc SieveADN} instances are executed in series).

\begin{theorem}
{\sc BasicReduction} uses $O(Lb\gamma\epsilon^{-1}\log k)$ time to process each
batch of arriving edges, and $O(Lk\epsilon^{-1}\log k)$ memory to store the
intermediate results.
\end{theorem}

\remarks.

\textbullet\,{\sc SieveADN} instances in {\sc BasicReduction} can be executed in
parallel.
In this regard, the computational efficiency can be greatly improved.

\textbullet\,Notice that edges with lifetime $l$ will be input to $\CA_i^{(t)}$
with $i\leq l$.
Hence, edges with large lifetime will fan out to a large fraction of {\sc
SieveADN} instances, and incur high CPU and memory usage, especially when $L$ is
large.
This is the main bottleneck of {\sc BasicReduction}.

\textbullet\,On the other hand, edges with small lifetime only need to be
processed by a few {\sc SieveADN} instances.
If edge lifetime is mainly distributed over small lifetimes, e.g., geometrically
distributed, exponentially distributed or power-law distributed, then {\sc
BasicReduction} could be as efficient as {\sc SieveADN}.

\section{HistApprox: Improving Efficiency Using
Histogram Approximation}
\label{sec:histapprox}

{\sc BasicReduction} needs to maintain $L$ {\sc SieveADN} instances.
Processing large lifetime edges is a bottleneck.
In this section, we design {\sc HistApprox} to address its weakness.
\textsc{HistApprox} allows infinitely large $L$ and improves the efficiency of
\textsc{BasicReduction} significantly.

\subsection{Basic Idea}

{\sc BasicReduction} does not leverage the outputs of {\sc SieveADN} instances
until they are shifted to the head (refer to Fig.~\ref{f:shifting}).
We show that these intermediate outputs are actually useful and can be used to
determine whether a {\sc SieveADN} instance is redundant.
Roughly speaking, if outputs of two {\sc SieveADN} instances are close to each
other, it is not necessary to maintain both of them because one of them is
redundant; hence, we can terminate one of them earlier.
In this way, because we maintain less than $L$ {\sc SieveADN} instances, the
update time and memory usage both decrease.
On the other hand, early terminations of {\sc SieveADN} instances will incur a
loss in solution quality.
We will show how to bound this loss by using the {\em smooth submodular histogram}
property~\cite{Chen2016f,Epasto2017}.

Specifically, let $g_t(l)$ denote the value of output of instance $\CA_l^{(t)}$ at
time $t$.
We know that $g_t(1)\geq (1/2-\epsilon)\OPT_t$ at any time $t$.
Instead of maintaining $L$ {\sc SieveADN} instances, we propose {\sc HistApprox}
that removes redundant {\sc SieveADN} instances whose output values are close to
an {\em active} {\sc SieveADN} instance.
{\sc HistApprox} can be viewed as using a histogram $\{g_t(l)\colon l\in\bx_t\}$
to approximate $g_t(l)$, as illustrated in Fig.~\ref{fig:histogram}.
Here, $\bx_t=\{x_1^{(t)},x_2^{(t)},\ldots\}$ is a set of indices in descending
order, and each index $x_i^{(t)}\in\{1,\ldots,L\}$ indexes an {\em active} {\sc
SieveADN} instance\footnote{We will omit superscript $t$ if time $t$ is clear
from context.}.
Finally, {\sc HistApprox} approximates $g_t(1)$ by $g_t(x_1)$ at any time $t$.
Because {\sc HistApprox} only maintains a few {\sc SieveADN} instances at any time
$t$, i.e., those $\CA_l^{(t)}$'s with $l\in\bx_t$ and $|\bx_t|\leq L$, {\sc
HistApprox} reduces both update time and memory usage with a little loss in
solution quality which can still be bounded.

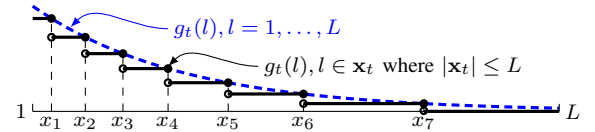
\begin{figure}[htp]
\centering
\begin{tikzpicture}[
declare function={f(\x)=1.4*exp(-0.5*\x);},
every node/.style={inner sep=0pt},
vline/.style={dashed},
hline/.style={very thick},
txt/.style={inner sep=2pt, font=\footnotesize},
circ/.style={circle, draw, thick, minimum size=2.5pt},
disk/.style={circ, fill},
]

\draw[domain=0:7,blue,smooth,very thick,densely dashed,variable=\x]
plot({\x},{f(\x)});

\draw (0,0) -- (7,0);
\draw (0,0) -- ++(0,.1);
\draw (7,0) -- ++(0,.1);

\foreach \x[count=\i] in {.25,.7,1.2,1.8,2.6,3.6,5.2}{
\draw[vline] (\x,0) -- (\x,{f(\x)});
\node[txt,anchor=north] at (\x,0) {$x_\i$};
}

\node[txt,inner sep=2pt,anchor=east] at (0,0) {$1$};
\node[txt,inner sep=2pt,anchor=west] at (7,0) {$L$};

\foreach \x/\y in {0/.25,.25/.7,.7/1.2,1.2/1.8,1.8/2.6,2.6/3.6,3.6/5.2}{
\draw[hline] (\x,{f(\y)}) -- (\y,{f(\y)});
}
\draw[hline] (5.2,0) -- (7,0);

\foreach \x/\y in {.25/.7,.7/1.2,1.2/1.8,1.8/2.6,2.6/3.6,3.6/5.2,5.2/10}{
\node[disk] at (\x,{f(\x)}) {};
\node[circ] at (\x,{f(\y)}) {};
}

\draw[-latex,blue] (1.8,1.1)
node[txt,anchor=west,blue]{$g_t(l),l=1,\ldots,L$} to[out=180,in=45] (.5,{f(.45)});

\draw[-latex] (3,.6)
node[txt,anchor=west]{$g_t(l),l\in\bx_t$ where $|\bx_t|\leq L$} to[out=180,in=45] (1.8,{f(1.7)});

\end{tikzpicture}

\caption{Approximating $\{g_t(l)\colon l=1,\ldots,L\}$ by $\{g_t(l)\colon
l\in\bx_t\}$.
(Note that $g_t(l)$ may not be a monotone function of $l$.)}
\label{fig:histogram}
\end{figure}

\subsection{Algorithm Description}

{\sc HistApprox} mainly consists of two steps: 1) creating and updating {\sc
SieveADN} instances; and 2) removing redundant {\sc SieveADN} instances.

\header{Creating and Updating \textsc{SieveADN} Instances.}
Consider a set of edges $\bar{E}_l^{(t)}$ with lifetime $l$ arrived at time $t$.
If $l\in\bx_t$, we simply feed $\bar{E}_l^{(t)}$ to $\{\CA_i^{(t)}\colon
i\in\bx_t\wedge i\leq l\}$ (as illustrated in Fig.~\ref{f:insert_exists}).
Otherwise, we need to create a new {\sc SieveADN} instance $\CA_l^{(t)}$ first,
and then feed $\bar{E}_l^{(t)}$ to $\{\CA_i^{(t)}\colon i\in\bx_t\wedge i\leq
l\}$.
There are two cases when creating $\CA_l^{(t)}$.

(1) If $l$ has no successor in $\bx_t$, as illustrated in
Fig.~\ref{f:insert_tail}, we simply create a new {\sc SieveADN} instance as
$\CA_l^{(t)}$.

(2) Otherwise, let $l^*$ denote $l$'s successor in $\bx_t$, as illustrated in
Fig.~\ref{f:insert_middle}.
Let $\CA_l^{(t)}$ denote a copy of $\CA_{l^*}^{(t)}$.
Recall that $\CA_l^{(t)}$ needs to process edges with lifetime no less than $l$ at
time $t$.
Because $\CA_{l^*}^{(t)}$ has processed edges with lifetime no less than $l^*$,
then, we only need to feed $\CA_l^{(t)}$ with edges in $G_t$ such that their
lifetimes are in interval $[l,l^*)$.

After this step, we guarantee that each active {\sc SieveADN} instances processed
the edges that they should process.

\begin{figure}[htp]
\centering
\subfloat[$l\in\bx_t$\label{f:insert_exists}]{%
\begin{tikzpicture}[
every node/.style={inner sep=0pt},
txt/.style={inner sep=3pt, anchor=north, font=\footnotesize},
add/.style={red,<-,thick,>=stealth},
vline/.style={densely dashed},
hline/.style={very thick},
circ/.style={circle, draw, thick, minimum size=3pt},
disk/.style={circ, fill},
change/.style={-latex, thick, blue,double},
]

\path[use as bounding box] (0,-.3) rectangle (7,.7);
\begin{scope}
\draw (.2,0) -- (2.8,0);
\draw[vline] (1,0) -- (1,.6) (2,0) -- (2,.3);
\draw[hline] (.2,.6) -- (1,.6) (1,.3)--(2,.3) (2,.1)--(2.8,.1);
\node[circ] at (1,.3) {};
\node[circ] at (2,.1) {};
\node[disk] at (1,.6) {};
\node[disk] at (2,.3) {};
\node[txt] at (1,0) {};
\node[txt] at (2,0) {$l$};
\draw[add] (2,.6) -- ++(0,.2);
\draw[change] (3.3,.4) -- ++(.5,0);

\coordinate (A1) at (1,.6);
\coordinate (B1) at (2,.3);
\coordinate (C1) at (2.8,.1);
\end{scope}

\begin{scope}[xshift=4cm]
\draw (.2,0) -- (2.8,0);
\draw[vline] (1,0) -- (1,.7) (2,0) -- (2,.4);
\draw[hline,red] (.2,.7) -- (1,.7) (1,.4) -- (2,.4);
\draw[hline] (2,.1)--(2.8,.1);
\node[circ,red] at (1,.4) {};
\node[circ] at (2,.1) {};
\node[disk,red] at (1,.7) {};
\node[disk,red] at (2,.4) {};
\node[txt] at (1,0) {};
\node[txt] at (2,0) {$l$};

\coordinate (A2) at (1,.6);
\coordinate (B2) at (2,.3);
\coordinate (C2) at (2,.1);
\end{scope}

\draw[dotted] (A1) -- (A2) (B1) -- (B2) (C1) -- (C2);
\end{tikzpicture}

}\\
\subfloat[$l\notin\bx_t$ and $l$ has no successor\label{f:insert_tail}]{%
\begin{tikzpicture}[
every node/.style={inner sep=0pt},
curve/.style={domain=.2:2.8,smooth,thick,densely dotted,variable=\x},
txt/.style={inner sep=3pt, anchor=north, font=\footnotesize},
add/.style={red,<-,thick,>=stealth},
vline/.style={densely dashed},
hline/.style={very thick},
circ/.style={circle, draw, thick, minimum size=3pt},
disk/.style={circ, fill},
change/.style={-latex, thick, blue,double},
]

\path[use as bounding box] (0,-.3) rectangle (7,.7);

\begin{scope}
\draw (.2,0) -- (2.8,0);
\draw[vline] (1,0) -- (1,.4);
\draw[hline] (.2,.4) -- (1,.4) (1,0) -- (2.8,0);
\node[circ] at (1,0) {};
\node[disk] at (1,.4) {};
\node[txt] at (1,0) {};
\node[txt] at (2,0) {$l$};
\node[txt,anchor=west] at (2.8,0) {$L$};
\draw[add] (2,.5) -- ++(0,.2);
\draw (2,0) -- ++(0,.1);
\draw (2.8,0) -- ++(0,.1);
\draw[change] (3.3,.4) -- ++(.5,0);

\coordinate (A1) at (1,.4);
\end{scope}

\begin{scope}[xshift=4cm]
\draw (.2,0) -- (2.8,0);
\draw[vline] (1,0) -- (1,.6) (2,0) -- (2,.3);
\draw[hline,red] (.2,.6) -- (1,.6) (1,.3) -- (2,.3);
\draw[hline] (2,0) -- (2.8,0);
\node[circ,red] at (1,.3) {};
\node[circ] at (2,0) {};
\node[disk,red] at (1,.6) {};
\node[disk,red] at (2,.3) {};
\node[txt] at (1,0) {};
\node[txt] at (2,0) {$l$};
\node[txt,anchor=west] at (2.8,0) {$L$};
\draw (2.8,0) -- ++(0,.1);

\coordinate (A2) at (1,.4);
\end{scope}

\draw[dotted] (A1) -- (A2);
\end{tikzpicture}

}\\
\subfloat[$l\notin\bx_t$ and $l$ has a successor $l^*$\label{f:insert_middle}]{%
\begin{tikzpicture}[
every node/.style={inner sep=0pt},
txt/.style={inner sep=3pt, anchor=north, font=\footnotesize},
add/.style={red,<-,thick,>=stealth},
vline/.style={densely dashed},
hline/.style={very thick},
circ/.style={circle, draw, thick, minimum size=3pt},
disk/.style={circ, fill},
change/.style={-latex, thick, blue,double},
]

\path[use as bounding box] (0,-.3) rectangle (7,.7);
\begin{scope}
\draw (.2,0) -- (2.8,0);
\draw[vline] (.7,0) -- (.7,.6) (2.3,0) -- (2.3,.2);
\draw[hline] (.2,.6)--(.7,.6) (.7,.2)--(2.3,.2) (2.3,0)--(2.8,0);
\node[circ] at (.7,.2) {};
\node[circ] at (2.3,0) {};
\node[disk] at (.7,.6) {};
\node[disk] at (2.3,.2) {};
\node[txt] at (.7,0) {};
\node[txt] at (1.5,0) {$l$};
\node[txt] at (2.3,0) {$l^*$};
\draw[add] (1.5,.6) -- ++(0,.2);
\draw (1.5,0) -- ++(0,.1);
\draw[change] (3.3,.5) -- ++(.5,0);

\coordinate (A1) at (.7,.6);
\coordinate (B1) at (2.3,.2);
\end{scope}

\begin{scope}[xshift=4cm]
\draw (.2,0) -- (2.8,0);
\draw[vline] (.7,0) -- (.7,.7) (1.5,0) -- (1.5,.4) (2.3,0) -- (2.3,.2);
\draw[hline,red] (.2,.7) -- (.7,.7) (.7,.4) -- (1.5,.4);
\draw[hline] (1.5,.2)--(2.3,.2) (2.3,0)--(2.8,0);
\node[circ,red] at (.7,.4) {};
\node[circ] at (1.5,.2) {};
\node[circ] at (2.3,0) {};
\node[disk,red] at (.7,.7) {};
\node[disk,red] at (1.5,.4) {};
\node[disk] at (2.3,.2) {};
\node[txt] at (.7,0) {};
\node[txt] at (1.5,0) {$l$};
\node[txt] at (2.3,0) {$l^*$};

\coordinate (A2) at (.7,.6);
\coordinate (B2) at (2.3,.2);
\end{scope}

\draw[dotted] (A1) -- (A2) (B1) -- (B2);
\end{tikzpicture}

}
\caption{Processing arrived edges with lifetime $l$}
\end{figure}
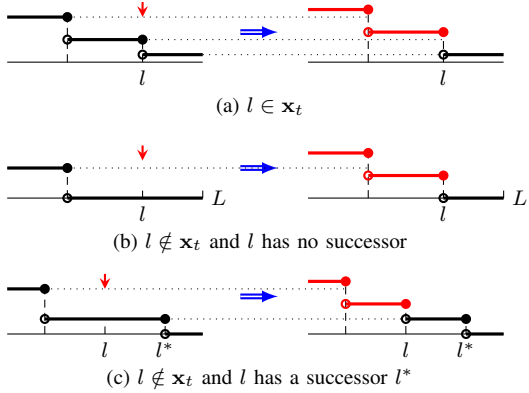

\header{Removing Redundant \textsc{SieveADN} Instances.}
Intuitively, if the outputs of two {\sc SieveADN} instances are close to each
other, there is no need to maintain both of them because one of them is redundant.
To quantify the redundancy of {\sc SieveADN} instances, we need the following
formal definition.

\begin{definition}[$\epsilon$-redundancy]
At time $t$, consider two {\sc SieveADN} instances $\CA_i^{(t)}$ and
$\CA_l^{(t)}$ with $i<l$.
We say $\CA_l^{(t)}$ is $\epsilon$-redundant with $\CA_i^{(t)}$ if their exists
$j>l$ such that $g_t(j)\geq (1-\epsilon)g_t(i)$.
\end{definition}

The above definition simply states that, since $\CA_i^{(t)}$ and $\CA_j^{(t)}$ are
already close with each other, then the {\sc SieveADN} instances between them are
considered to be redundant and it is not necessary to maintain them.
In {\sc HistApprox}, after processing each batch of edges, we check the outputs of
each {\sc SieveADN} instances and terminate those redundant ones.

The complete {\sc HistApprox} is given in Alg.~\ref{alg:histogram}.
Pseudo-codes of the two steps are described in \ProcessEdges and
\ReduceRedundancy, respectively.

\begin{algorithm}[ht]
\KwIn{A sequence of edges arriving over time}
\KwOut{A set of influential nodes at each time $t$}

$\bx_1\gets\emptyset$\;
\For{$t=1,2,\ldots$}{
\lForEach{$l=1,\ldots,L$}{\ProcessEdges{$\bar{E}_l^{(t)}$}}
$S_t\gets$ output of $\CA_{x_1}^{(t)}$\;
\lIf{$x_1=1$}{Terminate $\CA_1^{(t)}$, let $\bx_t\gets\bx_t\backslash\{1\}$}
\For{$i=1,\ldots,|\bx_t|$}{
$\CA^{(t+1)}_{x_i-1}\gets\CA^{(t)}_{x_i}$,
$x_i^{(t+1)}\gets x_i^{(t)}-1$\;
}
}

\myproc{\ProcessEdges{$\bar{E}_l^{(t)}$}}{ \label{ln:process_edges}
\If{$l\notin\bx_t$}{
\uIf(\tcp*[f]{refer to Fig.~\ref{f:insert_tail}})
{$l$ has no successor in $\bx_t$}{
Create and initialize $\CA_l^{(t)}$\;
}
\Else(\tcp*[f]{refer to Fig.~\ref{f:insert_middle}}){
Let $l^*$ denote the successor of $l$ in $\bx_t$\;
$\CA_l^{(t)}\gets$ a copy of $\CA_{l^*}^{(t)}$\;
Feed $\CA_l^{(t)}$ with edges $\{e\colon e\in E_t\wedge l\leq l_e<l^*\}$.
}
$\bx_t\gets\bx_t\cup\{l\}$\;
}
\lForEach{$l'\in\bx_t$ and $l'\leq l$}{Feed $\CA_{l'}^{(t)}$ with $\bar{E}_l^{(t)}$}
\ReduceRedundancy{}\;
}

\myproc{\ReduceRedundancy{}}{\label{ln:remove_redundancy}
\ForEach{$i\in\bx_t$}{
Find the largest $j>i$ in $\bx_t$ s.t. $g_t(j)\geq (1-\epsilon)g_t(i)$\;
Delete each index $l\in\bx_t$ s.t. $i<l<j$ and kill $\CA_l^{(t)}$\;
}
}
\caption{\textsc{HistApprox}}
\label{alg:histogram}
\end{algorithm}

\subsection{Algorithm Analysis}
\label{ss:analysis}

We now theoretically show that {\sc HistApprox} achieves a constant approximation
factor.

Notice that indices $x\in\bx_t$ and $x+1\in\bx_{t-1}$ are actually the same index
but appear at different time.
In general, we say $x'\in\bx_{t'}$ is an {\em ancestor} of $x\in\bx_t$ if $t'\leq
t$ and $x'=x+t-t'$.
An index and its ancestors will be considered as the same index.
We will use $x'$ to denote $x$'s ancestor in the follows.

First, histogram $\{g_t(l)\colon l\in\bx_t\}$ maintained by {\sc HistApprox} has
the following property.

\begin{theorem}\label{thm:histogram_property}
For two consecutive indices $x_i,x_{i+1}\in\bx_t$ at any time $t$, one of the
following two cases holds:

\begin{itemize}
\item[\textbf{C1}] $G_t$ contains no edge with lifetime in interval $(x_i,
x_{i+1})$.
\item[\textbf{C2}] $g_{t'}(x_{i+1}')\geq (1-\epsilon)g_{t'}(x_i')$ at some time
$t'\leq t$, and from time $t'$ to $t$, there is no edge with lifetime between
$x_i'$ and $x_{i+1}'$ arrived, exclusive.
\end{itemize}
\end{theorem}
\begin{proof}
Please refer to our Technique Report~\cite{TR}.
\end{proof}

Histogram with property C2 is known as a {\em smooth
histogram}~\cite{Braverman2007}.
Smooth histogram together with the submodularity of objective function can ensure
a constant factor approximation factor of $g_t(x_1)$.

\begin{theorem}\label{thm:histogram_guarantee}
{\sc HistApprox} achieves a $(1/3 - \epsilon)$ approximation guarantee, i.e.,
$g_t(x_1)\geq (1/3-\epsilon)\OPT_t$ at any time $t$.
\end{theorem}
\begin{proof}
The high-level idea is to leverage the property found in
Theorem~\ref{thm:histogram_property} and {\em smooth submodular histogram}
property reported in~\cite{Chen2016f,Epasto2017}.
Please refer to our Technique Report~\cite{TR} for details.
\end{proof}

{\sc HistApprox} also significantly reduces the complexity.

\begin{theorem}\label{thm:histogram_complexity}
{\sc HistApprox} uses $O(b(\gamma+1)\epsilon^{-2}\log^2k)$ time to process each
batch of edges and $O(k\epsilon^{-2}\log^2k)$ memory to store intermediate
results.
\end{theorem}
\begin{proof}
Please refer to our Technique Report~\cite{TR}.
\end{proof}

\remark.
Can {\sc HistApprox} achieve the $(1/2-\epsilon)$ approximation guarantee?
Notice that {\sc HistApprox} outputs slightly worse solution than {\sc
BasicReduction} due to the fact that $\CA_{x_1}^{(t)}$ did not process all of
the edges in $G_t$, i.e., those edges with lifetime less than $x_1$ in $G_t$ are
not processed.
Therefore, we can slightly modify Alg.~\ref{alg:histogram} and feed
$\CA_{x_1}^{(t)}$ with these unprocessed edges.
Then, {\sc HistApprox} will output a solution with $(1/2-\epsilon)$ approximation
guarantee.
In practice, we observe that $g_t(x_1)$ is already close to $\OPT_t$, hence it is
not necessary to conduct this additional processing.
More discussion on further improving the efficiency of {\sc HistApprox} can be
found in our technical report~\cite{TR}.

\section{Experiments}
\label{sec:experiments}

In this section, we validate the performance of our methods on various real-world
interaction datasets.
A summary of these datasets is given in Table~\ref{tab:data}.

\begin{table}[htp]
\centering
\caption{Summary of interaction datasets}
\label{tab:data}
\begin{tabular}{|c|r|r|}
\hline
{\bf interaction dataset} & {\bf \# of nodes}   & {\bf \# of interactions} \\
\hline
Brightkite (users/places) & $51,406/772,966$    & $4,747,281$              \\
Gowalla (users/places)    & $107,092/1,280,969$ & $6,442,892$              \\
\hline
Twitter-Higgs             & $304,198$           & $555,481$                \\
Twitter-HK                & $49,808$            & $2,930,439$              \\
\hline
StackOverflow-c2q         & $1,627,635$         & $13,664,641$             \\
StackOverflow-c2a         & $1,639,761$         & $17,535,031$             \\
\hline
\end{tabular}
\end{table}

\subsection{Interaction Datasets}
\label{sec:dataset}

\bullethdr{LBSN Check-in Interactions}~\cite{snap}.
Brightkite and Gowalla are two location based online social networks (LBSNs) where
users can check in places.
A check-in record is viewed as an interaction between a user and a place.
If user $u$ checked in a place $y$ at time $t$, we denote this interaction by
$\langle {y,u,t} \rangle$.
Because $\langle {y,u,t} \rangle$ implies that place $y$ attracts user $u$ to
check in, it thus reflects $y$'s influence on $u$.
In this particular example, a place's influence (or {\em popularity}), is
equivalent to the number of distinct users who checked in the place.
Our goal is to maintain $k$ most popular places at any time.

\bullethdr{Twitter Re-tweet/Mention Interactions}~\cite{snap,Zhao2015}.
In Twitter, a user $v$ can re-tweet another user $u$'s tweet, or mention another
user $u$ (i.e., @$u$).
We denote this interaction by $\langle {u,v,t} \rangle$, which reflects $u$'s
influence on $v$ at time $t$.
We use two Twitter re-tweet/mention interaction datasets Twitter-Higgs and
Twitter-HK.
Twitter-Higgs dataset is built after monitoring the tweets related to the
announcement of the discovery of a new particle with the features of the elusive
Higgs boson on 4th July 2012.
The detailed description of this dataset is given in~\cite{snap}.
Twitter-HK dataset is built after monitoring the tweets related to Umbrella
Movement happened in Hong Kong from September to December in 2014.
The detailed collection method and description of this dataset is given
in~\cite{Zhao2015}.
Our goal is to maintain $k$ most influential users at any time.

\bullethdr{Stack Overflow Interactions}~\cite{snap}.
Stack Overflow is an online question and answer website where users can ask
questions, give answers, and comment on questions/answers.
We use the comment on question interactions (StackOverflow-c2q) and comment on
answer interactions (StackOverflow-c2a).
In StackOverflow-c2q, if user $v$ comments on user $u$'s question at time $t$, we
create an interaction $\langle u,v,t \rangle$ (which reflects $u$'s influence on
$v$ because $u$ attracts $v$ to answer his question).
Similarly, in StackOverflow-c2a, if user $v$ comments on user $u$'s answer at time
$t$, we create an interaction $\langle u,v,t \rangle$.
Our goal is to maintain $k$ most influential users at any time.

\subsection{Settings}

We assign each interaction a lifetime sampled from a geometric distribution
$\mathit{Pr}(l_\tau(e)=l)\propto (1-p)^{l-1}p$ truncated at the maximum lifetime
$L$.
As discussed in Example~\ref{eg:pdn}, this particular lifetime assignment actually
means that we forget each existing interaction with probability $p$.
Here $p$ controls the skewness of the distribution, i.e., for larger $p$, more
interactions tend to have short lifetimes.
We emphasize that other lifetime assignment methods are also allowed in our
framework.

Each interaction will be input to our algorithms sequentially according to their
timestamps and we assume one interaction arrives at a time.

Note that in the following experiments, all of our algorithms are {\bf implemented
in series}, on a laptop with a $2.66$GHz Intel Core I3 CPU running Linux Mint
19.
We refer to our technical report~\cite{TR} for more discussion on implementation
details.

\subsection{Baselines}

\bullethdr{Greedy}~\cite{Nemhauser1978}.
We run a greedy algorithm on $G_t$ which chooses a node with the maximum marginal
gain in each round, and repeats $k$ rounds.
We apply the {\em lazy evaluation} trick~\cite{Minoux1978} to further reduce the
number of oracle calls\footnote{An {\em oracle call} refers to an evaluation of
$f_t$.}.

\bullethdr{Random.}
We randomly pick a set of $k$ nodes from $G_t$ at each time $t$.

\bullethdr{DIM}\footnote{\url{https://github.com/todo314/dynamic-influence-analysis}}~\cite{Ohsaka2016}.
DIM updates the index (or called sketch) dynamically and supports handling fully
dynamical networks, i.e., edges arrive, disappear, or change the diffusion
probabilities.
We set the parameter $\beta=32$ as suggested in~\cite{Ohsaka2016}.

\bullethdr{IMM}\footnote{\url{https://sourceforge.net/projects/im-imm/}}~\cite{Tang2015}.
IMM is an index-based method that uses martingales, and it is designed for
handling {\em static graphs}.
We set the parameter $\epsilon=0.3$.

\bullethdr{TIM+}\footnote{\url{https://sourceforge.net/projects/timplus/}}~\cite{Tang2014}.
TIM+ is an index-based method with the two-phase strategy, and it is designed for
handling {\em static graphs}.
We set the parameter $\epsilon=0.3$.

Note that the first five methods can be used to address our problem (even though
some of them are obviously not efficient); while SIM cannot be used to address our
general problem as it is only designed for the sliding-window model.
Some of the above methods assume that the diffusion probabilities on each edge in
$G_t$ are given in advance.
Strictly, these probabilities should be learned use complex inference
methods~\cite{Saito2008,Goyal2010,Kutzkov2013}, which will further harm the
efficiency of existing methods.
Here, for simplicity, if node $u$ imposed $x$ interactions on node $v$ at time
$t$, we assign edge $(u,v)$ a diffusion probability $p_{uv}=2/(1+\exp(-0.2x))-1$.

When evaluating computational efficiency, we prefer to use the {\em number of
oracle calls}.
Because an oracle call is the most expensive operation in each algorithm, and,
more importantly, the number of oracle calls is independent of algorithm
implementation (e.g., parallel or series) and experimental hardwares.
The less the number of oracle calls an algorithm uses, the faster the algorithm
is.

\subsection{Results}
\label{sec:results}

\header{Comparing \textsc{BasicReduction} with \textsc{HistApprox}.}
We first compare the performance of {\sc HistApprox} and {\sc BasicReduction}.
Our aim is to understand how lifetime distribution affects the efficiency of {\sc
BasicReduction}, and how significantly {\sc HistApprox} improves the efficiency
upon {\sc BasicReduction}.

We run the two methods on the two LBSN datasets for $5000$ time steps, and
maintain $k=10$ places at each time step.
We vary $p$ from $0.001$ to $0.008$.
For each $p$, we average the solution value and number of oracle calls along time,
and show the results in Fig.~\ref{fig:cmp_basic_hist}.

\begin{figure}[htp]
\centering
\subfloat[Solution value (Brightkite)]{%
\includegraphics[width=.5\linewidth]{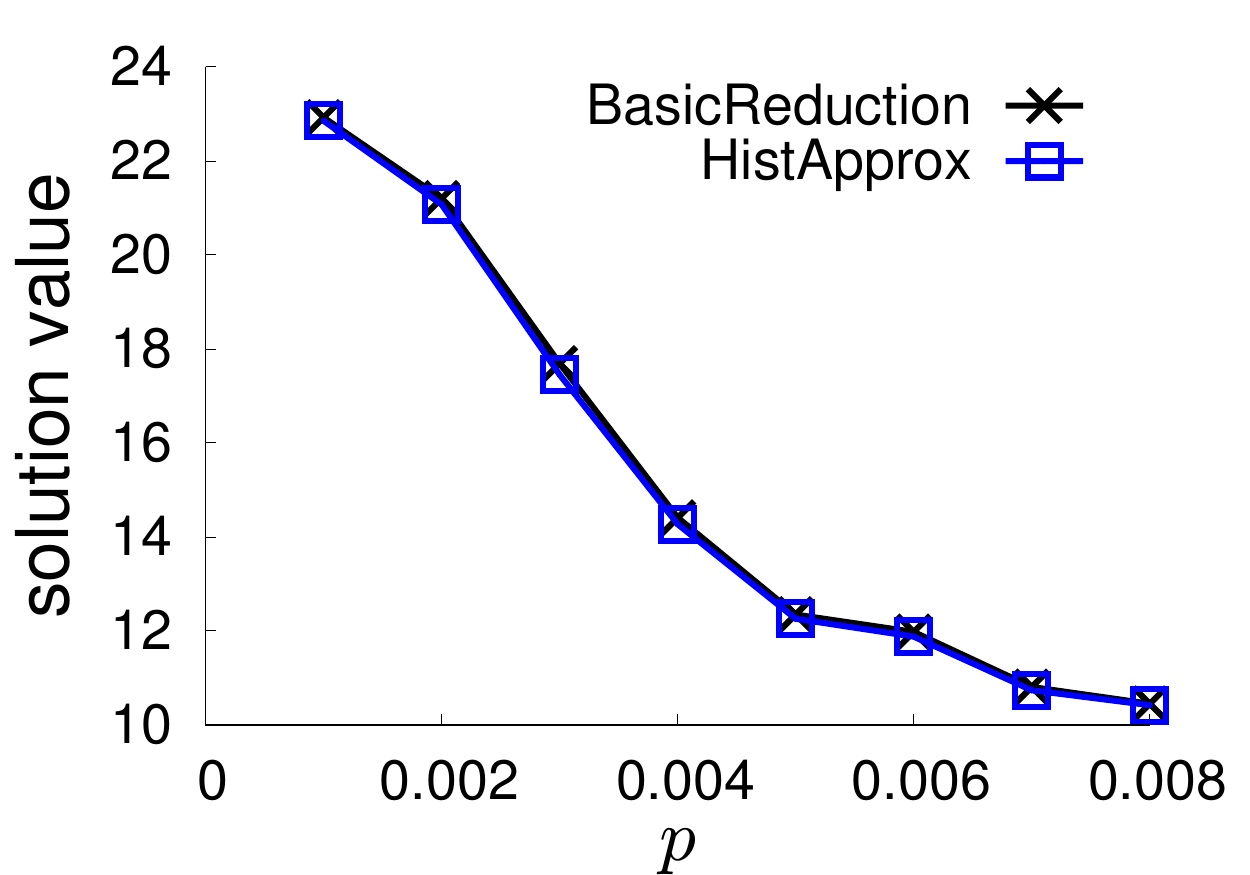}}
\subfloat[Efficiency (Brightkite)]{%
\includegraphics[width=.5\linewidth]{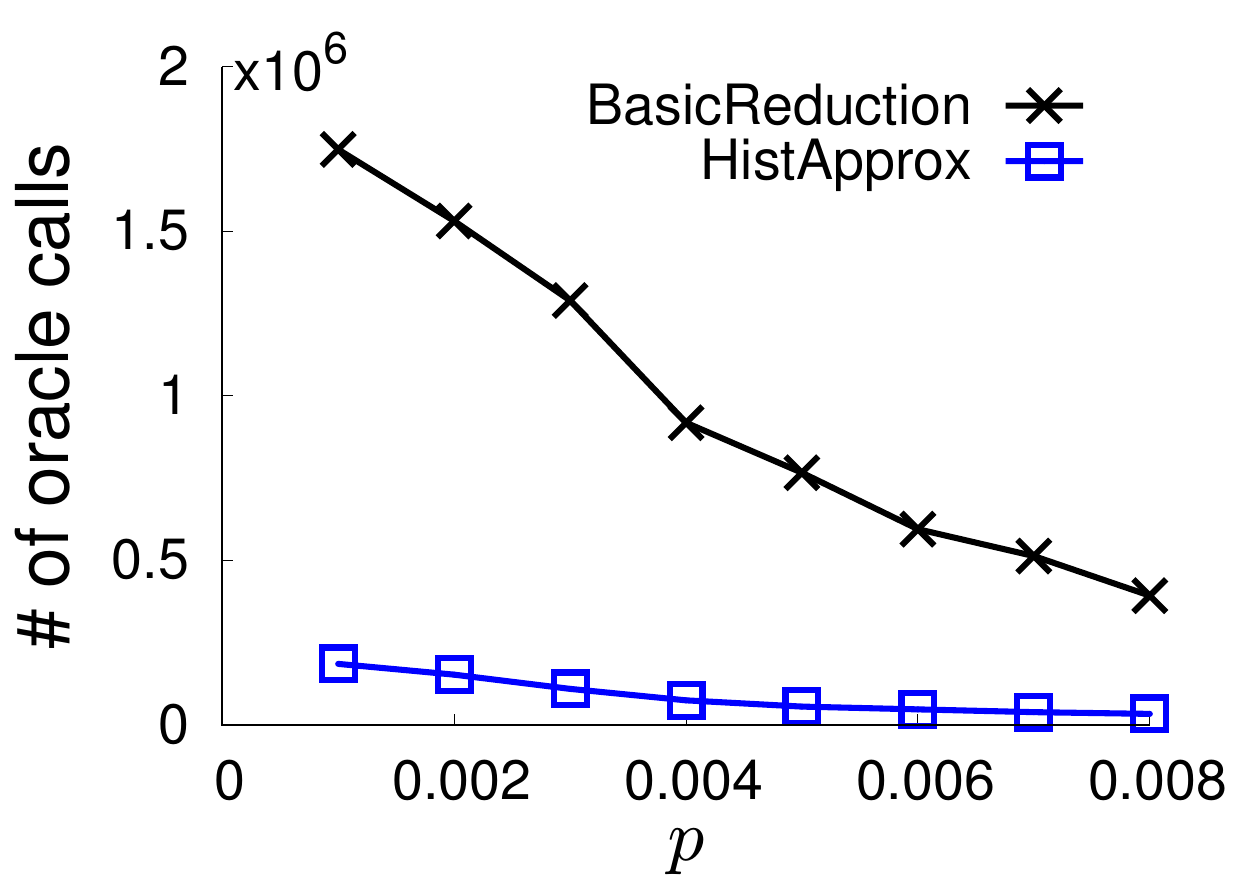}} \\
\subfloat[Solution value (Gowalla)]{%
\includegraphics[width=.5\linewidth]{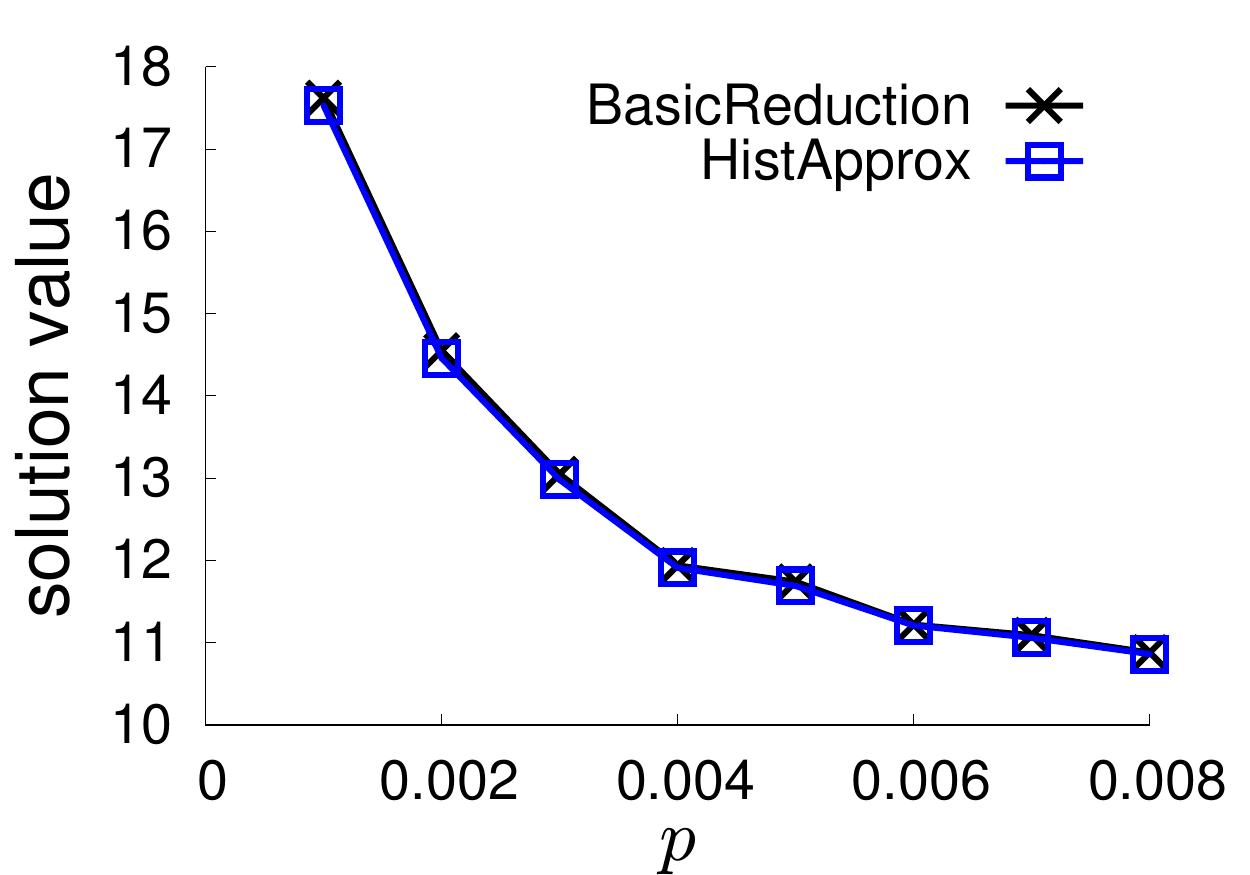}}
\subfloat[Efficiency (Gowalla)]{%
\includegraphics[width=.5\linewidth]{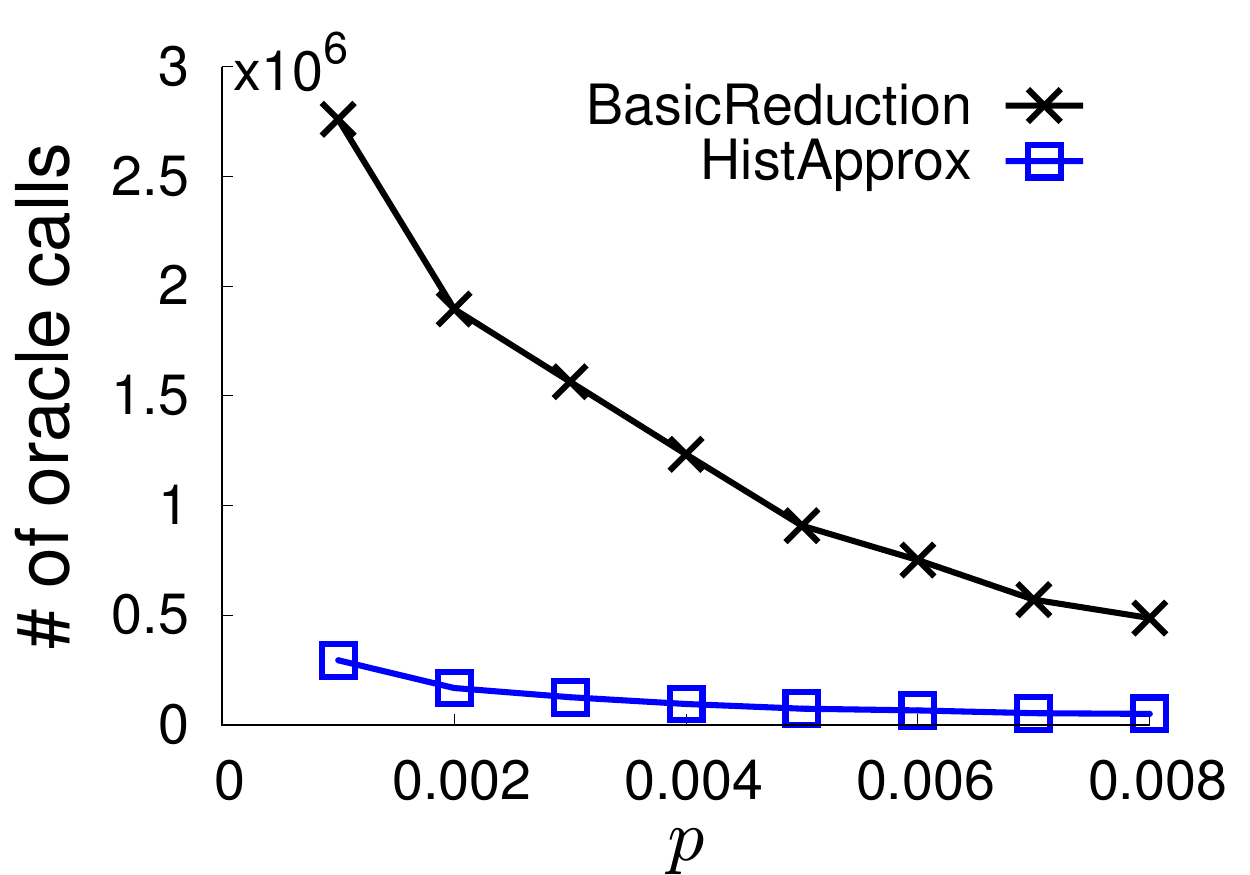}}
\caption{Comparing {\sc BasicReduction} with {\sc HistApprox}. ($\epsilon=0.1,
k=10, L=1000$, each point is averaged over $5000$ time steps)}
\label{fig:cmp_basic_hist}
\end{figure}

First, we observe that the solution value of {\sc HistApprox} is very close to
{\sc BasicReduction}.
The solution value ratio of {\sc HistApprox} to {\sc BasicReduction} is larger
than $0.98$.
Second, we observe that the number of oracle calls of {\sc BasicReduction}
decreases as $p$ increases, i.e., more interactions tend to short lifetimes.
This observation consists with our analysis that processing edges with long
lifetimes is the main bottleneck of {\sc BasicReduction}.
Third, {\sc HistApprox} needs much less oracle calls than {\sc BasicReduction},
and the ratio of the number of oracle calls of {\sc HistApprox} to {\sc
BasicReduction} is less than $0.1$.

This experiment demonstrates that {\sc HistApprox} outputs solutions with value
very close to {\sc BasicReduction} but is much efficient than {\sc
BasicReduction}.

\header{Evaluating Solution Quality of \textsc{HistApprox} Over Time.}
We conduct more in-depth analysis of {\sc HistApprox} in the following
experiments.
First, we evaluate the solution quality of {\sc HistApprox} in comparison with
other baseline methods.

We run {\sc HistApprox} with $\epsilon=0.1, 0.15$ and $0.2$ for $5000$ time steps
on the six datasets, respectively, and maintain $k=10$ nodes at each time step.
We compare the solution value of {\sc HistApprox} with Greedy and Random.
The results are depicted in Fig.~\ref{fig:quality}.

\begin{figure}[t]
\centering
\subfloat[Brightkite]{%
\includegraphics[width=.5\linewidth]{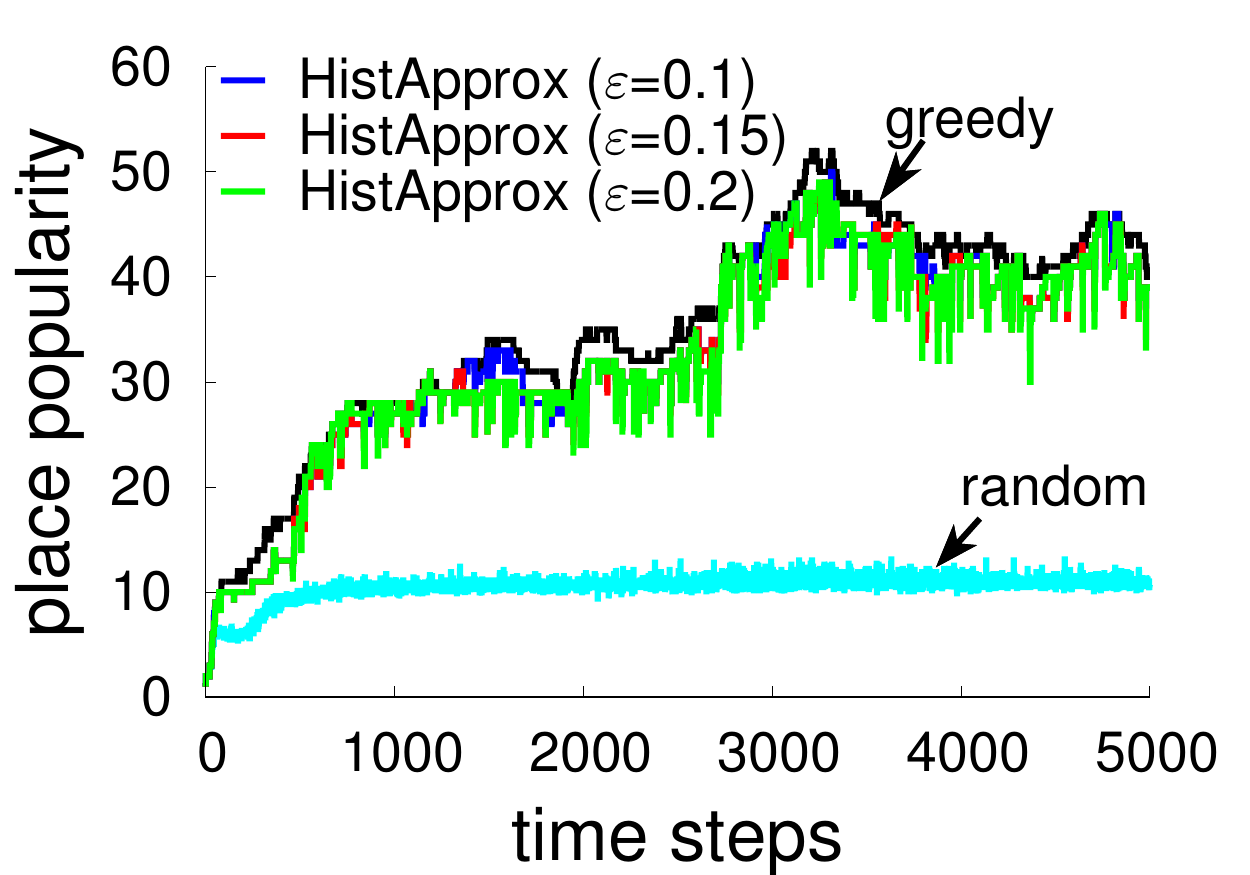}}
\subfloat[Gowalla]{%
\includegraphics[width=.5\linewidth]{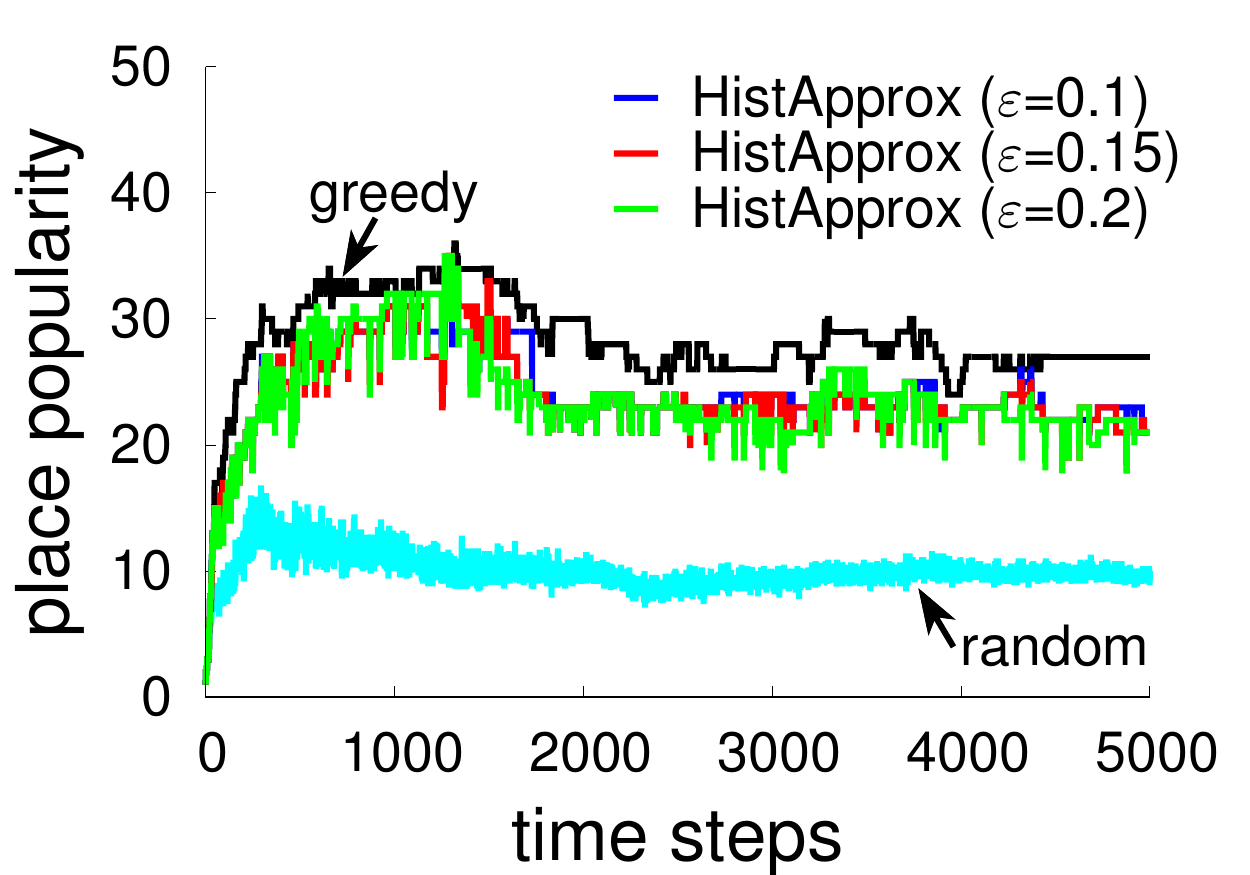}}\\
\subfloat[Twitter-Higgs]{%
\includegraphics[width=.5\linewidth]{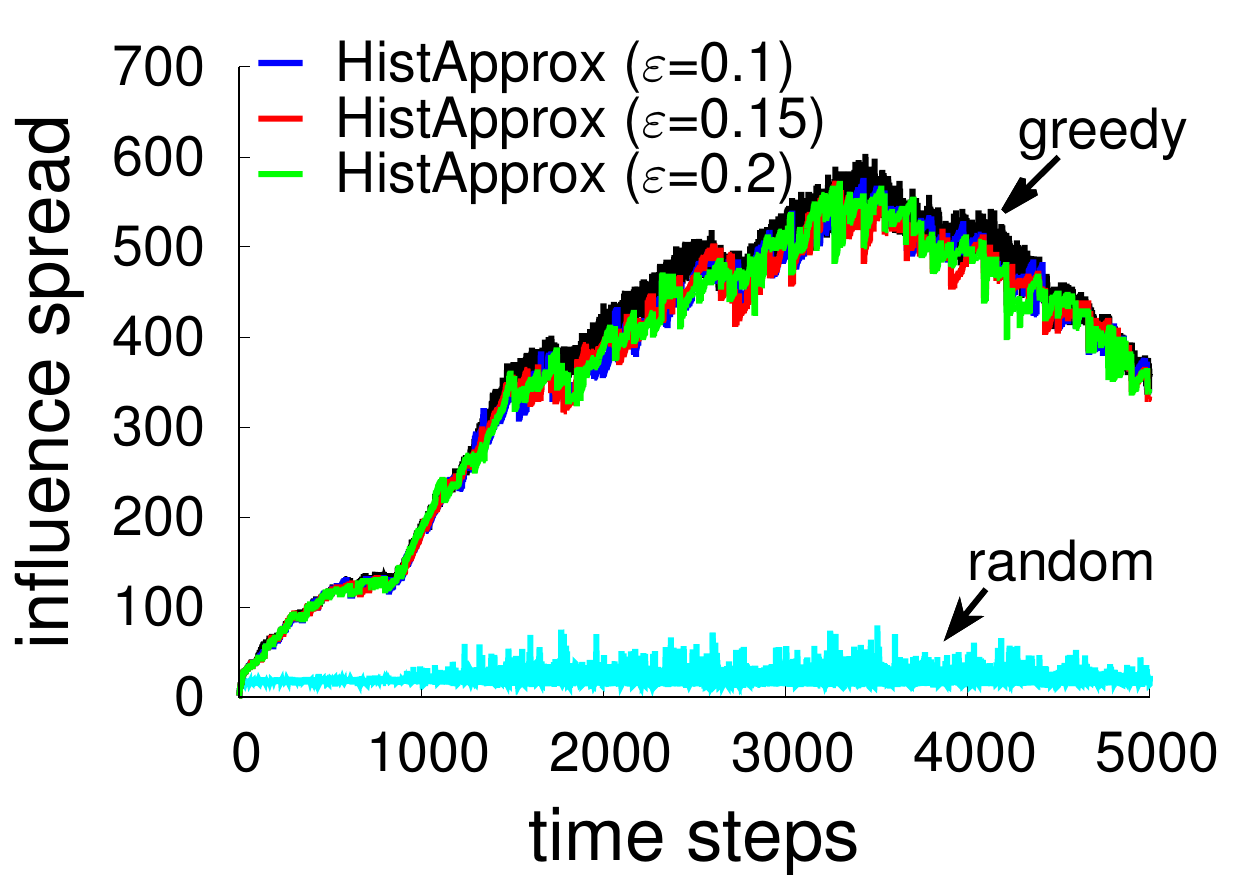}}
\subfloat[Twitter-HK]{%
\includegraphics[width=.5\linewidth]{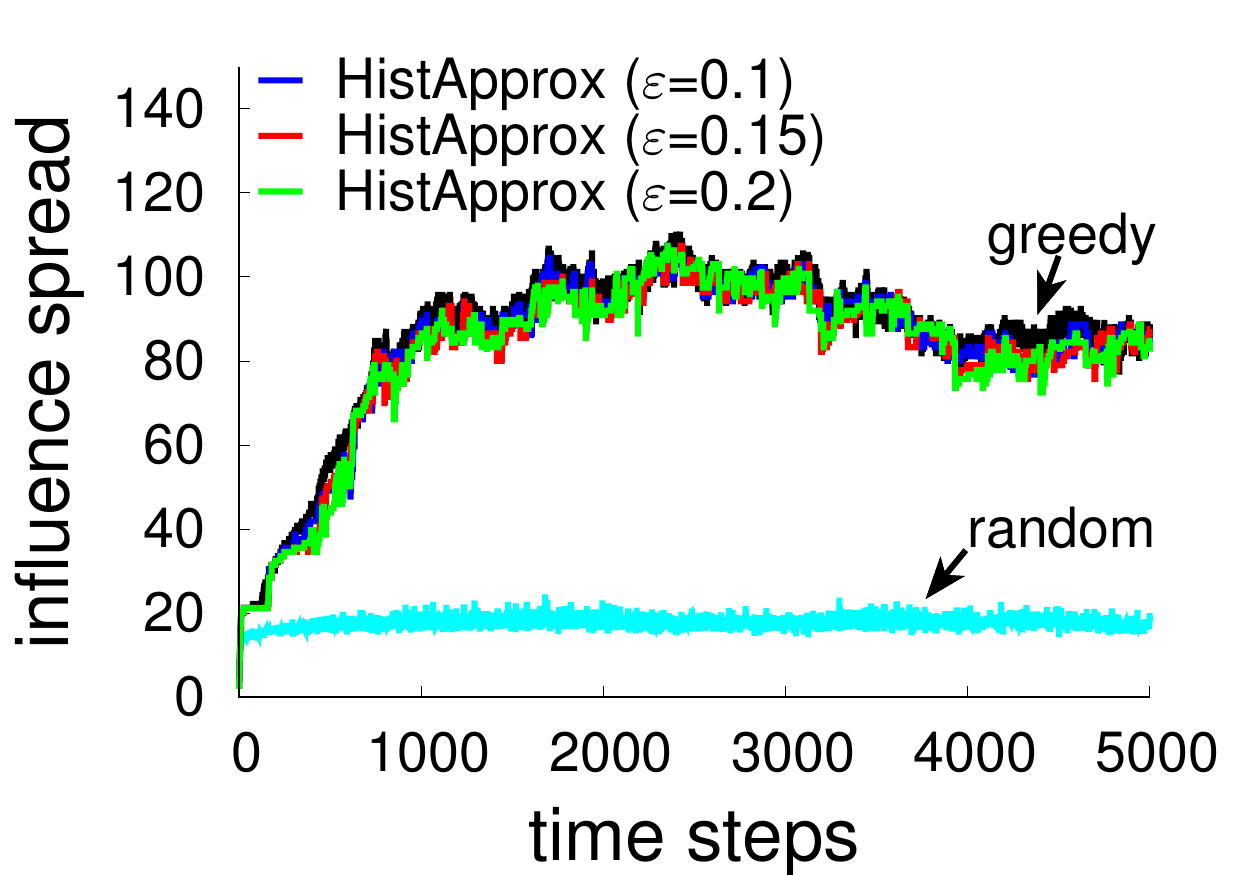}}\\
\subfloat[StackOverflow-c2q]{%
\includegraphics[width=.5\linewidth]{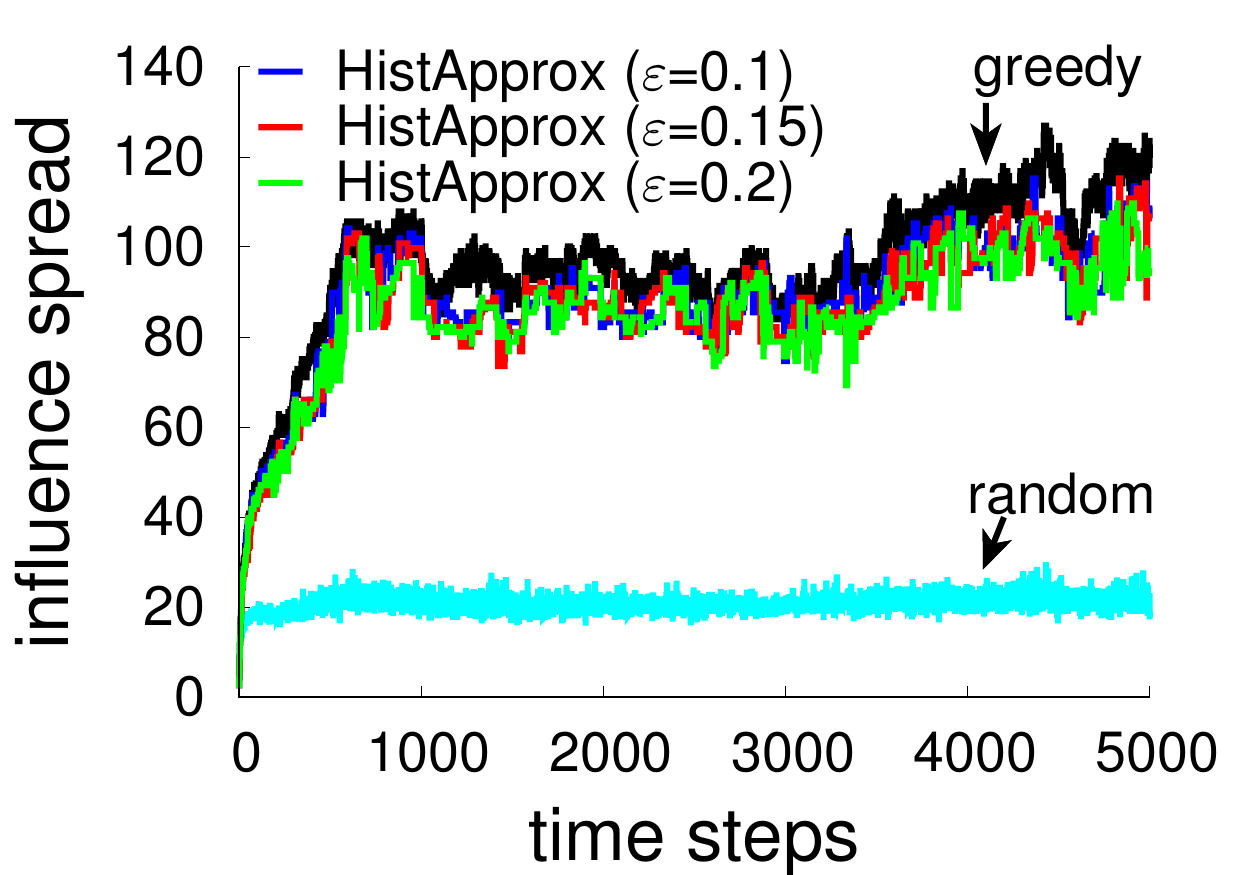}}
\subfloat[StackOverflow-c2a]{%
\includegraphics[width=.5\linewidth]{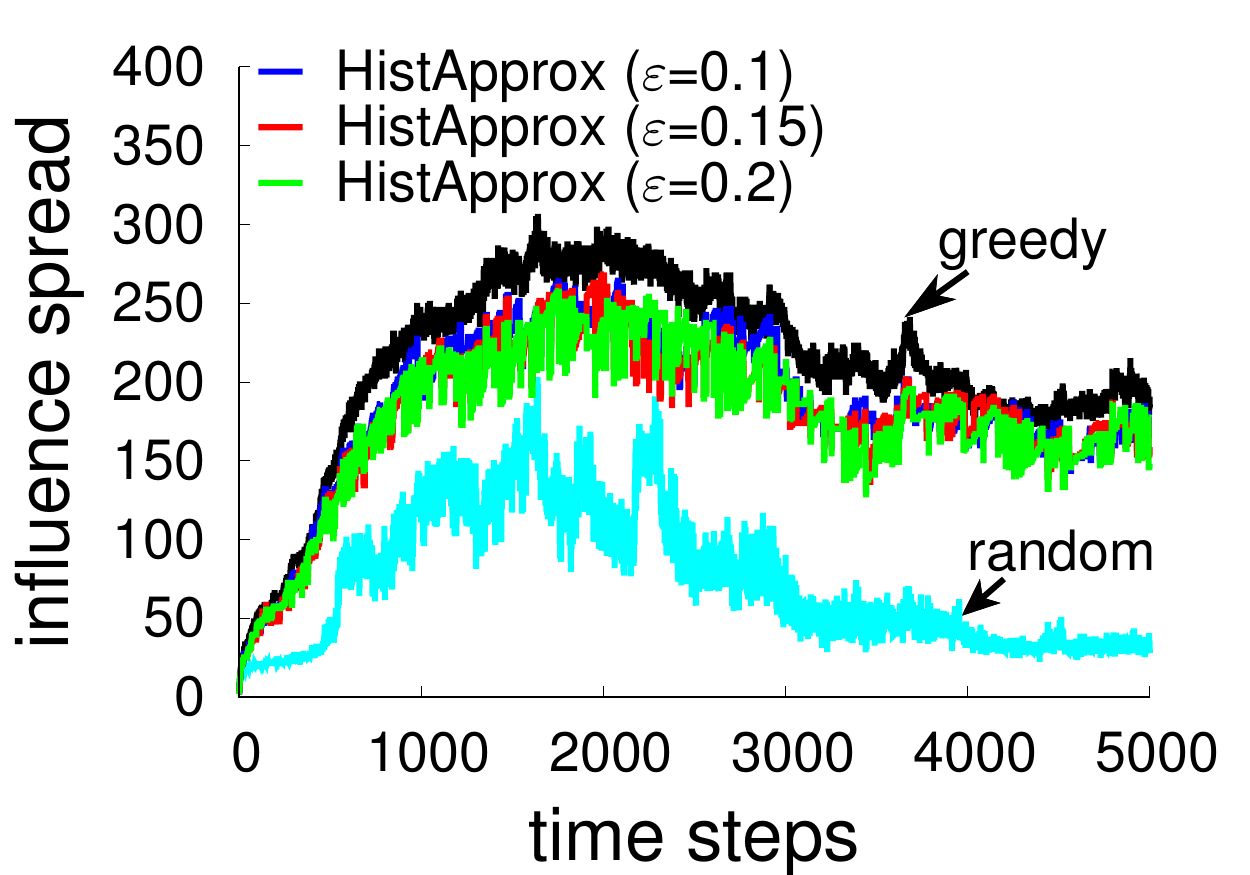}}
\caption{Solution value over time (higher is better).
$k=10, L=10K$}
\label{fig:quality}
\end{figure}

We observe that on all of the six datasets, Greedy achieves the highest solution
value, and Random achieves the lowest.
The solution value of {\sc HistApprox} is very close to Greedy, and is much better
than Random.
To clearly judge $\epsilon$'s effect on solution quality, we calculate the ratio
of solution value of {\sc HistApprox} to Greedy, and average the ratios along
time, and show the results in Fig.~\ref{fig:avg_value}.
It is clear to see that when $\epsilon$ increases, the solution value decreases in
general.

\begin{figure}[htp]
\centering
\includegraphics[width=.65\linewidth]{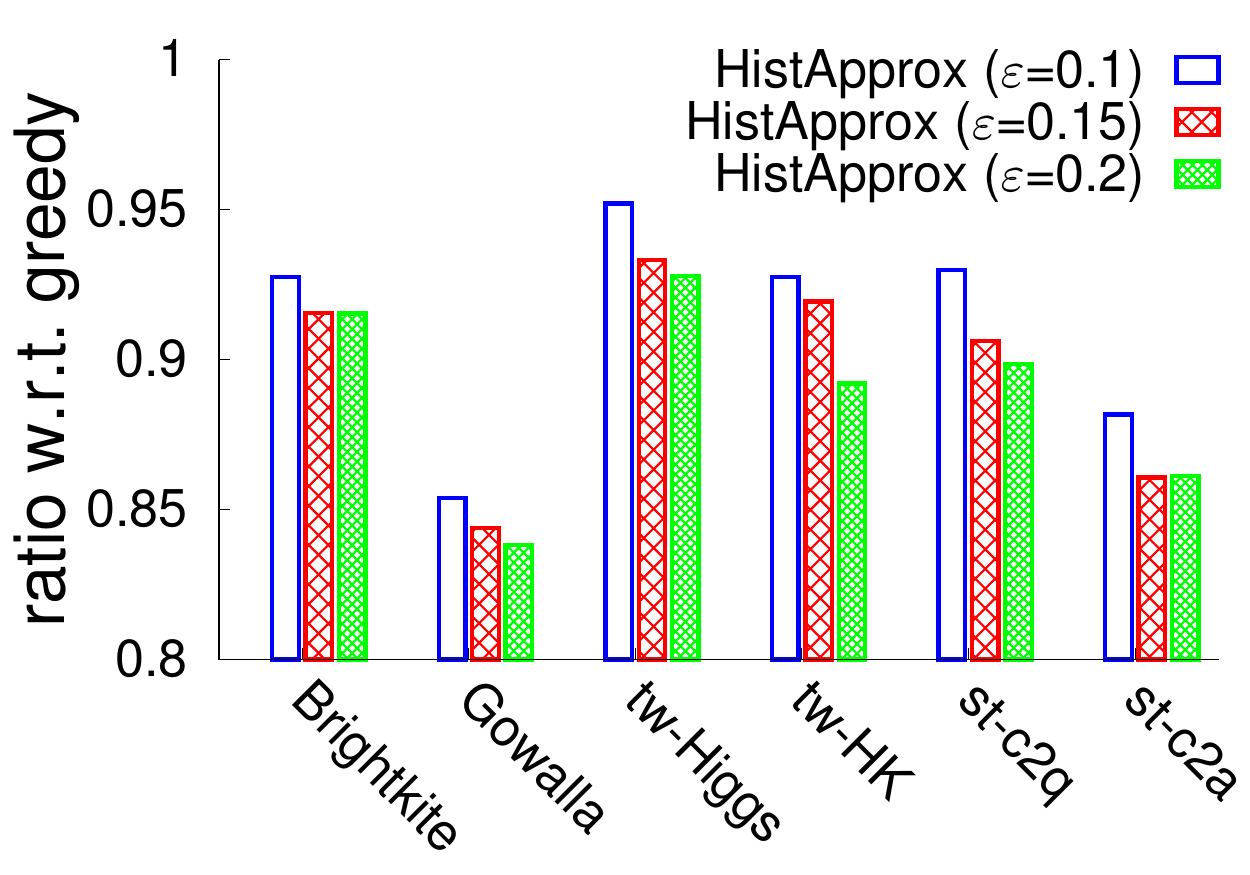}
\caption{Ratio of solution value averaged along time (higher is better)}
\label{fig:avg_value}
\end{figure}

These experiments demonstrate that {\sc HistApprox} achieves comparable solution
quality with Greedy, and smaller $\epsilon$ helps improving {\sc HistApprox}'s
solution quality.

\header{Evaluating Computational Efficiency of \textsc{HistApprox} Over Time.}
Next we compare {\sc HistApprox}'s computational efficiency with Greedy.
When calculating the number of oracle calls of Greedy, we use the lazy
evaluation~\cite{Minoux1978} trick, which is a useful heuristic to reduce Greedy's
number of oracle calls.

We run {\sc HistApprox} with $\epsilon=0.1, 0.15$ and $0.2$ on the six datasets,
respectively, and calculate the ratio of cumulative number of oracle calls at each
time step of {\sc HistApprox} to Greedy.
The results are depicted in Fig.~\ref{fig:efficiency}.

\begin{figure}[t]
\centering
\subfloat[Brightkite]{%
\includegraphics[width=.5\linewidth]{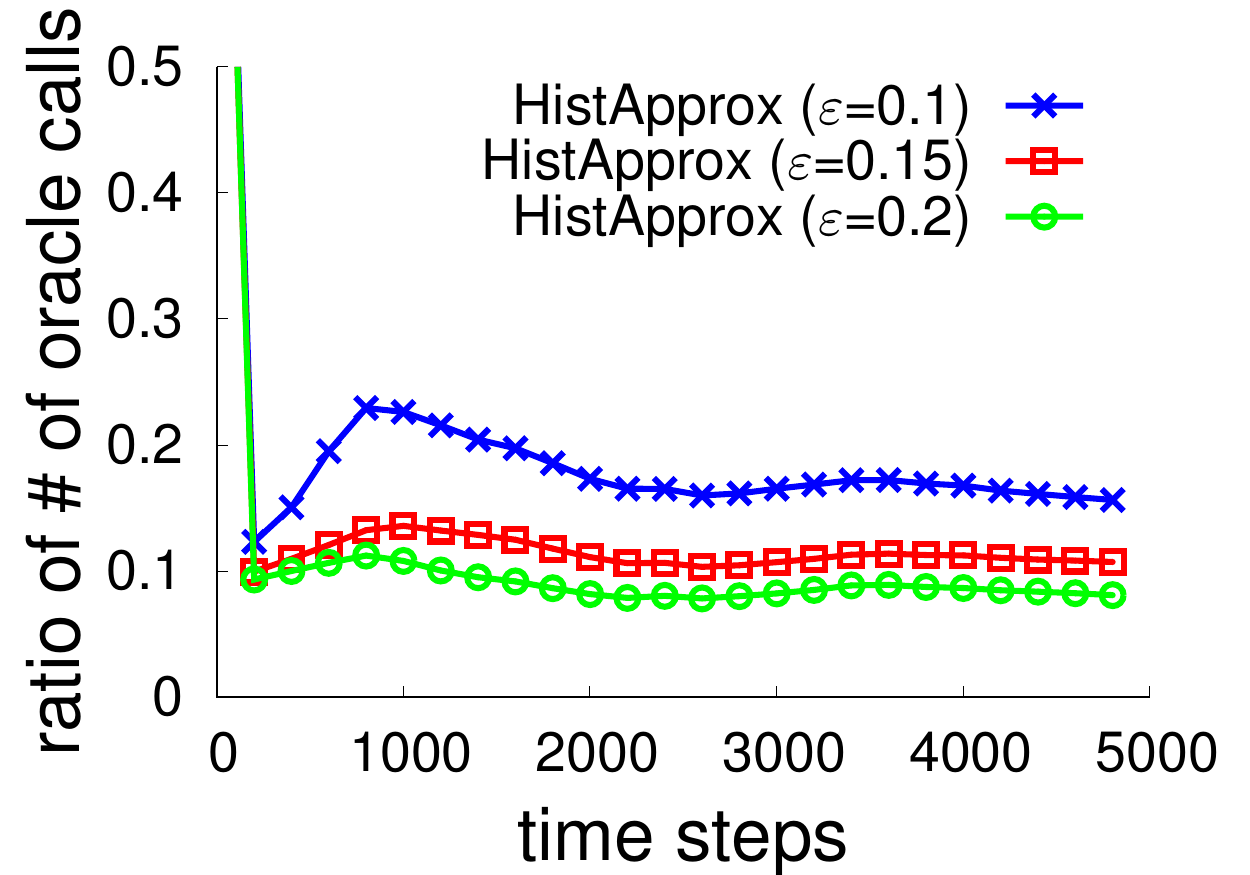}}
\subfloat[Gowalla]{%
\includegraphics[width=.5\linewidth]{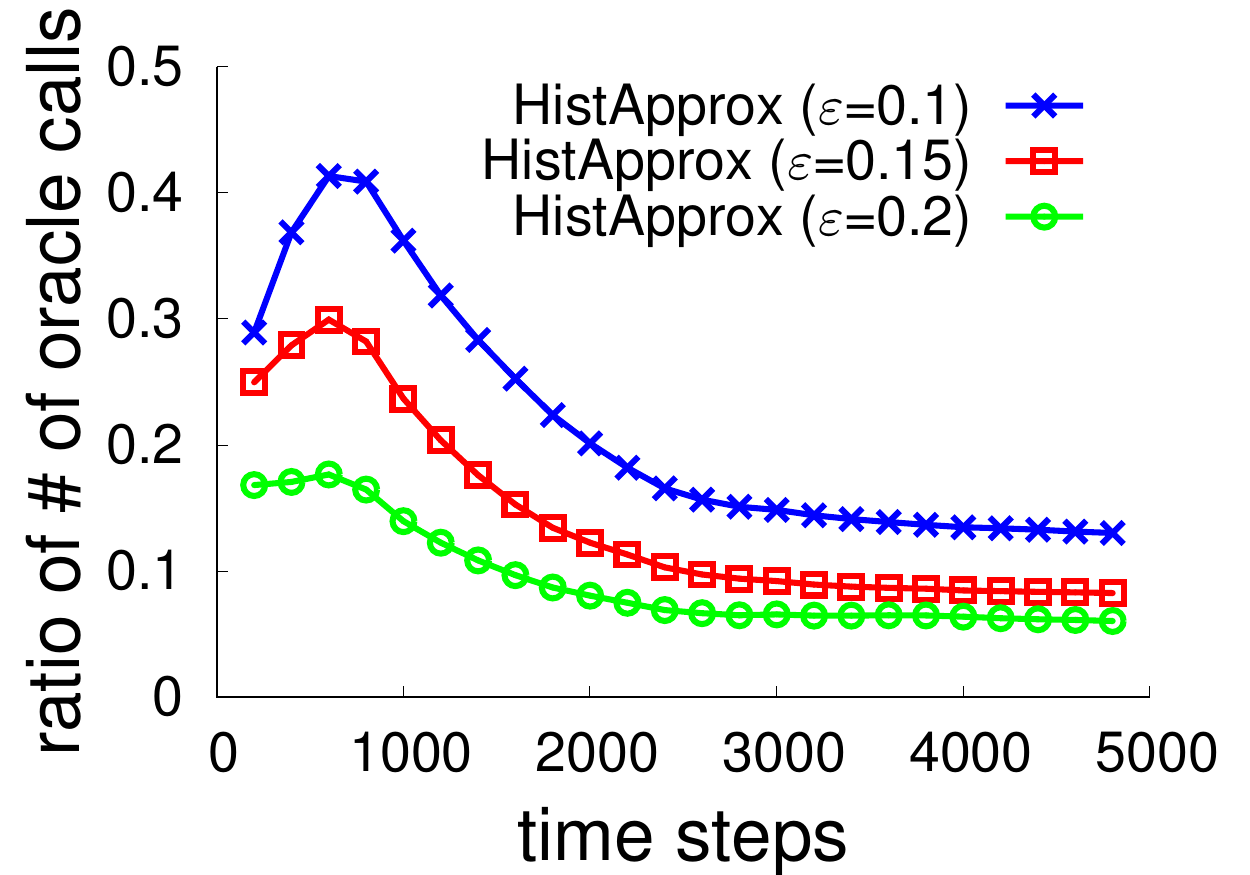}}\\
\subfloat[Twitter-Higgs]{%
\includegraphics[width=.5\linewidth]{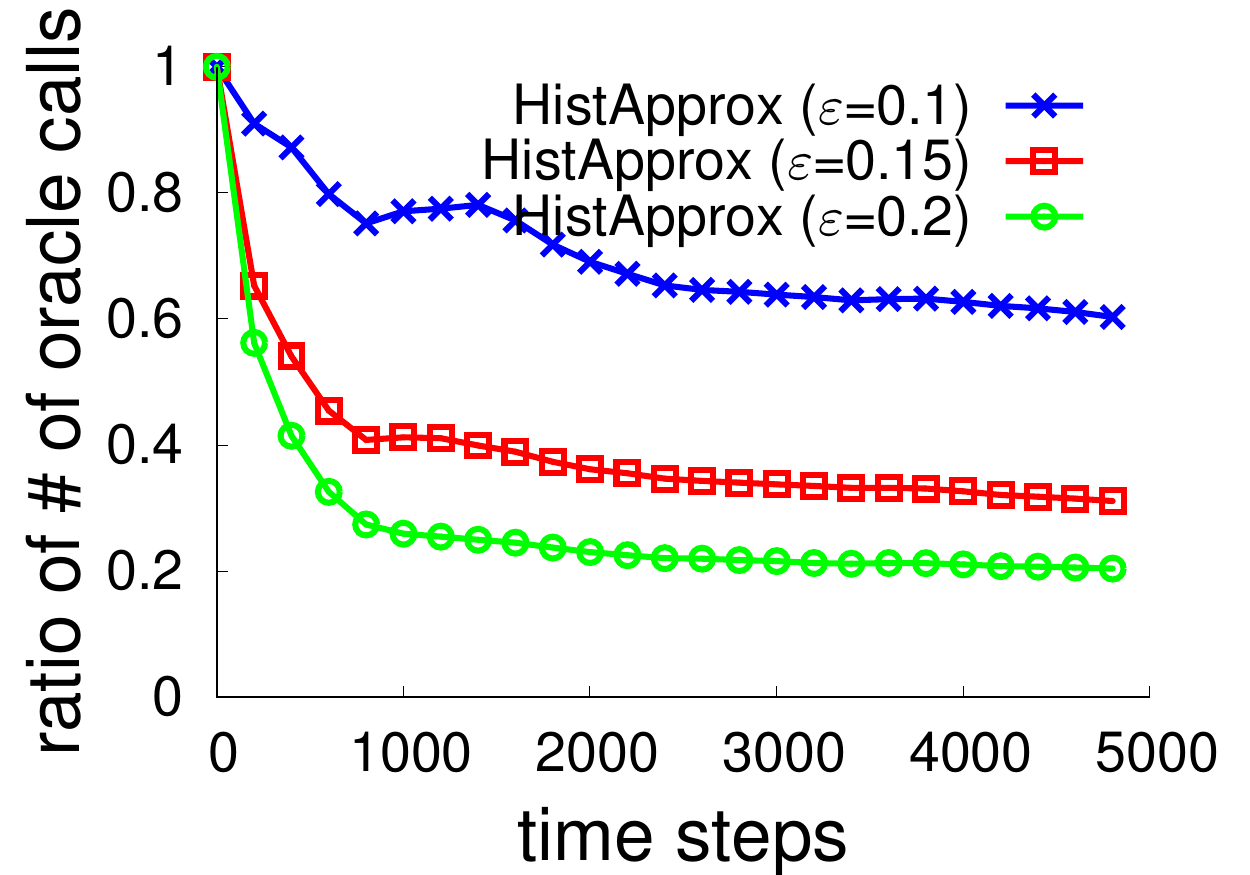}}
\subfloat[Twitter-HK]{%
\includegraphics[width=.5\linewidth]{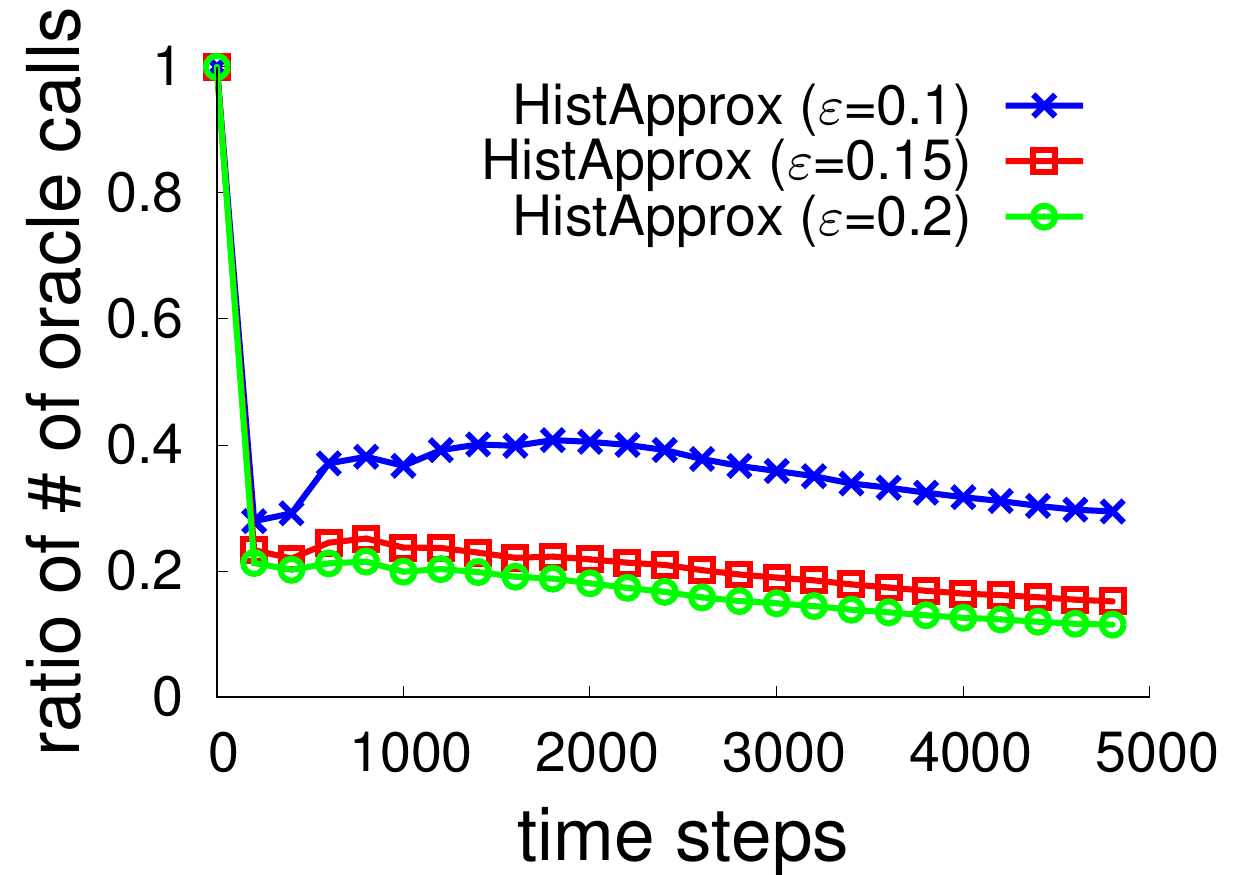}}\\
\subfloat[StackOverflow-c2q]{%
\includegraphics[width=.5\linewidth]{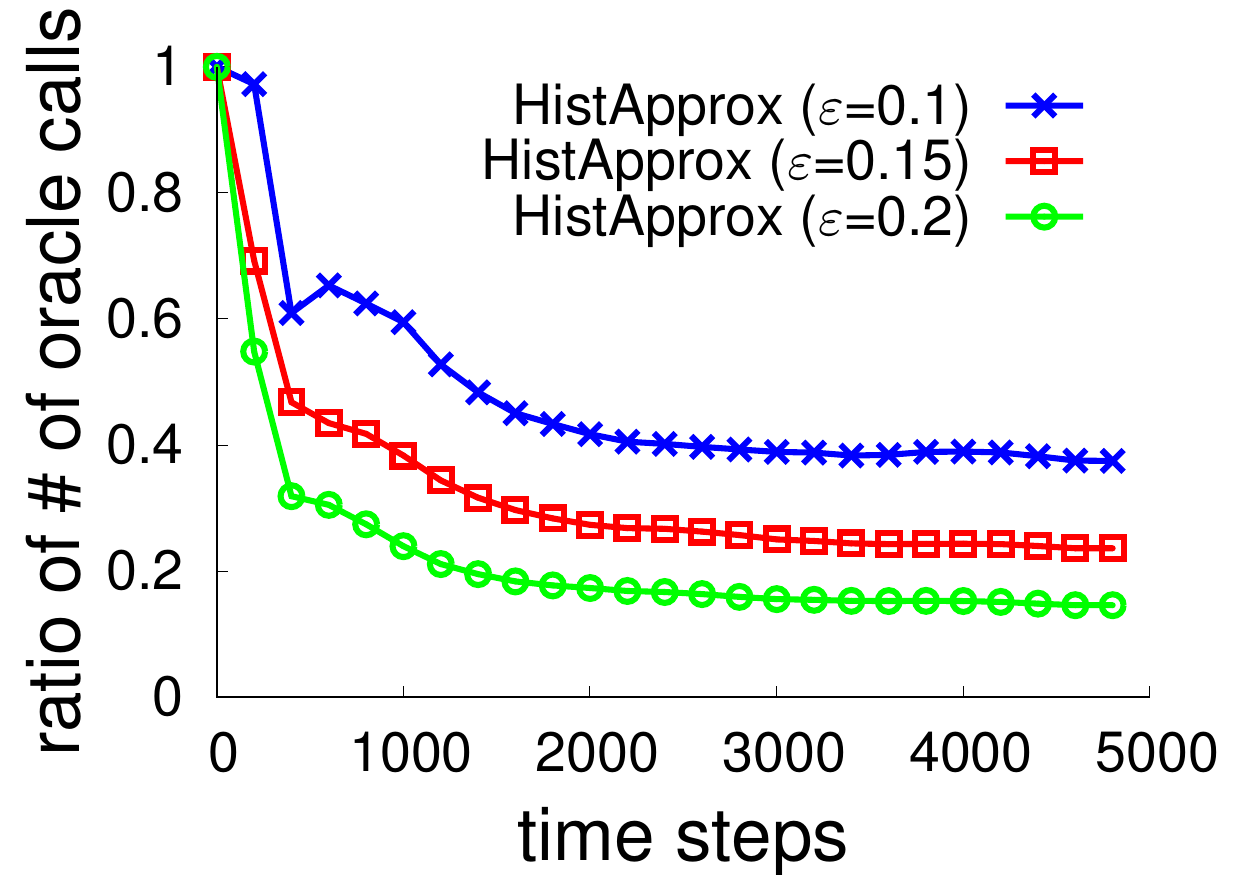}}
\subfloat[StackOverflow-c2a]{%
\includegraphics[width=.5\linewidth]{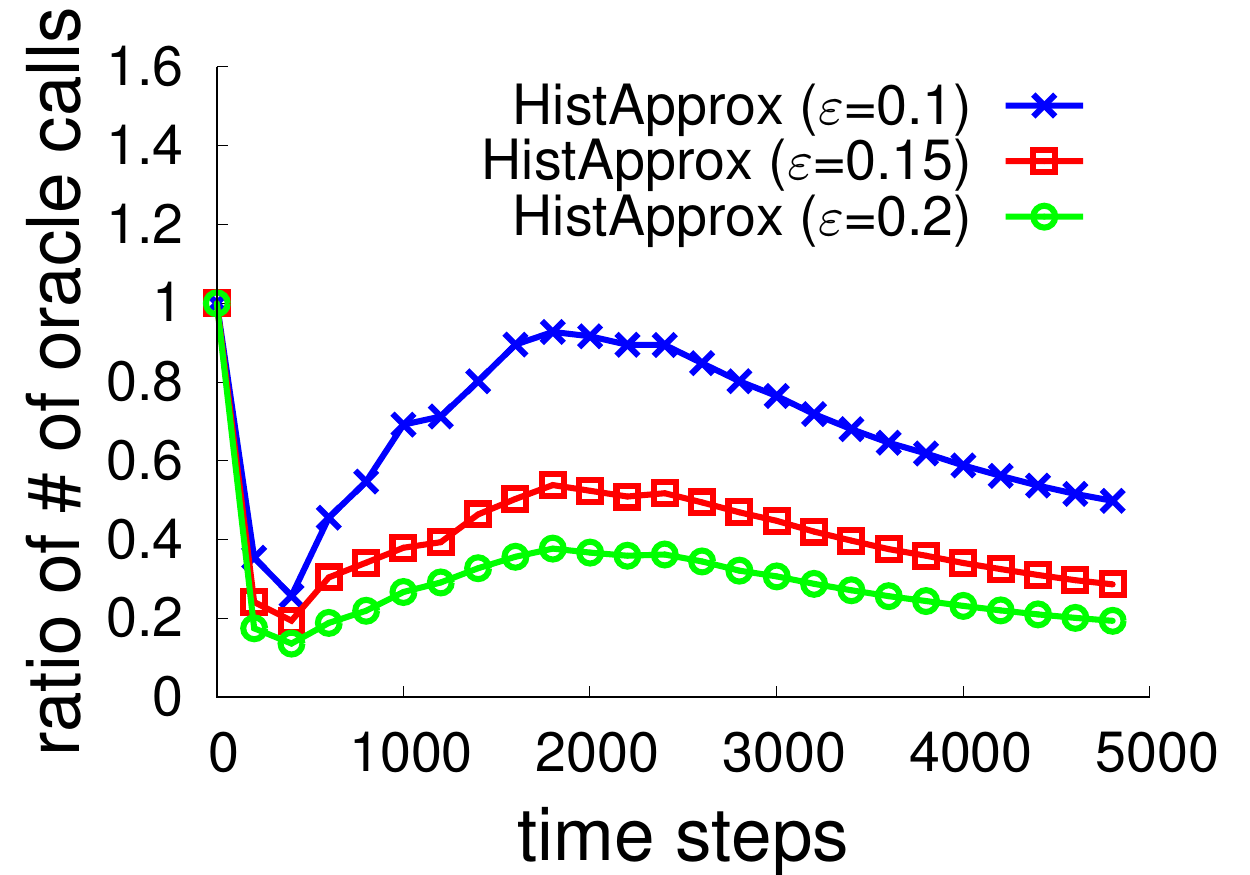}}
\caption{Number of oracle calls ratio (lower is better).
$k=10, L=10K$}
\label{fig:efficiency}
\end{figure}

On all of the six datasets, {\sc HistApprox} uses much less oracle calls than
Greedy.
When $\epsilon$ increases, the number of oracle calls of {\sc HistApprox}
decreases further.
For example, with $\epsilon=0.2$, {\sc HistApprox} uses $5$ to $15$ times less
oracle calls than Greedy.

This experiment demonstrates that {\sc HistApprox} is more efficient than Greedy.
Combining with the previous results, we conclude that $\epsilon$ can trade off
between solution quality and computational efficiency: larger $\epsilon$ makes
{\sc HistApprox} run faster but also decreases solution quality.

\header{Effects of Parameter $k$ and $L$.}
We evaluate {\sc HistApprox}'s performance using different budgets
$k=10,\ldots,100$ and lifetime upper bounds $L=10K,\ldots,100K$.
For each budget $k$ (and $L$), we run {\sc HistApprox} for $5000$ time steps, and
calculate two ratios: (1) the ratio of {\sc HistApprox}'s solution value to
Greedy's solution value; and (2) the ratio of {\sc HistApprox}'s number of oracle
calls to Greedy's number of oracle calls.
We average these ratios along time, and show the results in
Figs.~\ref{fig:ratio_k}~and~\ref{fig:ratio_l}.

\begin{figure}[tp]
\centering
\subfloat[Brightkite]{%
\includegraphics[width=.5\linewidth]{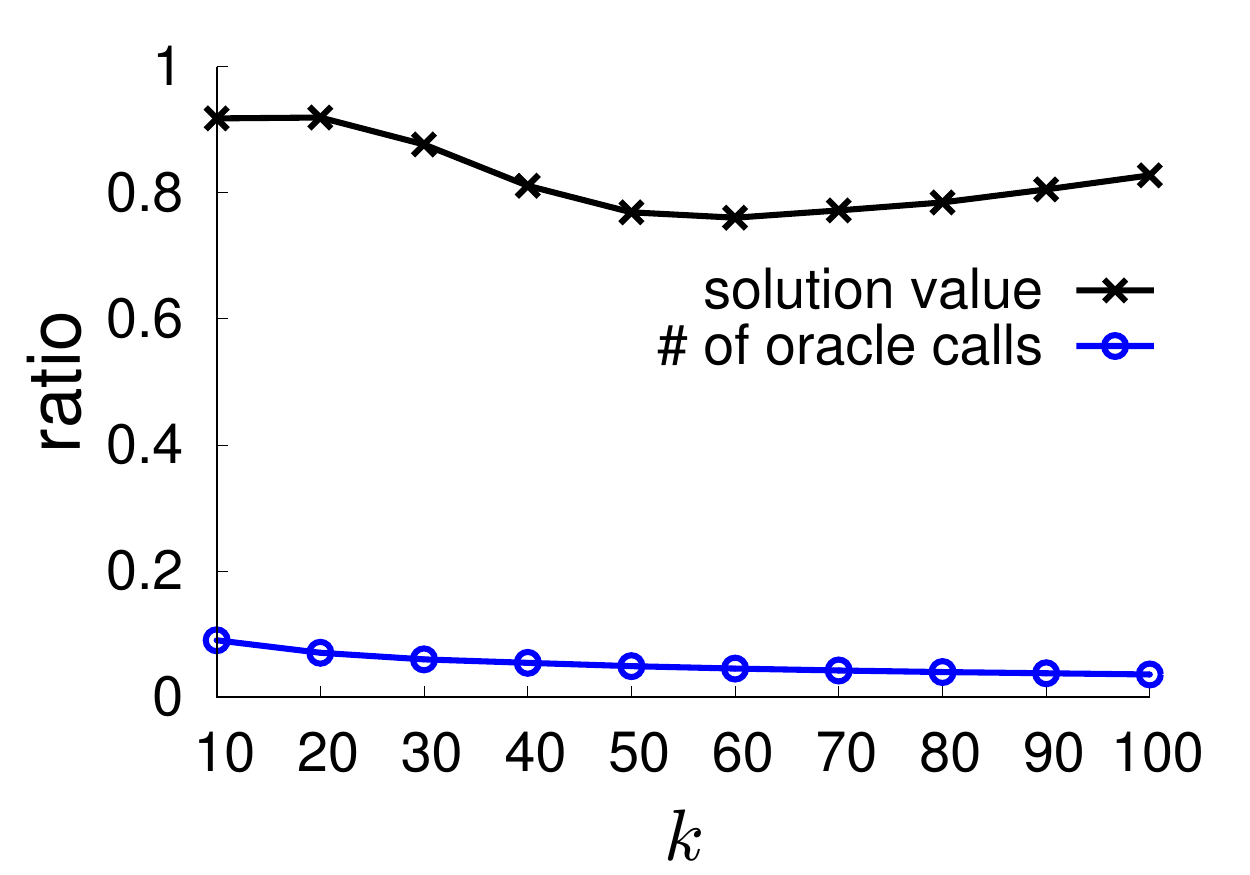}}
\subfloat[Gowalla]{%
\includegraphics[width=.5\linewidth]{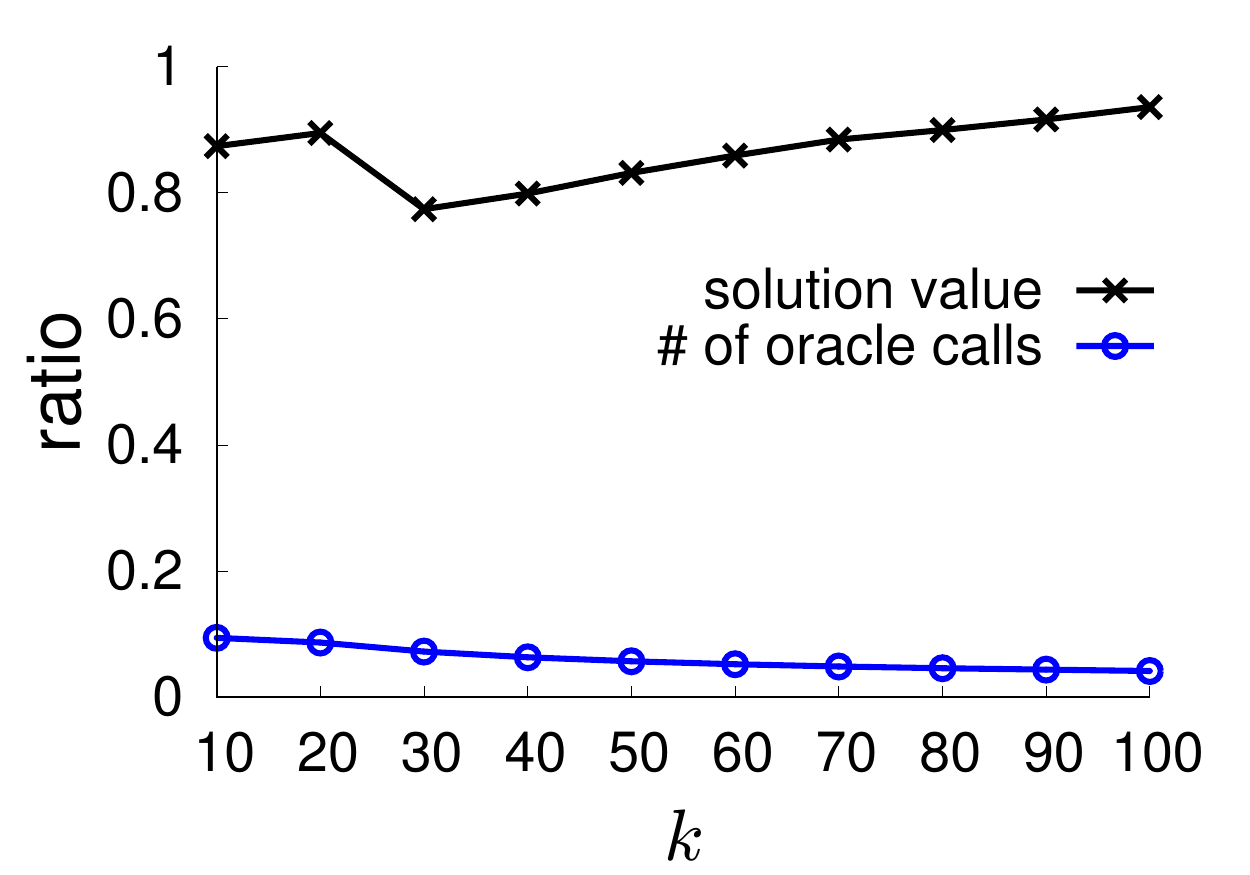}}
\caption{Performance of {\sc HistApprox} w.r.t.~different $k$.
($\epsilon=0.2, L=10K$. Each point is averaged over $5000$ time steps.)}
\label{fig:ratio_k}
\end{figure}

\begin{figure}[tp]
\centering
\subfloat[Brightkite]{%
\includegraphics[width=.5\linewidth]{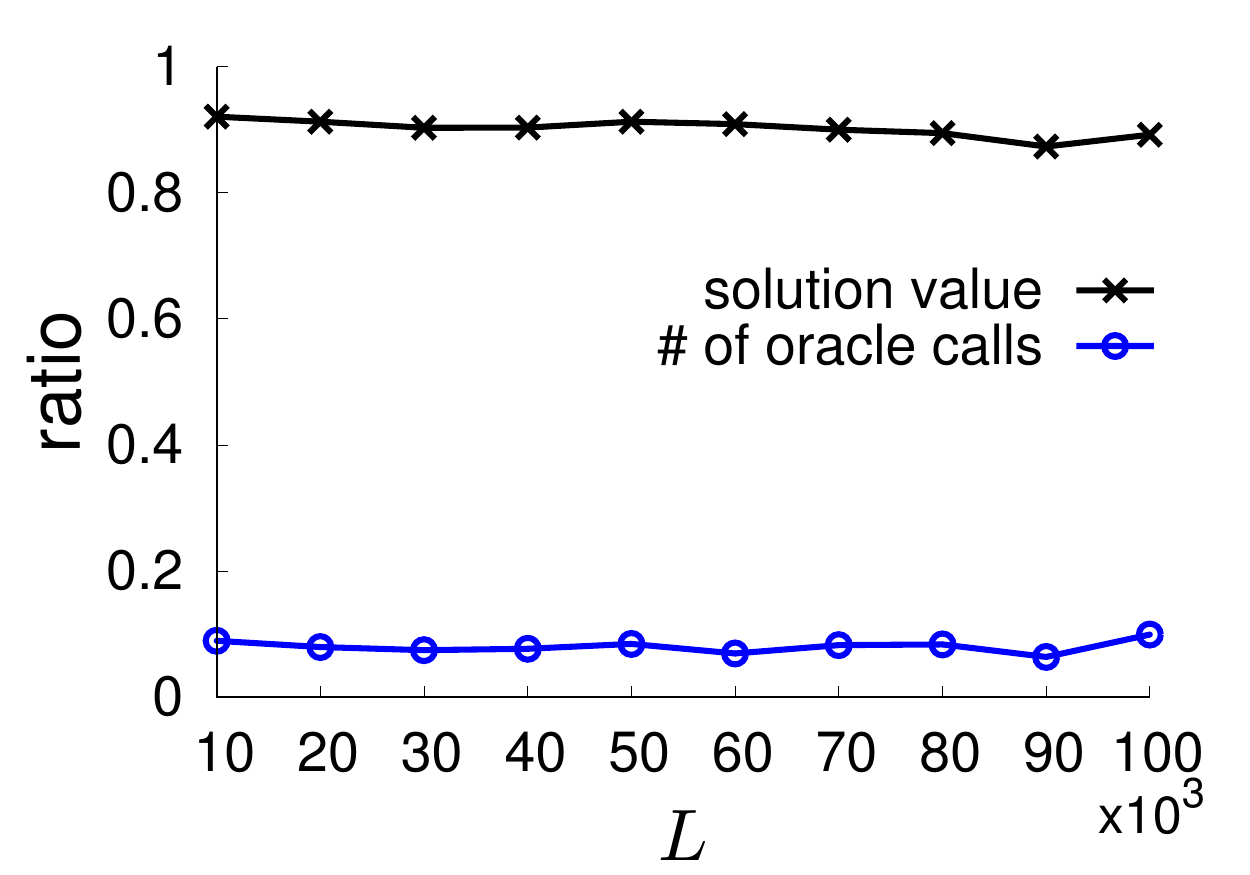}}
\subfloat[Gowalla]{%
\includegraphics[width=.5\linewidth]{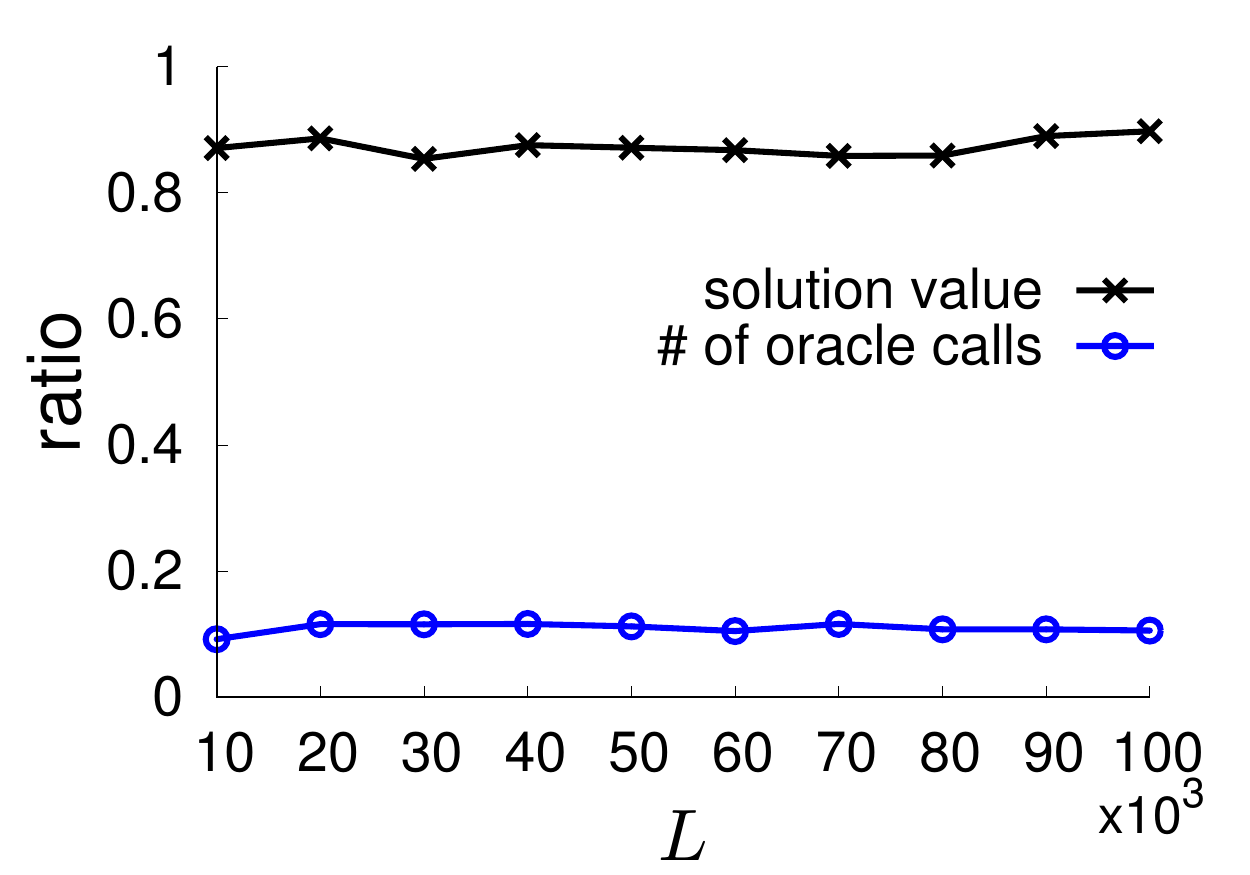}}
\caption{Performance of {\sc HistApprox} w.r.t.~different $L$.
($\epsilon=0.2, k=10$. Each point is averaged over $5000$ time steps.)}
\label{fig:ratio_l}
\end{figure}

We find that under different parameter settings, {\sc HistApprox}'s performance
consists with our previous observations.
In addition, we find that, for larger $k$, the efficiency of {\sc HistApprox}
improves more significantly than Greedy.
This is due to the fact that {\sc HistApprox}'s time complexity increases
logarithmically with $k$; while Greedy's time complexity increases linearly with
$k$.
We also find that $L$ does not affect {\sc HistApprox}'s performance very much.

\header{Performance and Throughput Comparisons.}
Next, we compare the solution quality and throughput with the other baseline
methods.
Since Greedy achieves the highest solution quality among all the methods, we use
Greedy as a reference and show solution value ratio w.r.t.~Greedy when evaluating
solution quality.
For throughput analysis, we will show the maximum stream processing speed (i.e.,
number of processed edges per second) using different methods.

We set budgets $k=10,\ldots,100$.
Lifetimes are sampled from $\textrm{Geo}(0.001)$ with upper bounds
$L=10K,\ldots,100K$ respectively.
We set $\epsilon=0.3$ in \textsc{HistApprox}.
All of the algorithms are ran for $10,000$ time steps.
For throughput analysis, we fix the budget to be $k=10$.
The results are depicted in Figs.~\ref{fig:val_cmp} and~\ref{fig:throughput}.

\begin{figure}[htp]
\centering
\subfloat[Twitter-Higgs]{%
\includegraphics[width=.5\linewidth]{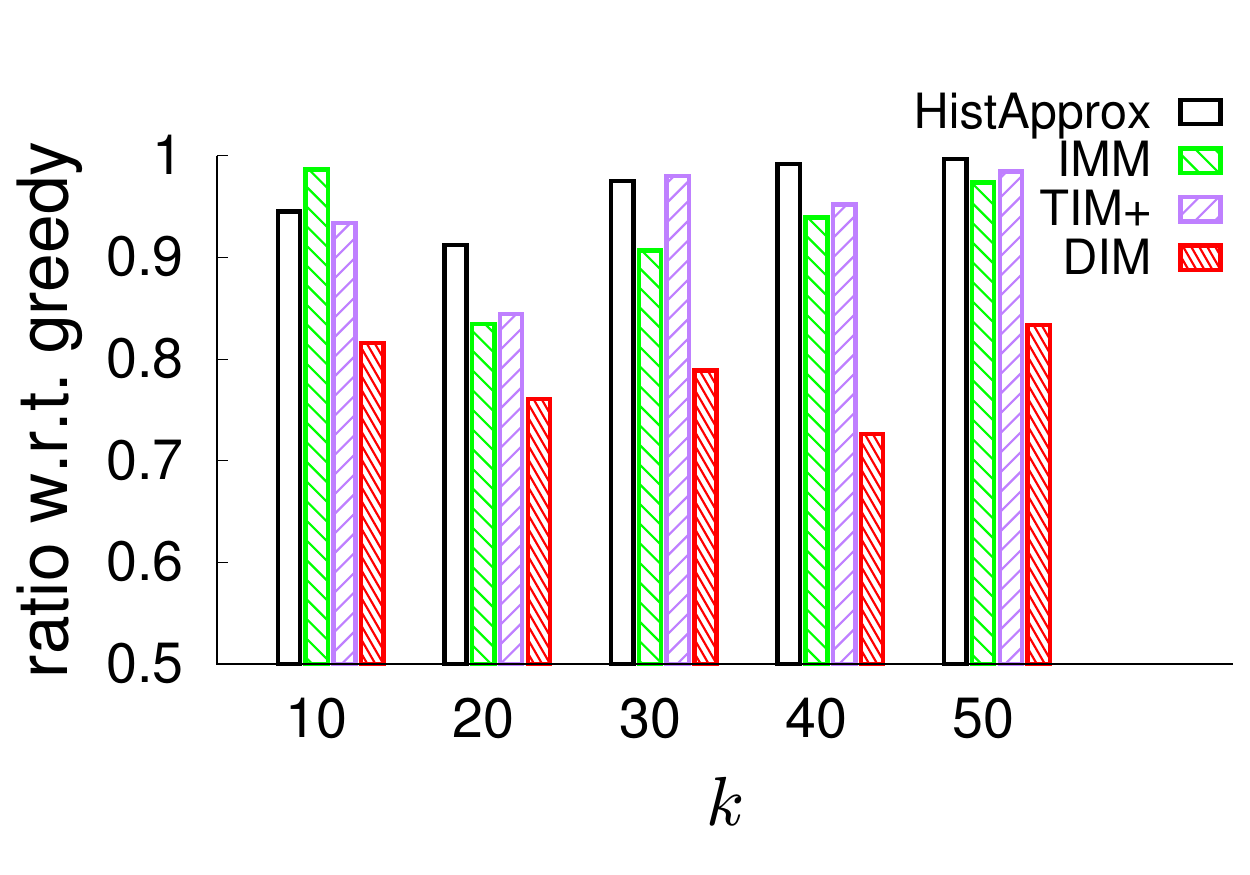}}
\subfloat[Twitter-Higgs]{%
\includegraphics[width=.5\linewidth]{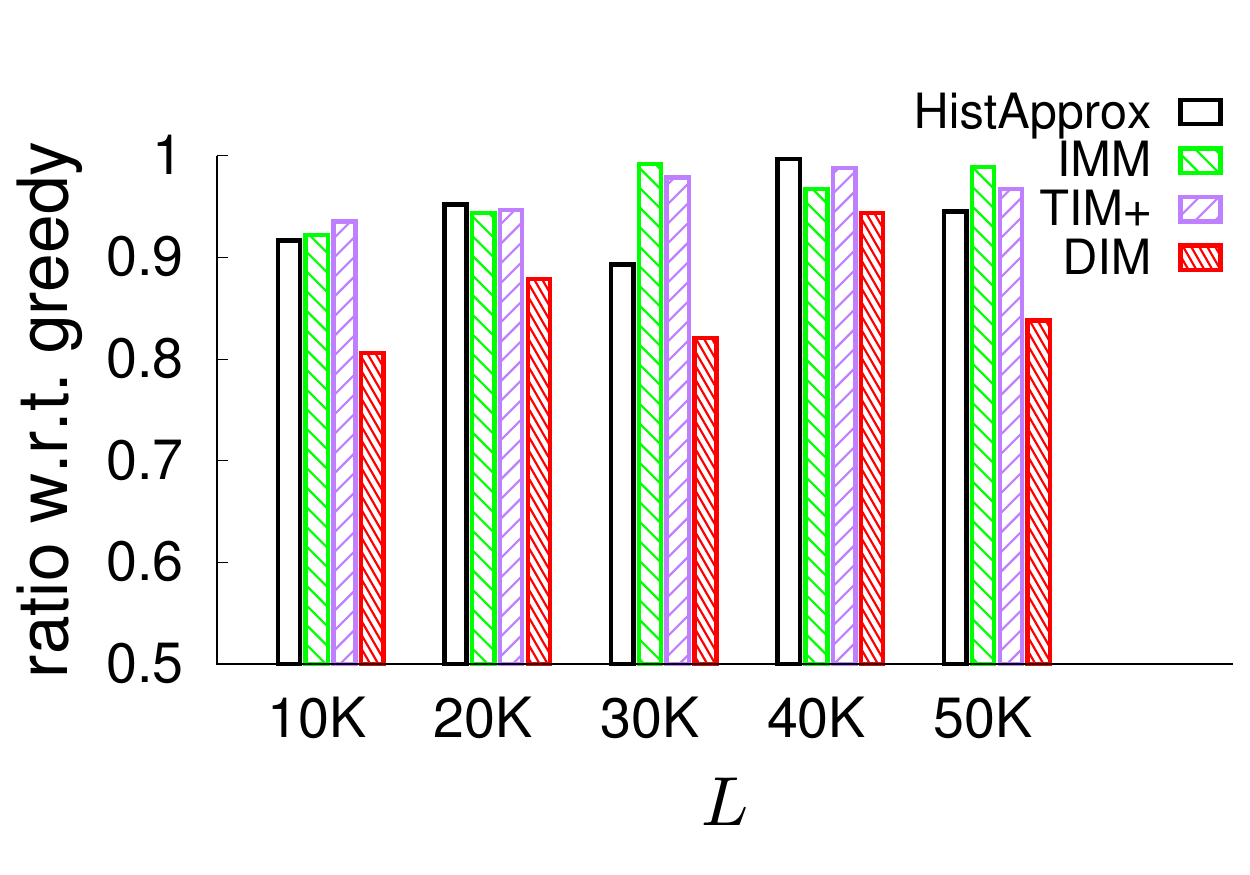}} \\
\subfloat[StackOverflow-c2q]{%
\includegraphics[width=.5\linewidth]{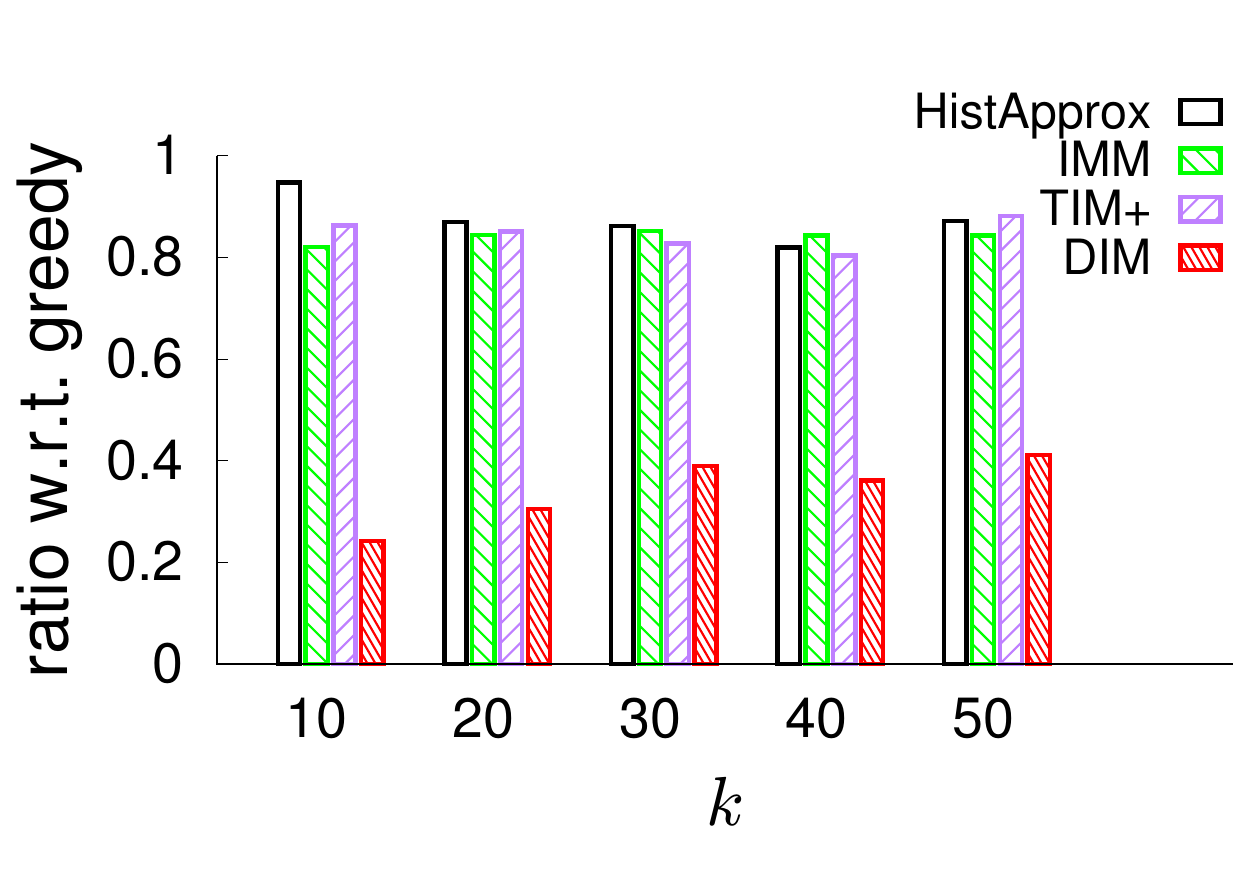}}
\subfloat[StackOverflow-c2q]{%
\includegraphics[width=.5\linewidth]{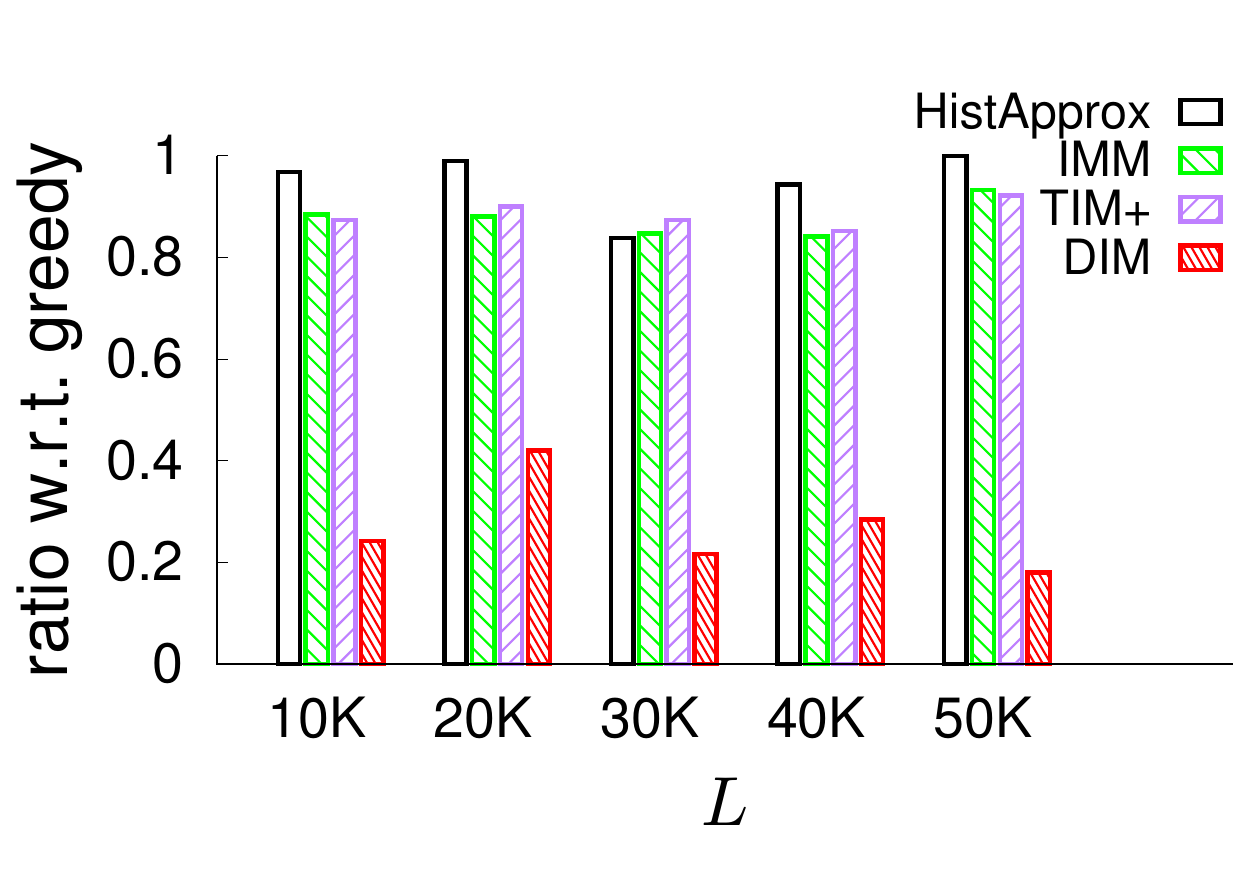}}
\caption{Solution quality comparison (higher is better)}
\label{fig:val_cmp}
\end{figure}

\begin{figure}[htp]
\centering
\subfloat[Twitter-Higgs]{%
\includegraphics[width=.5\linewidth]{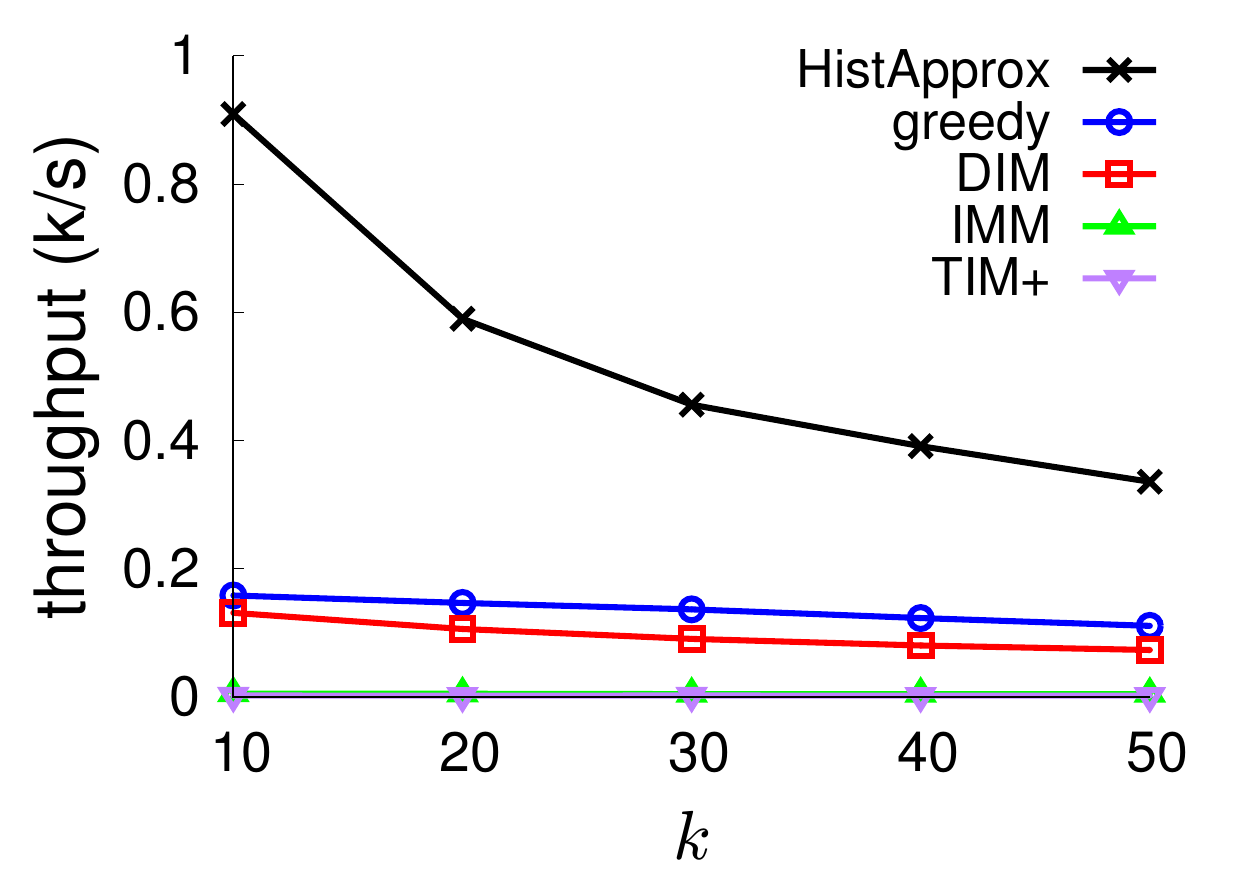}}
\subfloat[Twitter-Higgs]{%
\includegraphics[width=.5\linewidth]{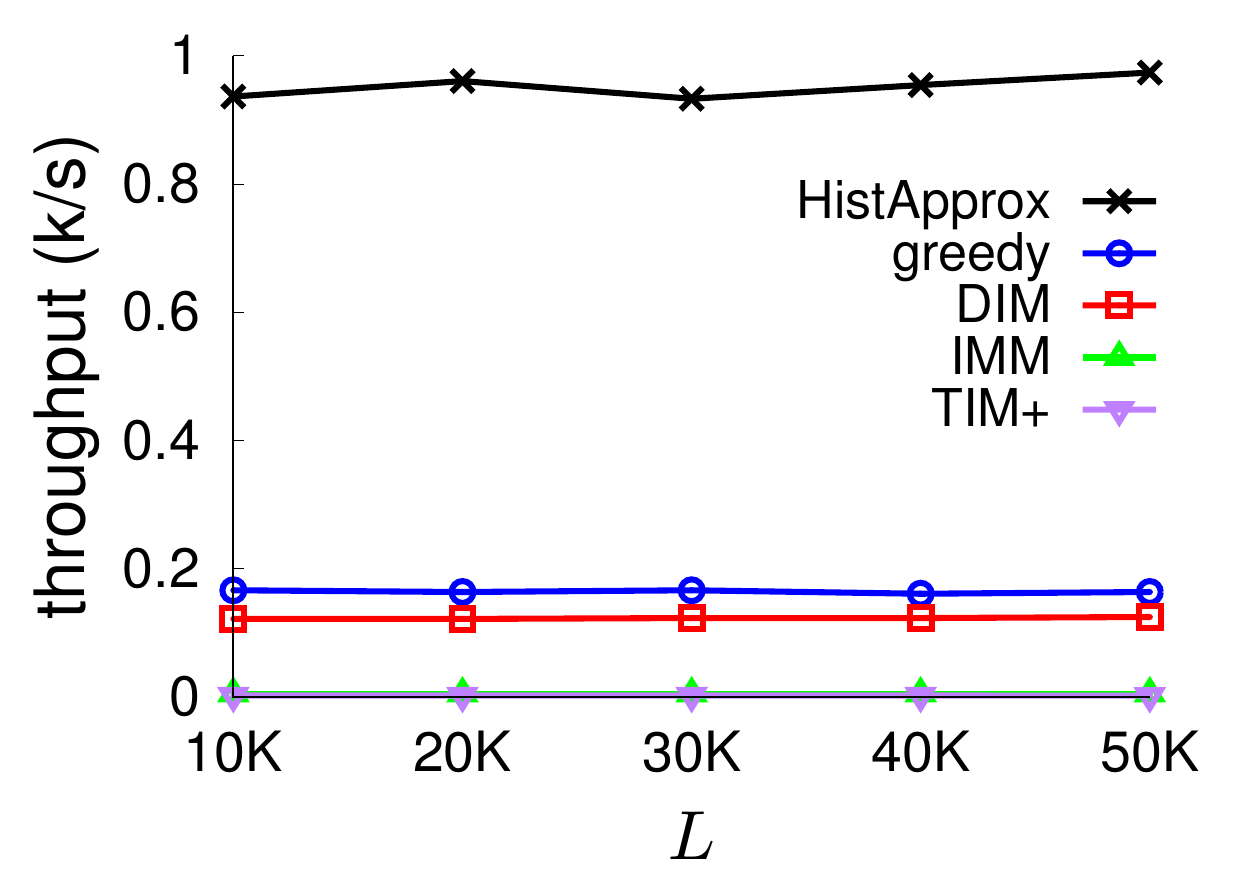}} \\
\subfloat[StackOverflow-c2q]{%
\includegraphics[width=.5\linewidth]{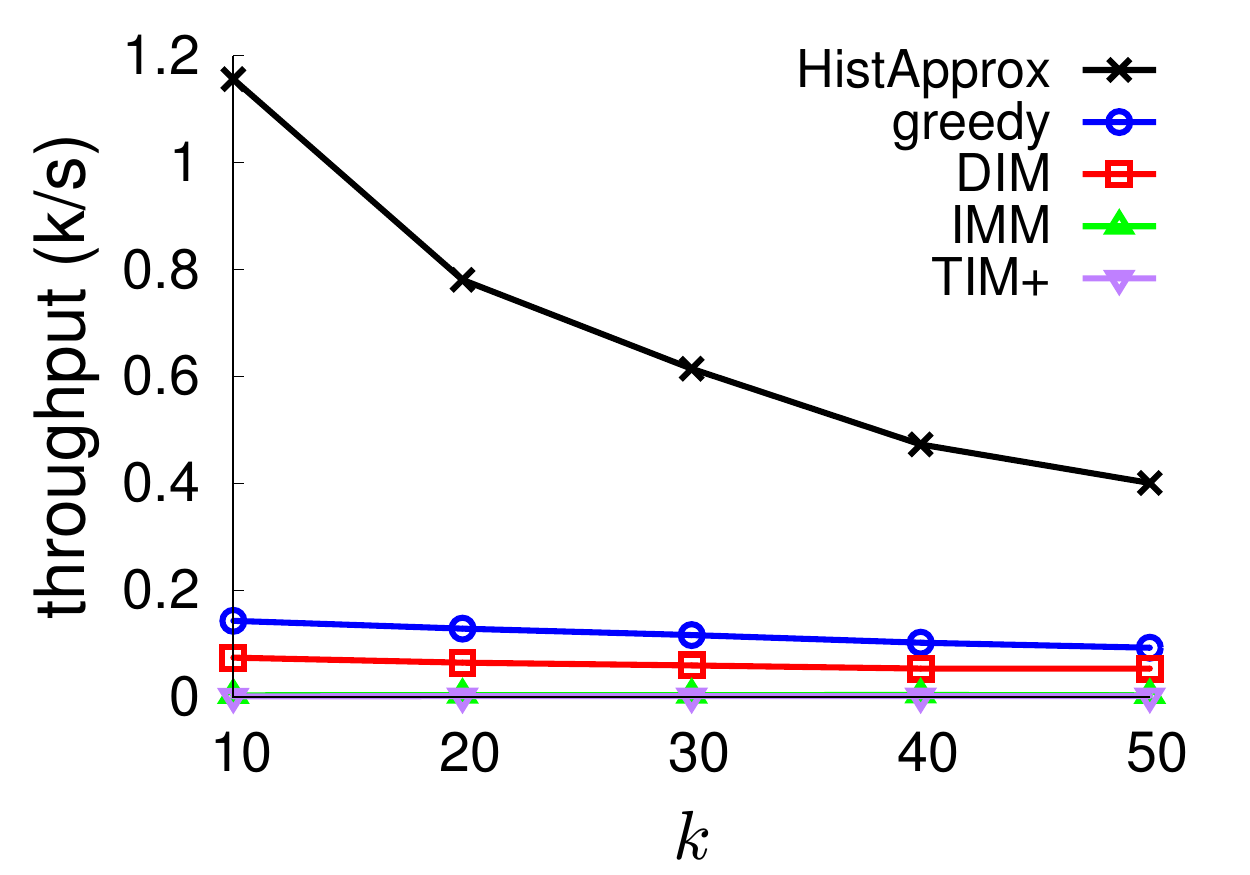}}
\subfloat[StackOverflow-c2q]{%
\includegraphics[width=.5\linewidth]{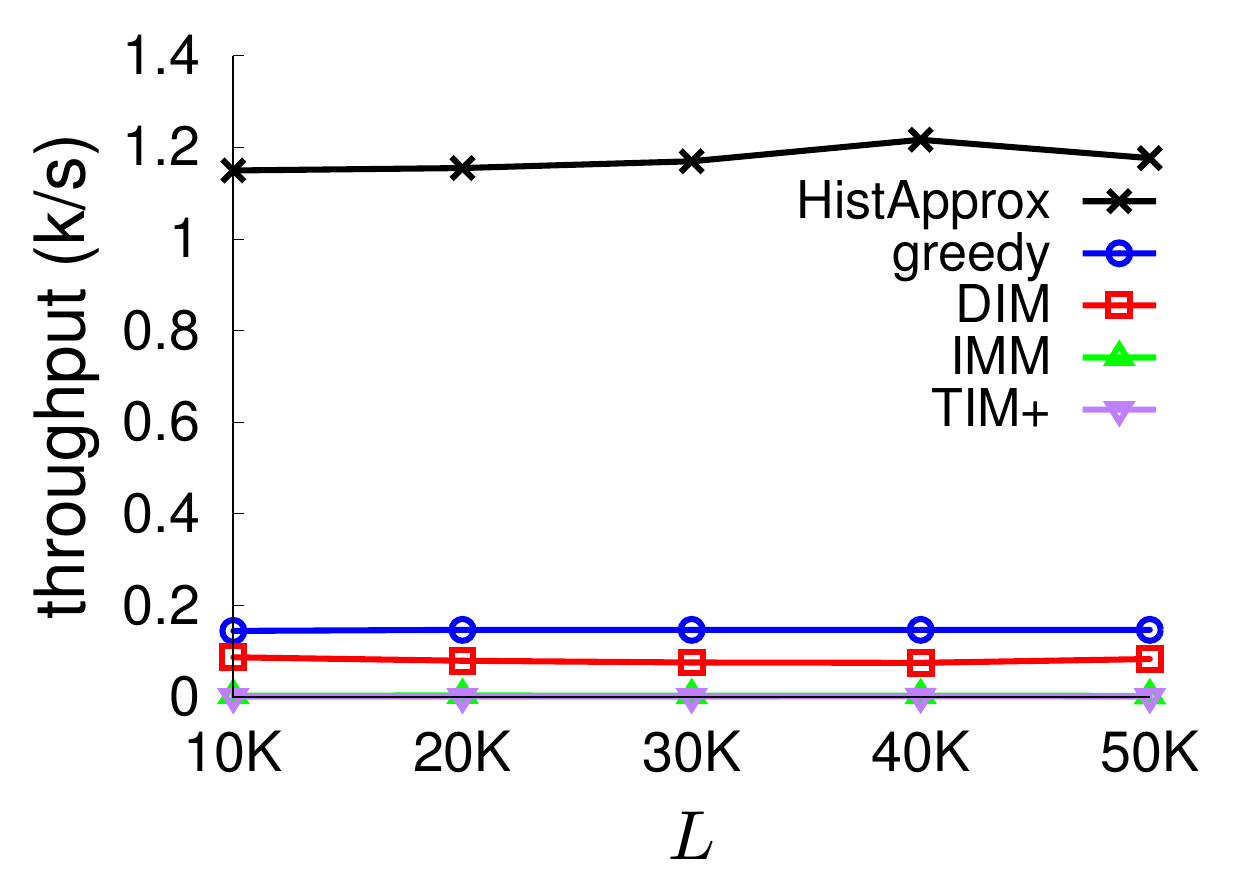}}
\caption{Throughput comparison (higher is better)}
\label{fig:throughput}
\end{figure}

From Fig.~\ref{fig:val_cmp}, we observe that \textsc{HistApprox}, IMM and TIM+
always find high quality solutions.
In contrast, DIM seems not so stable, and it performs even worse on the
StackOverflow-c2q dataset than on the Twitter-Higgs dataset.
From Fig.~\ref{fig:throughput}, we observe that \textsc{HistApprox} achieves the
highest throughput than the other methods, then comes Greedy and DIM, the two
methods IMM and TIM+ designed for static graphs have the lowest throughput (and
they almost overlap with each other).
Greedy using the lazy evaluation trick is faster than DIM, IMM, and TIM+ because
the oracle call implementation in Greedy is much faster than oracle call
implementations in the other three methods, which relay on Monte-Carlo
simulations.

This experiment demonstrates that our proposed \textsc{HistApprox} algorithm can
find high quality solutions with high throughput.

\section{\textbf{Related Work}}
\label{sec:related_work}

\header{Influence Maximization (IM) over Dynamic Networks.}
IM and its variants on static networks have been extensively
studied
(e.g.,~\cite{Kempe2003,Borgs2014,Chen2010,Tang2014,Lucier2015,Tang2015,Chen2016d,Litou2017,Lin2017},
to name a few).
Here we devote to review some literature related to IM on dynamic networks.

Aggarwal et al.~\cite{Aggarwal2012} and Zhuang et al.~\cite{Zhuang2013} propose
heuristic approaches to solve IM on dynamic networks.
These heuristic approaches, however, have no theoretical guarantee on the quality
of selected nodes.

Song et al.~\cite{Song2017} use the interchange greedy
approach~\cite{Nemhauser1978} to solve IM over a sequence of networks with the
assumption that these networks change smoothly.
The interchange greedy approach updates current influential nodes based on
influential nodes found in previous network and avoids constructing the solution
from an empty set.
This approach has an approximation factor $(1/2-\epsilon)$.
However, the interchange greedy approach degrades to running from scratch if these
networks do not change smoothly.
This limits its application on highly dynamic networks.

Both~\cite{Ohsaka2016} and~\cite{Yang2017c} extend the reverse-reachable (RR)
sets~\cite{Borgs2014} to updatable index structures for faster computing the
influence of a node or a set of nodes in dynamic networks.
In particular, \cite{Yang2017c} considers finding influential individuals in
dynamic networks, which is actually a different problem.
It is worth noting that RR-sets are approximation methods designed to speed up the
computation of the influence of a set of nodes under the IC or LT model, which is
\#P-hard; while our approach is a data-driven approach without any assumption of
influence spreading model.

There are also some variants of IM on dynamic networks such as topic-specific
influencers query~\cite{Subbian2013,SubbianEtAl2016} and online learning approach
for IM~\cite{Lei2015}.

\header{Streaming Submodular Optimization (SSO) Methods.}
SSO over insertion-only streams is first studied in~\cite{Kumar2013a} and the
state-of-the-art approach is the threshold based {\sc SieveStreaming}
algorithm~\cite{Badanidiyuru2014a} which has an approximation factor
$(1/2-\epsilon)$.

The sliding-window stream model is first introduced in~\cite{Datar2002} for
maintaining summation aggregates on a data stream (e.g., counting the number of
$1$'s in a $0$-$1$ stream).
Chen et al.~\cite{Chen2016f} leverage the smooth histogram
technique~\cite{Braverman2007} to solve SSO under the sliding-window stream, and
achieves a $(1/4-\epsilon)$ approximation factor.
Epasto et al.~\cite{Epasto2017} further improve the approximation factor to
$(1/3-\epsilon)$.

Note that our problem requires solving SSO over time-decaying streams, which is
more general than sliding-window streams.
To the best of our knowledge, there is no previous art that can solve SSO on
time-decaying streams.

Recently, \cite{Wang2017} leveraged the results of~\cite{Chen2016f} to develop a
streaming method to solve IM over sliding-window streams.
The best approximation factor of their framework is $(1/4-\epsilon)$.
In our work, we consider the more general time-decaying stream, and our approach
can achieve the $(1/2-\epsilon)$ approximation factor.

\header{Maintaining Time-Decaying Stream Aggregates.}
Cohen et al.~\cite{Cohen2006}~first extend the sliding-window model
in~\cite{Datar2002} to the general time-decaying model for the purpose of
approximating summation aggregates in data streams.
Cormode et al.~\cite{Cormode2009}~consider the similar problem by designing
time-decaying sketches.
These studies have inspired us to consider the more general time-decaying streams.

\section{\textbf{Conclusion}}
\label{sec:conclusion}

This work studied the problem of identifying influential nodes from highly dynamic
node interaction streams.
Our methods are based on a general TDN model which allows discarding outdated data
in a smooth manner.
We designed three algorithms, i.e., {\sc SieveADN}, {\sc BasicReduction}, and {\sc
HistApprox}.
{\sc SieveADN} identifies influential nodes over ADNs.
{\sc BasicReduction} uses {\sc SieveADN} as a basic building block to identify
influential nodes over TDNs.
{\sc HistApprox} significantly improves the efficiency of {\sc BasicReduction}.
We conducted extensive experiments on real interaction datasets and the results
demonstrated the efficiency and effectiveness of our methods.

\section*{Acknowledgment}

This work is supported by the King Abdullah University of Science and Technology
(KAUST), Saudi Arabia.

\bibliographystyle{IEEEtran}

\clearpage
\section*{Proof of Theorem~\ref{thm:ADN}}

The proof goes along the lines of proof in~\cite{Badanidiyuru2014a} with the major
difference that we allow duplicate nodes in the node stream and the time-varying
nature of objective $f_t$ in our problem.

Let $\OPT_t$ denote the value of an optimal solution at time $t$.
Let $\Delta_t$ denote the maximum value of singletons at time $t$.
The threshold set $\Theta$ in {\sc SieveADN} guarantees that there always exists a
threshold $\theta\in\Theta$ falling into interval $[(1 -
\epsilon)\frac{\OPT_t}{2k}, (1 + \epsilon)\frac{\OPT_t}{2k}]$.
To see this, because $\Delta_t\leq\OPT_t\leq k\Delta_t$, therefore the smallest
threshold in $\Theta$ is at most $(1+\epsilon)\frac{\OPT_t}{2k}$.
Setting $\theta$ to be the largest threshold in $\Theta$ that does not exceed $(1
+ \epsilon)\frac{\OPT_t}{2k}$ will satisfy the claim.

Let $S_\theta^{(t)}$ denote the set of nodes corresponding to threshold $\theta$
maintained in {\sc SieveADN} at time $t$.
We partition set $S_\theta^{(t)}$ into two disjoint subsets $S_\theta^{(t-1)}$ and
$ \bar{S}_\theta^{(t)}$ where $S_\theta^{(t-1)}$ is the maintained nodes at time
$t-1$ and $\bar{S}_\theta^{(t)}\subseteq \bar{V}_t$ is the set of nodes newly
selected from $\bar{V}_t$ at time $t$.
Our goal is to prove that $f_t(S_\theta^{(t)})\geq(1/2-\epsilon)\OPT_t$.

We first inductively show that $f_t(S_\theta^{(t)})\geq |S_\theta^{(t)}|\theta$.
Thus if $|S_\theta^{(t)}|=k$, then $f_t(S_\theta^{(t)})\geq
k\theta\geq\frac{1-\epsilon}{2}\OPT_t$.

At time $t=1$, by definition, $S_\theta^{(1)}=\bar{S}_\theta^{(1)}$ consists of
nodes such that their marginal gains are at least $\theta$.
Therefore $f_1(S_\theta^{(1)})\geq |S_\theta^{(1)}|\theta$.
Assume it holds that $f_{t-1}(S_\theta^{(t-1)})\geq |S_\theta^{(t-1)}|\theta$ at
time $t-1$, then let us consider $f_t(S_\theta^{(t)})$ at time $t$:
\begin{align*}
f_t(S_\theta^{(t)})
&= f_t(S_\theta^{(t-1)}\cup\bar{S}_\theta^{(t)}) \\
&= \underset{\text{first part}}{
\underbrace{f_t(S_\theta^{(t-1)}\cup\bar{S}_\theta^{(t)})-f_t(S_\theta^{(t-1)})}}
+ \underset{\text{second part}}{\underbrace{f_t(S_\theta^{(t-1)})}}
\end{align*}

The first part corresponds to the gain of newly selected nodes
$\bar{S}_\theta^{(t)}$ at time $t$ with respect to the nodes selected previously.
Because the newly selected nodes all have marginal gain at least $\theta$,
therefore, the first part $\geq |\bar{S}_\theta^{(t)}|\theta$.

For the second part, according to the property of ADNs, we have
$f_t(S_\theta^{(t-1)}) \geq f_{t-1}(S_\theta^{(t-1)})$.
Using the induction assumption, it follows that the second part $\geq
|S_\theta^{(t-1)}|\theta$.

We conclude that $f_t(S_\theta^{(t)})\geq (|\bar{S}_\theta^{(t)}| +
|S_\theta^{(t-1)}|)\theta= |S_\theta^{(t)}|\theta$, and $f_t(S_\theta^{(t)})\geq
\frac{1-\epsilon}{2}\OPT_t$ when $|S_\theta^{(t)}|=k$.

In the case $|S_\theta^{(t)}|<k$, let $S_t^*$ denote an optimal set of nodes at
time $t$.
We aim to upper bound the gap between $\OPT_t$ and $f_t(S_\theta^{(t)})$.
Using the submodularity of $f_t$, we have
\begin{align*}
&\OPT_t - f_t(S_\theta^{(t)}) \\
\leq& \sum_{x\in S_t^*\backslash S_\theta^{(t)}}\delta_{S_\theta^{(t)}}(x) \\
=& \sum_{x\in S_{t,1}^*\backslash \bar{S}_\theta^{(t)}}\delta_{S_\theta^{(t)}}(x)
+\sum_{y\in S_{t,2}^*\backslash S_\theta^{(t-1)}}\delta_{S_\theta^{(t)}}(y)
\end{align*}
where $S_{t,1}^*\triangleq S_t^*\cap\bar{V}_t$ is the set of optimal nodes in
$\bar{V}_t$, and $S_{t,2}^*\triangleq S_t^*\backslash S_{t,1}^*$ is the set of
rest optimal nodes.
For the first part, it is obvious that $\delta_{S_\theta^{(t)}}(x) < \theta$.
For the second part, because $S_\theta^{(t-1)}\subseteq S_\theta^{(t)}$, then
$\delta_{S_\theta^{(t)}}(y) \leq \delta_{S_\theta^{(t-1)}}(y)$ due to
submodularity.
By definition,
\begin{align*}
\delta_{S_\theta^{(t-1)}}(y)
&= f_t(\{y\}\cup S_\theta^{(t-1)}) - f_t(S_\theta^{(t-1)}) \\
&\leq f_{t-1}(\{y\}\cup S_\theta^{(t-1)}) - f_{t-1}(S_\theta^{(t-1)}) \\
&\leq \theta.
\end{align*}
The first inequality holds due to the fact that the influence spread of
$y\notin\bar{V}_t$ does not increase at $t$; thus adding edges to the graph at
time $t$ cannot increase $y$'s marginal gain.
For the second inequality, we continue the reasoning until time $t'<t$ when
$y\in\bar{V}_{t'}$; since $y$ is not selected, then its marginal gain is less than
$\theta$.
Therefore, it follows that
\[
\OPT_t - f_t(S_\theta^{(t)})
\leq k\theta
\leq k(1+\epsilon)\frac{\OPT_t}{2k}
=\frac{1+\epsilon}{2}\OPT_t,
\]
which implies $f_t(S_\theta^{(t)})\geq\frac{1-\epsilon}{2}\OPT_t$.

Because {\sc SieveADN} outputs a set of nodes whose value is at least
$f_t(S_\theta^{(t)})$, we thus conclude that {\sc SieveADN} achieves an $(1/2 -
\epsilon)$ approximation factor.
\qed

\section*{Proof of Theorem~\ref{thm:SieveADN-complexity}}

Notice that set $\Theta$ contains $\log_{1+\epsilon}2k=O(\epsilon^{-1}\log k)$
thresholds.
For each node, we need to calculate the marginal gains $|\Theta|$ times and each
uses time $O(\gamma)$.
Thus, for $|\bar{V}_t|=b$ nodes, the time complexity is
$O(b\gamma\epsilon^{-1}\log k)$.
{\sc SieveADN} needs to maintain sets $\{S_\theta\colon \theta\in\Theta\}$ in
memory.
Each set has cardinality at most $k$, and there are $|\Theta|=O(\epsilon^{-1}\log
k)$ sets at any time.
Thus the space complexity is $O(k\epsilon^{-1}\log k)$.
\qed

\section*{Proof of Theorem~\ref{thm:histogram_property}}

If $x_{i+1}'$ became the successor of $x_i'$ due to the removal of indices between
them at some most recent time $t'\leq t$, then procedure \ReduceRedundancy in
Alg.~\ref{alg:histogram} guarantees that $g_{t'}(x_{i+1}')\geq
(1-\epsilon)g_{t'}(x_i')$ after the removal at time $t'$.
From time $t'$ to $t$, it is also impossible to have edge with lifetime between
the two indices arriving.
Otherwise we will meet a contradiction: either these edges form redundant {\sc
SieveADN} instances again thus $t'$ is not the most recent time as claimed, or
these edges form non-redundant {\sc SieveADN} instances thus $x_i$ and $x_{i+1}$
cannot be consecutive at time $t$.
We thus get \textbf{C2}.

Otherwise $x_{i+1}'$ became the successor of $x_i'$ when one of them is inserted
in the histogram at some time $t'\leq t$.
Without lose of generality, let us assume $x_{i+1}'$ is inserted after $x_i'$ at
time $t'$.
If edges with lifetimes between the two indices arrive from time $t'$ to $t$,
these edges must form redundant {\sc SieveADN} instances.
We still get \textbf{C2}.
Or, there is no edge with lifetime between the two indices at all, i.e., $G_t$
contains no edge with lifetime between $x_i$ and $x_{i+1}$.
We get \textbf{C1}.
\qed

\section*{Proof of Theorem~\ref{thm:histogram_guarantee}}

If $x_1=1$ at time $t$, then $\CA_1^{(t)}$ exists.
By Theorem~\ref{thm:basic-reduction}, we have $g_t(x_1)=g_t(1)\geq (1/2 -
\epsilon)\OPT_t$.

Otherwise we have $x_1>1$ at time $t$.
If $G_t$ contains edges with lifetime less than $x_1$, then $\CA_{x_1}^{(t)}$ does
not process all of the edges in $G_t$, thus incurs a loss of solution quality.
Our goal is to bound this loss.

Let $x_0$ denote the most recent expired predecessor of $x_1$ at time $t$, and let
$t'<t$ denote the last time $x_0$'s ancestor $x_0'$ was still alive, i.e.,
$x_0'=1$ at time $t'$ (refer to Fig.~\ref{fig:proof_histapprox}).
For ease of presentation, we commonly refer to $x_0$ and $x_0$'s ancestors as the
left index, and refer to $x_1$ and $x_1$'s ancestors as the right index.
Obviously, in time interval $(t',t]$, no edge with lifetime less than the right
index arrives; otherwise, these edges would create new indices before the right
index; then $x_1$ will not be the first index at time $t$, or $x_0$ is not the
most recent expired predecessor of $x_1$ at time $t$.

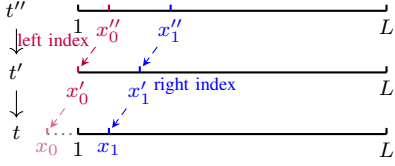
\begin{figure}[htp]
\centering
\begin{tikzpicture}[
scale=.82,
txt/.style={font=\footnotesize, inner sep=3pt, text depth=.25ex, align=center},
arr/.style={densely dashed, ->,>=stealth},
]

\node[txt] (tpp) at (-1,2) {$t''$};
\draw[thick] (0,2) -- (5,2);

\draw[thick] (0,2) -- ++(0,.1);
\draw[thick, purple] (.5,2) -- ++(0,.1);
\draw[thick, blue] (1.5,2) -- ++(0,.1);
\draw[thick] (5,2) -- ++(0,.1);

\node[txt, anchor=north] at (0,2) {$1$};
\node[txt, purple, anchor=north] (x0pp) at (.5,2) {$x_0''$};
\node[txt, blue, anchor=north] (x1pp) at (1.5,2) {$x_1''$};
\node[txt, anchor=north] at (5,2) {$L$};

\node[txt] (tp) at (-1,1) {$t'$};
\draw[thick] (0,1) -- (5,1);

\draw[thick, purple] (0,1) -- ++(0,.1);
\draw[thick, blue] (1,1) -- ++(0,.1);
\draw[thick] (5,1) -- ++(0,.1);

\node[txt, purple, anchor=north] (x0p) at (0,1) {$x_0'$};
\node[txt, blue, anchor=north] (x1p) at (1,1) {$x_1'$};
\node[txt, anchor=north] at (5,1) {$L$};

\node[txt] (t) at (-1,0) {$t$};
\draw[thick] (0,0) -- (5,0);
\draw[thick,dotted, gray] (-.5,0) -- (0,0);

\draw[thick, purple!60] (-.5,0) -- ++(0,.1);
\draw[thick] (0,0) -- ++(0,.1);
\draw[thick, blue] (.5,0) -- ++(0,.1);
\draw[thick] (5,0) -- ++(0,.1);

\node[txt, purple!60, anchor=north] at (-.5,0) {$x_0$};
\node[txt, anchor=north] at (0,0) {$1$};
\node[txt, anchor=north] at (5,0) {$L$};
\node[txt, blue, anchor=north] (x1) at (.5,0) {$x_1$};

\draw[arr, blue] (x1pp) -- (1.05,1.1);
\draw[arr, blue] (x1p) -- (.55,.1);
\draw[arr, purple] (x0pp) -- (.05,1.1);
\draw[arr, purple!60] (x0p) -- (-.45,.1);

\draw[->] (tpp) -- (tp);
\draw[->] (tp) -- (t);

\node[txt,purple,font=\scriptsize] at (-.4,1.5) {left index};
\node[txt,blue,font=\scriptsize] at (1.9,.8) {right index};
\end{tikzpicture}

\caption{Indices relations at time $t''\leq t'<t$.
$x_0$ (resp.~$x_1$) and its ancestors are commonly referred to as the left
(right) index.
}
\label{fig:proof_histapprox}
\end{figure}

Notice that $x_0'$ and $x_1'$ are two consecutive indices at time $t'$.
By Theorem~\ref{thm:histogram_property}, we have two cases.

\bullethdr{If C1 holds.}
In this case, $G_{t'}$ contains no edge with lifetime between $x_0'$ and $x_1'$.
Because there is also no edge with lifetime less than the right index from time
$t'$ to $t$, then $G_t$ has no edge with lifetime less than $x_1$ at time $t$.
Therefore, $\CA_{x_1}^{(t)}$ processed all of the edges in $G_t$.
By Theorem~\ref{thm:basic-reduction}, we still have $g_t(x_1) = g_t(1)\geq (1/2 -
\epsilon)\OPT_t$.

\bullethdr{If C2 holds.}
In this case, there exists time $t''\leq t'$ s.t.
$g_{t''}(x_1'')\geq (1 - \epsilon)g_{t''}(x_0'')$ holds (refer to
Fig.~\ref{fig:proof_histapprox}), and from time $t''$ to $t'$, no edge with
lifetime between the two indices arrived (however $G_t$ may have edges with
lifetime less than $x_1$ at time $t$ and these edges arrived before time $t''$).
Notice that edges with lifetime no larger than the left index all expired after
time $t'$ and they do not affect the solution at time $t$; therefore, we can
safely ignore these edges in our analysis and only care edges with lifetime no
less than the right index arrived in interval $[t'',t]$.
Notice that these edges are only inserted on the right side of the right index.

In other words, at time $t''$, the output values of the two instances satisfy
$g_{t''}(x_1'')\geq (1 - \epsilon)g_{t''}(x_0'')$; from time $t''$ to $t$, the two
instances are fed with same edges.
Such a scenario has been studied in the sliding-window case~\cite{Epasto2017}.
By the submodularity of $f_t$ and monotonicity of the \textsc{SieveADN} algorithm,
the following lemma guarantees that $g_t(x_1)$ is close to $\OPT_t$.

\begin{lemma}[\cite{Epasto2017}]\label{lem:smooth_histogram}
Consider a cardinality constrained monotone submodular function maximization
problem.
Let $\CA(S)$ denote the output value of applying the \textsc{SieveStreaming}
algorithm on stream $S$.
Let $S\Vert S'$ denote the concatenation of two streams $S$ and $S'$.
If $\CA(S_2)\geq(1-\epsilon)\CA(S_1)$ for $S_2\subseteq S_1$ (i.e., each element
in stream $S_2$ is also an element in stream $S_1$), then $\CA(S_2\Vert S)\geq
(1/3-\epsilon)\OPT$ for all $S$, where $\OPT$ is the value of an optimal
solution in stream $S_1\Vert S$.
\end{lemma}

In our scenario, at time $t''$ the two instances satisfy $g_{t''}(x_1'')\geq (1 -
\epsilon)g_{t''}(x_0'')$ and $\CA_{x_1''}$'s input is a subset of $\CA_{x_0''}$'s
input.
After time $t''$, the two instances are fed with same edges.
Because {\sc SieveADN} preserves the property of {\sc SieveStreaming}, hence our
case can be mapped to the scenario in Lemma~\ref{lem:smooth_histogram}.
We thus obtain $g_t(x_1)\geq (1/3 - \epsilon)\OPT_t$.

Combining the above results, we conclude that {\sc HistApprox} guarantees a $(1/3
- \epsilon)$ approximation factor.
\qed

\section*{Proof of Theorem~\ref{thm:histogram_complexity}}

At any time $t$, because $g_t(x_{i+2})<(1-\epsilon)g_t(x_i)$, and $g_t(l)\in
[\Delta, k\Delta]$, then the size of index set $\bx_t$ is upper bounded by
$O(\log_{(1-\epsilon)^{-1}}k)=O(\epsilon^{-1}\log k)$.
For each batch of edges, in the worst case, we need to update $|\bx_t|$ {\sc
SieveADN} instances, and each {\sc SieveADN} instance has update time
$O(b\gamma\epsilon^{-1}\log k)$ according to
Theorem~\ref{thm:SieveADN-complexity}.
In addition, procedure \ReduceRedundancy has time complexity
$O(\epsilon^{-2}\log^2k)$ and it is called at most $b$ times.
Thus the total time complexity is $O(b(\gamma+1)\epsilon^{-2}\log^2k)$.

For memory usage, because {\sc HistApprox} maintains $|\bx_t|$ {\sc SieveADN}
instances, and each instance uses memory $O(k\epsilon^{-1}\log k))$ according to
Theorem~\ref{thm:SieveADN-complexity}.
Thus the total memory used by {\sc HistApprox} is $O(k\epsilon^{-2}\log^2k))$.
\qed

\end{document}